\PassOptionsToPackage{table,xcdraw}{xcolor}
\documentclass{informs4}

\RequirePackage{tgtermes}
\RequirePackage{newtxtext}
\RequirePackage{newtxmath}
\RequirePackage{bm}
\RequirePackage{endnotes}
\OneAndAHalfSpacedXI

\usepackage{algorithm}
\usepackage{algpseudocode}
\usepackage{tikz}
\usepackage{natbib}
\bibpunct[, ]{(}{)}{,}{a}{}{,}
\def\bibfont{\small}

\usepackage{graphicx}
\usepackage[normalem]{ulem}
\usepackage[utf8]{inputenc}
\usepackage[T1]{fontenc}
\usepackage{hyperref}
\usepackage{url}
\usepackage{booktabs}
\usepackage{amsfonts}
\usepackage{nicefrac}
\usepackage{microtype}
\usepackage{amsmath}
\usepackage{amssymb}
\usepackage{multirow}
\usepackage{makecell}
\usepackage[capitalise]{cleveref}
\usepackage{eqparbox}
\usepackage{array}
\usepackage{caption}
\usepackage{threeparttable}
\usepackage{tabularx}
\usepackage{adjustbox}
\usepackage{placeins}
\usepackage{bbm}
\usetikzlibrary{arrows.meta,positioning,shapes.geometric,calc,arrows,decorations,decorations.markings,shadows,matrix,fit}

\EquationsNumberedThrough
\TheoremsNumberedThrough
\ECRepeatTheorems

\makeatletter

\let\theoremstyle\relax

\makeatother

\usepackage{amsthm}

\newsavebox{\lineareqbox}
\newtheorem{mydefinition}{Definition}[section]
\newtheorem{mytheorem}{Theorem}[section]

\newtheorem{mylemma}{Lemma}[section]
\newtheorem{myproposition}{Proposition}[section]
\newtheorem{mycorollary}{Corollary}[section]
\newtheorem{myexample}[myproposition]{Example}

\hypersetup{hidelinks}
\renewcommand{\qedsymbol}{$\blacksquare$}
\newenvironment{myproof}[1][Proof]{\par\noindent\textbf{#1.}\ }{\hfill\qedsymbol\par}
\let\code=\texttt
\let\proglang=\textsf
\newcommand{\pkg}[1]{{\fontseries{b}\selectfont #1}}
\Crefname{assumption}{Assumption}{Assumptions}
\Crefname{remark}{Remark}{Remarks}
\crefname{mylemma}{lemma}{lemmas}
\Crefname{mylemma}{Lemma}{Lemmas}
\crefname{myproposition}{proposition}{propositions}
\Crefname{myproposition}{Proposition}{Propositions}
\crefname{mytheorem}{theorem}{theorems}
\Crefname{mytheorem}{Theorem}{Theorems}
\crefname{mycorollary}{corollary}{corollaries}
\Crefname{mycorollary}{Corollary}{Corollaries}

\renewcommand{\EMAIL}[1]{\href{mailto:#1}{#1}}

\makeatletter
\def\theARTICLETOP{\vspace*{-24pt}}
\renewcommand{\theARTICLEABSTRACT}{%
  \HOOKb
  \vspace*{18pt}
  \begin{minipage}[t]{\textwidth}
    \parindent1em
    \ABSfont
    \noindent\theABSTRACT\endgraf
    \vskip5pt
    \theFUNDING
    \theKEYWORDS
    \theSUBJECTCLASS
    \theAREAOFREVIEW
    \theMSCCLASS
    \theORMSCLASS
    \if@BLINDREV\else\theHISTORY\fi
    \noindent\hrulefill
  \end{minipage}%
  \vspace*{0pt}%
}
\makeatother

\RRHSecondLine{}
\LRHSecondLine{}

\begin{document}

\RUNAUTHOR{Li, Chen, Taylor, and Mao}
\RUNTITLE{A Forecast Combination Framework for Hierarchical and Grouped Time Series Reconciliation}
\TITLE{A Forecast Combination Framework for Hierarchical and Grouped Time Series Reconciliation}

\ARTICLEAUTHORS{%

\AUTHOR{Xixi Li\footnotemark[2]}
\AFF{School of Economics and Management, Beihang University,
\EMAIL{xixili0720@gmail.com}}

\AUTHOR{Zijia Chen\footnotemark[2]}
\AFF{School of Economics and Management, Tsinghua University,
\EMAIL{chen-zj21@mails.tsinghua.edu.cn}}

\AUTHOR{James W. Taylor}
\AFF{Sa{\"i}d Business School, University of Oxford,
\EMAIL{james.taylor@sbs.ox.ac.uk}}

\AUTHOR{Xiaojie Mao\footnotemark[1]}
\AFF{School of Economics and Management, Tsinghua University,
\EMAIL{maoxj@sem.tsinghua.edu.cn}}

\footnotetext[2]{Xixi Li and Zijia Chen contributed equally and share first authorship.}
\footnotetext[1]{Corresponding author.}

}

\ABSTRACT{%
\noindent
Forecast combining and forecast reconciliation for hierarchical and grouped time series have largely developed as separate research areas. This paper connects the two by developing a forecast combination framework for forecast reconciliation. For each bottom-level series, we construct a maximal linearly independent set of structured candidate forecasts from aggregation constraints, and show that combining these candidates and aggregating the resulting bottom-level forecasts is equivalent to standard unbiased linear reconciliation. Within this representation, we prove that mean-squared-error optimal combination weights exactly recover the widely used Minimum Trace (MinT) reconciliation. We further show that the optimal weight problem is separable across different bottom-level series, each yielding a Bates--Granger optimal forecast combination. This reveals MinT as a collection of optimal combinations over hierarchy-induced candidate forecasts. For finite-sample estimation, we propose a modular penalized framework that nests existing MinT variants and supports rich extensions including covariance shrinkage, weight penalization, and scalable series-wise separate estimation. Empirical results show that the framework is practically implementable, competitive with existing methods, and can improve accuracy while preserving coherence. Overall, the forecast combination perspective offers new interpretations of existing reconciliation approaches and provides a flexible basis for designing new methods.
}
\maketitle

\section{Introduction}
\label{Introduction}

Hierarchical and grouped time series are prevalent in many forecasting applications, such as retail sales categorized by product and region~\citep{makridakis2022m5}, or electricity demand data organized geographically across multiple layers, from individual smart meters to local, regional, and national aggregates~\citep{taieb2021hierarchical}. 
At different levels of the hierarchy or grouped structure, forecasts serve distinct purposes: high-quality bottom-level forecasts facilitate efficient operational planning, whereas accurate aggregate forecasts at high levels guide strategic decision-making and resource allocation.
Because these decisions are interconnected, forecasts of hierarchical and grouped series should ideally adhere to aggregation constraints. These constraints require that forecasts for detailed components combine to equal the forecast for their aggregate total, thereby maintaining internal consistency across the forecasting structure. Such coherence is essential for ensuring that operational and strategic decisions are based on a single, aligned view of the future. 

In practice, however, coherent forecasts do not arise automatically. Forecasts are often generated in a decentralized manner, for instance, with regional teams producing local forecasts while a central office creates national projections~\citep{panagiotelis2021forecast}. Even within a centralized forecasting process, different series in the whole system may be modeled separately using different data, methods, objectives, or levels of aggregation. As a result, independently generated base forecasts across different levels typically fail to satisfy the aggregation constraints. This problem, known as forecast incoherence, means that the forecast for an aggregate series may differ from the sum of the forecasts for its components. Such inconsistencies undermine trust in the forecasting process and impair decision-making~\citep{wang2025optimal}, as operational and strategic plans may be based on conflicting information. Therefore, ensuring coherence is not only statistically  desirable but also a critical requirement for organizational effectiveness.

Several approaches have been developed to address this incoherence. Early methods enforce coherence by generating forecasts at a single level of the hierarchy and then implementing adjustments: the bottom-up method first forecasts the most disaggregated series and then aggregates them upward~\citep{dunn1976aggregate}, whereas top-down procedures first forecast aggregate series and then disaggregate them to lower levels~\citep{athanasopoulos2009hierarchical,gross1990disaggregation}. While simple and coherent by construction, these approaches use information from only one level of the hierarchy and ignore potentially useful predictive signals at other levels. A more general response is to treat independently generated forecasts at all levels as base forecasts and then apply an ex post adjustment to enforce aggregation constraints. This perspective motivates more recent forecast reconciliation methods, whose central idea is to adjust the base forecasts by exploiting information from all levels simultaneously and then project them onto the coherent forecast space. A notable example is the Minimum Trace (MinT) approach~\citep{wickramasuriya2019optimal}, which minimizes the trace of the reconciled forecast error covariance matrix under the assumption of unbiased base forecasts. Empirical evidence demonstrates that MinT and its extensions consistently outperform simple bottom-up and top-down approaches~\citep{wickramasuriya2019optimal}. The MinT approach is concerned with point forecasting, which is also our focus in this paper. Although probabilistic forecasts are increasingly important for decision making in hierarchical and grouped time series \citep{bertani2025joint}, point forecasts remain a key component of many probabilistic forecasting methods \citep{bansal2020estimating,gaba2019assessing}. For instance, \citet{taieb2021hierarchical} use MinT point forecasts within a copula-based framework for hierarchical probabilistic forecasting.

While forecast reconciliation provides a principled linear framework for achieving coherence, its cross-level adjustment mechanism can be difficult to interpret in practice~\citep{panagiotelis2021forecast}. 
Consider a retailer forecasting sales for thousands of products across many stores and regions. 
Under reconciliation, the forecast for a specific item, such as bottled water in one store, may be adjusted not only according to its own base forecast, but also according to forecasts for other products, store totals, regional aggregates, and even company-wide sales totals, because all these series are linked through the same aggregation system. 
Statistically, such adjustments can be justified when historical forecast errors exhibit cross-series dependence or when aggregate series contain more stable signals, such as smoother seasonal patterns or trends~\citep{kourentzes2021elucidate}. 
From a business perspective, however, the mechanism is less transparent: it may be unclear why the forecast for bottled water in one store should be affected by forecasts for completely unrelated categories or distant regional aggregates. 
When reconciliation is driven by estimated statistical dependencies rather than explicit domain logic, practitioners may find reconciled forecasts difficult to explain or defend. 
This lack of domain-level interpretability can reduce trust in reconciled forecasts and hinder their adoption in operational decision-making~\citep{panagiotelis2021forecast}.

This interpretability challenge suggests the need for a more transparent description of how information from different parts of the hierarchical or grouped structure contributes to each reconciled forecast. 
Forecast combination provides a potentially useful framework for this purpose. 
In classical forecast combination, multiple candidate forecasts for the same target variable are aggregated into a single prediction, with combination weights indicating the relative contribution of each candidate forecast~\citep{wang2023forecast}. 
These weights offer a direct and traceable summary of how the final forecast is constructed. 
If reconciled forecasts can be represented similarly, the cross-level adjustment mechanism in forecast reconciliation can be interpreted as combining different forecast sources for the same series, rather than as an opaque projection driven only by statistical dependence of arbitrary series. 
This raises a natural question: can forecast reconciliation for hierarchical and grouped time series be formally understood through the lens of forecast combination?

Recent work has begun to explore this perspective. 
In particular, \citet{hollyman2021understanding} propose the level-conditional coherent (LCC) approach, which decomposes a multi-level hierarchy into a sequence of two-level subsystems and reconciles each subsystem separately. 
Under simplifying assumptions, LCC yields reconciled bottom-level forecasts that can be expressed as linear combinations of the base forecast of a series and indirect forecasts constructed from higher-level aggregate and sibling series. 
This representation provides useful intuition by showing that reconciliation can be viewed as combining multiple hierarchy-induced candidate forecasts. 
However, the combination form in LCC arises as an algebraic consequence of subsystem-level reconciliation rather than from an explicit construction and optimization of candidate forecasts. 
As a result, its interpretation depends on restrictive simplifying assumptions, considers only a limited set of hierarchy-induced candidates, and produces weights that are locally induced by separate subsystem problems rather than jointly optimized across all available candidates. 
Thus, although LCC offers an important first step toward a combination-based interpretation of reconciliation, it does not provide a general forecast combination formulation for forecast reconciliation, nor does it characterize the combining weights underlying standard reconciliation methods such as MinT.

Motivated by these limitations, we develop a forecast combination framework for hierarchical and grouped forecast reconciliation. 
Rather than treating forecast combination as an ex post interpretation of an already reconciled forecast, we explicitly construct candidate forecasts for each target series and combine them directly. 
We then show that this combination-based formulation is equivalent to standard linear reconciliation, which includes MinT. 
This framework is useful in two ways: it provides a new conceptual lens for  understanding forecast reconciliation, and it leads to a natural modeling space to design new methods. 
The main contributions are summarized as follows.

\paragraph{A forecast-combination representation of linear reconciliation.}
We establish a novel forecast-combination representation of standard linear reconciliation under a common unbiasedness constraint.
For each bottom-level series, we construct a maximal linearly independent set of structured candidate forecasts, including the direct base forecast of itself and indirect forecasts derived from base forecasts of aggregate-level series.  
Each candidate forecast can be interpreted as an alternative estimate of the same target series implied by a different aggregation constraint. 
We show that combining these candidates and then aggregating the resulting bottom-level forecasts spans exactly the same class of standard unbiased linear reconciliations. 
Thus, the proposed framework preserves the full flexibility of standard linear reconciliation while making each adjustment interpretable as a weighted combination of transparent forecast sources.

\paragraph{A theoretical bridge between MinT and the Bates--Granger forecast combination.}
Within this forecast combining representation, we characterize the mean squared error (MSE)-optimal combination weights and show that  they lead to a   reconciliation matrix that exactly recovers the MinT solution. 
More importantly, we prove that the global optimization problem is intrinsically separable: the off-diagonal blocks of the candidate forecast error covariance matrix do not affect the optimal weights, and the solution decomposes into independent per-series subproblems. 
Each subproblem admits the classical Bates--Granger optimal combination formula \citep{bates1969combination}. This result gives MinT a new interpretation as a collection of optimal forecast combination problems over hierarchy-induced candidate forecasts. 
It also provides a computationally scalable approach to estimation, since the weights can be obtained from solving smaller per-series subproblems in parallel. 
This scalability enables the approach to harness the wisdom of the crowd of candidate forecasts, making it suitable for large hierarchies or grouped structures.

\paragraph{A modular and extensible finite-sample estimation framework.}
Building on the combination representation of linear  reconciliation, 
 we develop a modular penalized optimization framework for finite-sample estimation of  combination weights. 
The framework offers flexibility by explicitly separating three design choices, each of which can be specified independently: the covariance estimation approach used to form the optimization objective, the penalty used to regularize combination weights, and the implementation strategy used to solve the optimization problem. 
This structure recovers standard MinT variants as unpenalized special cases, while naturally accommodating extensions such as factor-based covariance shrinkage, egalitarian weight penalization, and an implementation, which we refer to as \textit{Separate}, that computes weights series by series to reduce problem dimension and enable parallel computation. 
These components provide complementary ways to stabilize finite-sample estimation, especially when the hierarchy or group of time series is large relative to the available training sample. Thus, the combination formulation does more than reproduce MinT: it exposes new modeling levers that are difficult to express directly in the standard reconciliation formulation. We obtain encouraging results for the framework in an electricity generation hierarchical dataset and a labor force grouped time series dataset.

The rest of the paper is organized as follows. Section~\ref{review} introduces the notation, reviews linear reconciliation and MinT, and summarizes the existing combination interpretation based on LCC. Section~\ref{reconcil_mc} develops the proposed combination-based reconciliation representation. Section~\ref{extensions} presents the finite-sample estimation framework. Section~\ref{sec:empirical} reports the empirical evaluation, and Section~\ref{conclusion} concludes.
The Electronic Companion (EC) provides supplementary material, with
proofs of all theoretical results collected in
Section~\ref{sec:ec-proofs}.

\section{Background}
\label{review}
This section introduces the notation for hierarchical and grouped time series, reviews the linear reconciliation framework, and summarizes the existing forecast-combination interpretation.

\subsection{Notation}
\label{notation}

A hierarchical or grouped time series can be viewed as a collection of $n$ related time series that are organized according to a known aggregation structure and are subject to linear constraints. 
Specifically, the aggregation structure consists of $n_b$ bottom-level series and $n_a$ aggregated series, giving a total of $n = n_a + n_b$ time series. Let $\mathbf{b}_t \in \mathbb{R}^{n_b}$ denote the vector of bottom-level observations at time $t$, and let $\mathbf{a}_t \in \mathbb{R}^{n_a}$ collect the aggregated observations. The aggregated observations $\mathbf{a}_t$ satisfy 
\begin{equation*}
    \mathbf{a}_t=A_{\mathrm{agg}} \mathbf{b}_t,
\end{equation*}
where the matrix $A_{\mathrm{agg}} \in \mathbb{R}^{n_a \times n_b}$ maps the bottom-level series to the aggregated series. 
Stacking the aggregated and bottom-level series gives the full observation vector $\mathbf{y}_t=(\mathbf{a}_t^{\prime},\mathbf{b}_t^{\prime})^{\prime}$:
\begin{equation}
    \label{coherence}
\mathbf{y}_t = {S}\mathbf{b}_t, \text{ where }  {S}=\left[\begin{array}{c}
  {A_{\mathrm{agg}}} \\
  {I}_{n_b}
  \end{array}\right] \in \mathbb{R}^{n \times n_b},
\end{equation}
with ${I}_{n_b}$ denoting the identity matrix of dimension $n_b$. The property that the observation $\mathbf{y}_t$ satisfies the linear identity in~(\ref{coherence}) is referred to as \emph{coherence}.

For the aggregation structure described above, let $\mathcal{A}$ and $\mathcal{B}$ denote the sets of aggregated and bottom-level series, respectively. For any aggregated series $j \in \mathcal{A}$, $\mathrm{Desc}_{\mathrm{bot}}(j)$ denotes its bottom-level descendants. 
For a bottom-level series $i\in\mathcal{B}$, define its ancestor set by
$\mathrm{Anc}(i)
=
\{j\in\mathcal{A}:i\in\mathrm{Desc}_{\mathrm{bot}}(j)\}$
and its collateral set by
$\mathrm{Col}(i)
=
\mathcal{A}\setminus\mathrm{Anc}(i)$.
Intuitively, $\mathrm{Anc}(i)$ collects the aggregated series on the path from $i$ to the top level, including the Total series; $\mathrm{Col}(i)$ collects the remaining non-ancestor aggregated series.

\begin{figure}[t]
  \centering
  \includegraphics[width=0.4\textwidth]{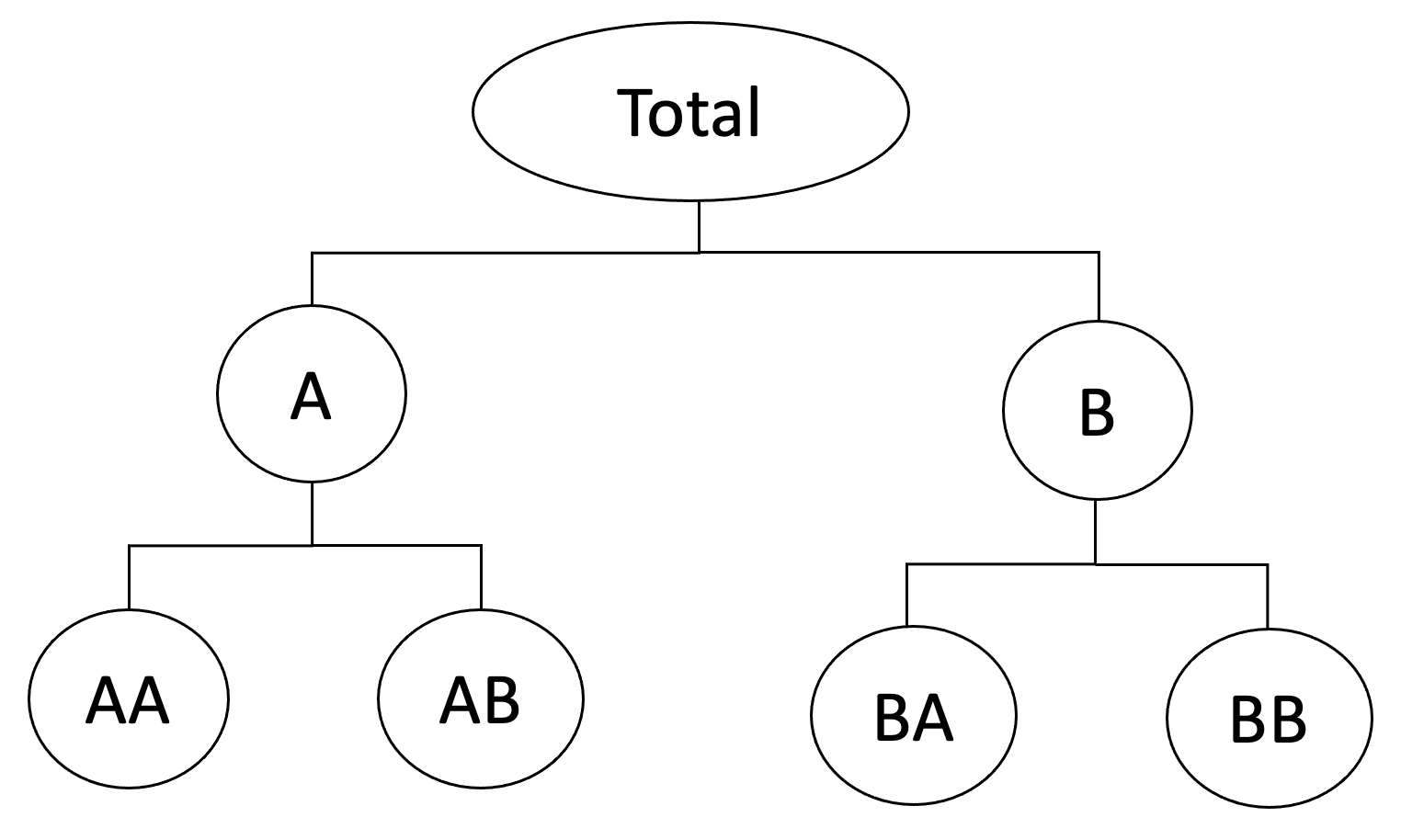}
  \vspace{-10pt}
  \caption{\centering A three-level hierarchical time series.}
  \label{example_1}
  \end{figure}

For ease of exposition, in the remainder of the paper we focus on hierarchical time series; all concepts and results extend directly to grouped time series under the same linear aggregation framework; see Section~\ref{app:grouped} for details.
As an example, Figure~\ref{example_1} shows a three-level hierarchy. As all paths from the top level series to the bottom level series are of the same length, the hierarchy is balanced \citep{di2024forecast}. The mathematical notation and aggregation representation introduced below are not restricted to balanced hierarchies; an unbalanced example is provided in Section~\ref{sec:unbalanced-structure}.

In this example, the full vector of series is
$\mathbf{y}_t=\bigl(y_t^{\mathrm{Total}},\,y_t^{A},\,y_t^{B},
\,y_t^{AA},\,y_t^{AB},\,y_t^{BA},\,y_t^{BB}\bigr)^{\prime}$.
The corresponding aggregated and bottom-level vectors are
$\mathbf{a}_t=\bigl(y_t^{\mathrm{Total}},\,y_t^{A},
\,y_t^{B}\bigr)^{\prime}$
and
$\mathbf{b}_t=\bigl(y_t^{AA},\,y_t^{AB},
\,y_t^{BA},\,y_t^{BB}\bigr)^{\prime}$,
respectively.
The corresponding aggregation (or summing) matrix $S$ that maps $\mathbf{b}_t$ to $\mathbf{y}_t$ is given by
\begin{align*}
S \;=\;
\begin{bmatrix}
1 & 1 & 1 & 1 \\
1 & 1 & 0 & 0 \\
0 & 0 & 1 & 1 \\
& & I_4 &
\end{bmatrix},
\end{align*}
where the first three rows correspond to the aggregation $y_t^{\mathrm{Total}}=y_t^{AA}+y_t^{AB}+y_t^{BA}+y_t^{BB},\, y_t^{A}=y_t^{AA}+y_t^{AB},\, y_t^{B}=y_t^{BA}+y_t^{BB}$, and 
$I_4$ denotes a $4\times 4$ identity matrix for the bottom-level series.

Using the notation defined above, $\mathcal{A}=\{\mathrm{Total},A,B\}$ and $\mathcal{B}=\{AA,AB,BA,BB\}$. For the bottom-level series $AA$, the ancestor and collateral sets are $\mathrm{Anc}(AA)=\{A,\mathrm{Total}\}$ and $\mathrm{Col}(AA)=\{B\}$, respectively. We use this hierarchy as a running example below.

Building on the coherence property of the actual time series, we extend the notion to forecasts.  Let $\hat{\mathbf{y}}_{t+h|t}=(\hat{\mathbf{a}}_{t+h|t}^{\prime},\hat{\mathbf{b}}_{t+h|t}^{\prime})^{\prime}$ denote the $h$-step-ahead forecasts for all series at time $t=1,\dots,T$, where $\hat{\mathbf{a}}_{t+h|t}$ and  $\hat{\mathbf{b}}_{t+h|t}$ denote the aggregated and bottom-level forecasts. The notion of coherence is extended from observations to forecasts as follows.

\begin{mydefinition}[Coherent forecasts]
\label{def-coherence}
The forecasts $\hat{\mathbf{y}}_{t+h|t}=
(\hat{\mathbf{a}}_{t+h|t}^{\prime},\hat{\mathbf{b}}_{t+h|t}^{\prime})^{\prime}$
are coherent if their components satisfy 
$\hat{\mathbf{a}}_{t+h|t} = A_{\mathrm{agg}}\,\hat{\mathbf{b}}_{t+h|t}$, or 
equivalently $\hat{\mathbf{y}}_{t+h|t} = S\,\hat{\mathbf{b}}_{t+h|t}$ for $t= 1,\dots, T$. We refer to these aggregating equations as the aggregation constraints.
\end{mydefinition}

 As we discussed in Section~\ref{Introduction}, the series in a hierarchy are often modeled separately, which typically produces base forecasts that violate the aggregation constraints in Definition~\ref{def-coherence}. Forecast reconciliation addresses this problem by adjusting base forecasts so that they satisfy the aggregation constraints~\citep{hyndman2011optimal, wickramasuriya2019optimal,panagiotelis2021forecast}.

\subsection{Linear Forecast Reconciliation and MinT}
\label{linear-reconc}

The most widely used reconciliation methods are linear.
They can be expressed through a two-step procedure that we formalize below.

\begin{mydefinition}[Standard Linear forecast reconciliation]
\label{def:linear-reconciliation}
A linear forecast reconciliation method first maps base forecasts $\hat{\mathbf{y}}_{t+h\mid t}$ to revised bottom-level forecasts $\tilde{\mathbf{b}}_{t+h\mid t}$ through a reconciliation matrix $P \in \mathbb{R}^{n_b \times n}$ and then aggregates them to all series by the summing matrix $S$:
\[
\tilde{\mathbf{b}}_{t+h\mid t}
    = P \, \hat{\mathbf{y}}_{t+h\mid t},
\qquad
\tilde{\mathbf{y}}_{t+h\mid t}
    = S \, \tilde{\mathbf{b}}_{t+h\mid t}.
\]
 Equivalently, the complete reconciliation map is
\begin{equation}
\label{linear-recon-1}
    \tilde{\mathbf{y}}_{t+h\mid t} = S\,P\,\hat{\mathbf{y}}_{t+h\mid t}.
\end{equation}
The final multiplication by $S$ guarantees that the reconciled forecasts are coherent. 
\end{mydefinition}

The reconciliation matrix $P$ in Definition~\ref{def:linear-reconciliation} is usually restricted by an admissibility condition. In the standard unbiasedness analysis of \citet{wickramasuriya2019optimal}, base forecasts are assumed to be conditionally unbiased and the reconciled forecasts remain unbiased only if
\[
PS=I_{n_b}.
\]
In other words, any matrix $P$ satisfying this constraint does not introduce additional bias through reconciliation.  
This constraint defines the class of admissible  reconciliation matrices considered below.

Within this class, the MinT method of
\citet{wickramasuriya2019optimal} chooses $P$ to minimize the trace of the conditional covariance matrix of the reconciled forecast errors:
\begin{equation} \label{mint-loss} P_{\mathrm{mint}}^\ast \in \operatorname*{arg\,min}_{P:\,PS=I_{n_b}} \mathbb{E}\bigl[ \lVert \mathbf{y}_{t+h} - SP\hat{\mathbf{y}}_{t+h\mid t} \rVert_2^2 \mid \mathcal{I}_t \bigr]. \end{equation}
Here, $\mathcal{I}_t$ denotes the information set available when the $h$-step-ahead base forecasts are produced.
Letting $\hat{\mathbf{e}}_{t+h|t}:=\mathbf{y}_{t+h}-\hat{\mathbf{y}}_{t+h|t}$ be the base forecast error, they further assume that the error covariance matrix $\Sigma_h=\mathrm{Var}(\hat{\mathbf{e}}_{t+h|t}\mid\mathcal{I}_t)$ only depends on $h$ but not $t$. When $\Sigma_h \succ {0}$ (i.e., strictly positive definite), the MinT  reconciliation matrix has the following unique closed-form solution: 
\begin{equation}
\label{mint-solution}
    P_{\mathrm{mint}}^\ast
    = (S^{\prime} \Sigma_h^{-1} S)^{-1} S^{\prime} \Sigma_h^{-1}.
\end{equation}

In practice, $\Sigma_h$ is estimated from in-sample forecast errors. Different covariance specifications yield the standard reconciliation methods summarized in Table~\ref{tab:recon-methods}, reported under the common proportional-covariance simplification $\Sigma_h=\kappa_h\Sigma_1$ with $\kappa_h=1$~\citep{wickramasuriya2019optimal}.
Here, $\Sigma_1$ denotes the covariance matrix of one-step-ahead base forecast errors, commonly estimated by the raw sample covariance
$\hat{\Sigma}_1^{\mathrm{raw}}
=
\frac{1}{T-1}
\sum_{t=1}^{T-1}
(\hat{\mathbf e}_{t+1|t}-\bar{\mathbf e})
(\hat{\mathbf e}_{t+1|t}-\bar{\mathbf e})^{\prime}$, 
where $\hat{\mathbf e}_{t+1|t}=\mathbf y_{t+1}-\hat{\mathbf y}_{t+1|t}$ and
$\bar{\mathbf e}=(T-1)^{-1}\sum_{t=1}^{T-1}\hat{\mathbf e}_{t+1|t}$.

\begin{table}[H]
\small
\refstepcounter{table}
\label{tab:recon-methods}

\noindent
\begin{minipage}{\textwidth}
{\small\renewcommand{\baselinestretch}{1}\selectfont
\textbf{Table~\thetable}\quad
Standard reconciliation methods and covariance estimators.
$\mathbf{1}_{n_b}$ is a vector of ones, $\operatorname{Diag}(\cdot)$ retains the diagonal,
and $\delta \in [0,1]$ is a shrinkage parameter.
\par}
\end{minipage}

\vspace{4pt}

\begin{threeparttable}
\renewcommand{\arraystretch}{1.08}
\begin{tabularx}{\textwidth}{llX}
\hline
Reconciliation method & $\hat{\Sigma}_h$ & Description \\
\hline
OLS \citep{hyndman2011optimal}
& ${I}$
& Data-independent, $P=\left(S^{\prime} S\right)^{-1} S^{\prime}$ \\
\noalign{\vskip 3pt}

WLSs \citep{athanasopoulos2017forecasting}
& $\operatorname{diag}({S} \mathbf{1}_{n_b})$
& Uncorrelated errors with homogeneous variance assumption \\
\noalign{\vskip 3pt}

WLSv \citep{hyndman2016fast}
& $\operatorname{Diag}(\hat{{\Sigma}}_1^{\mathrm{raw}})$
& Variance-based estimator, ignores cross-correlations \\
\noalign{\vskip 3pt}

MinT(Sample) \citep{wickramasuriya2019optimal}
& $\hat{{\Sigma}}_1^{\mathrm{raw}}$
& Full covariance estimator \\
\noalign{\vskip 3pt}

MinT(Shrink) \citep{wickramasuriya2019optimal}
& $\delta \operatorname{Diag}(\hat{{\Sigma}}_1^{\mathrm{raw}})+(1-\delta) \hat{{\Sigma}}_1^{\mathrm{raw}}$
& Shrinkage estimator: off-diagonal elements shrunk toward zero, diagonals preserved \\
\hline
\end{tabularx}
\end{threeparttable}
\end{table}

\subsection{Connection Between Forecast Reconciliation and Forecast Combination}\label{sec: connection}

The preceding discussion treats reconciliation through the matrix $P$. This subsection shifts to the forecast combination perspective, both to review the existing literature linking combination and reconciliation and to motivate the new framework developed in this paper.

Forecast combination refers to the process of aggregating a set of candidate forecasts into a single forecast. The seminal work of \citet{bates1969combination} proposed a linear combination with weights obtained by minimizing the forecast error variance. In finite samples, however, weight estimation can be unstable when there are many candidate forecasts, and the simple average can perform as well as, or even better than, more sophisticated weighted approaches~\citep{wang2023forecast}. This has motivated shrinkage and regularization methods, including shrinkage of weights toward equality~\citep{diebold2019machine, blanc2020bias} and Bayesian estimation of a common forecast error correlation through shrinkage toward a prior~\citep{soule2024heuristic}.

\citet{hollyman2021understanding} provide a preliminary connection between hierarchical forecast reconciliation and forecast combination through a level-conditional coherent (LCC) method. LCC decomposes a multi-level hierarchy into two-level subsystems and reconciles each subsystem separately. Under the simplifying assumption that bottom-level base forecast errors are uncorrelated, each local adjustment can be written as a weighted average of the target series' own base forecast and an indirect forecast formed from an ancestor aggregate after subtracting sibling contributions, all restricted to the corresponding local subsystem. The level-specific reconciled bottom-level forecasts are then averaged and mapped through $S$ to recover coherent forecasts for all series. A concrete example illustrating the LCC method is provided in \Cref{sec:LCC}.
Thus, LCC shows that reconciliation can sometimes be given a combination interpretation.
However, this interpretation is specific to the LCC construction and its simplifying assumptions, and therefore does not establish a general forecast-combination characterization of linear reconciliation.

Three limitations of LCC are particularly important for our purposes.
First, the candidate forecasts implicit in LCC are tied to ancestor-based local subsystems. For a given bottom-level series, the method uses only its own base forecast and indirect forecasts derived from the ancestor aggregates appearing in the corresponding two-level subsystems. This excludes other structurally valid forecast inputs formed from non-ancestor aggregates or from alternative decompositions of the hierarchy. Consequently, LCC works with a restricted candidate set and does not represent the full collection of valid forecast inputs.
Second, conditional on this restricted construction, the implied weights are generated by solving separate local reconciliation problems and then combined through simple averaging. Because each local problem is optimized in isolation, the final averaged weights are not guaranteed to achieve the global optimum of the full combination problem.
Third, LCC obtains its  combination form by assuming that the bottom-level base forecast errors within each two-level subsystem are mutually uncorrelated. This rules out cross-series dependence in the local reconciliation problems. Once such dependence is allowed, the local adjustment no longer reduces to the same simple weighted-average combination, so LCC does not provide a general combination characterization of linear reconciliation or MinT; see also the discussion in~\cite{di2024forecast}.

These limitations motivate a framework that starts from forecast combination itself rather than from local reconciliation steps. 
For each bottom-level series, we explicitly construct structurally valid candidate forecasts generated according to the aggregation constraints, select a maximal linearly independent and interpretable subset, and define reconciliation as a linear combination over this subset of forecasts. 
This formulation allows the combination weights to be studied directly, establishes equivalence with standard linear reconciliation without the local uncorrelated-error assumptions required by LCC, and provides the basis for connecting the resulting optimal weights to the widely-used MinT method. 
We next introduce the proposed framework formally.

\section{Reconciliation via Forecast Combination}
\label{reconcil_mc}
This section recasts forecast reconciliation as a structured forecast combination problem. We define candidate forecasts for each bottom-level series, construct the corresponding candidate-generation matrix, establish equivalence with standard linear reconciliation, and then characterize the population-optimal weights together with their connection to MinT.

\subsection{Candidate Forecasts for a Bottom-Level Series}
\label{indirect-forecast}

Within a hierarchical structure, each bottom-level series can often be expressed in multiple algebraically equivalent ways due to the aggregation constraints. 
For example, in the hierarchy of Figure~\ref{example_1}, the observation $y_{t+h}^{AA}$ for bottom-level series $AA$ can be represented directly as $y_{t+h}^{AA}$, or equivalently as $y_{t+h}^{A} - y_{t+h}^{AB}$, or as $y_{t+h}^{\mathrm{Total}} - y_{t+h}^{B} - y_{t+h}^{AB}$. 
In general, such identities take the form
\begin{equation}
\label{form-1}
    \mathbf{c}_i^{\prime} \mathbf{y}_{t+h} = b_{t+h}^{(i)},  \quad i=1,2,...,n_b,
\end{equation}
where $\mathbf{c}_i \in \mathbb{R}^n$ is a coefficient vector producing the $i$-th bottom-level observation from the full hierarchy. Different algebraic representations correspond to different choices of $\mathbf{c}_i$.

Substituting the hierarchical relation
$\mathbf{y}_{t+h}=S\mathbf{b}_{t+h}$ into~(\ref{form-1})
immediately yields, for each $i=1,2,\ldots,n_b$, the linear constraint
\begin{equation}
\label{linear-equality}
\mathbf{c}_i^{\prime}S\mathbf{b}_{t+h}
=b_{t+h}^{(i)}
=\mathbf{e}_i^\prime\mathbf{b}_{t+h},
 ~\forall\,\mathbf{b}_{t+h}
\iff
S^{\prime}\mathbf{c}_i=\mathbf{e}_i, 
\end{equation}
where $\mathbf{e}_i \in \mathbb{R}^{n_b}$ denotes the $i$-th standard basis vector. Although the representations in~(\ref{form-1}) are algebraically equivalent for the true values, forecast errors break the equalities, yielding multiple candidate forecasts for each bottom-level series. For instance, given base forecasts $\hat{\mathbf{y}}_{t+h|t}$, the three representations of $AA$ mentioned above produce $\hat{y}_{t+h|t}^{A A},\quad \hat{y}_{t+h|t}^{A} - \hat{y}_{t+h|t}^{AB}, $ and $ \hat{y}_{t+h|t}^{\mathrm{Total}} - \hat{y}_{t+h|t}^{B} - \hat{y}_{t+h|t}^{AB}$, which are three candidate forecasts that will typically not be the same. These forecasts are meaningful because they exploit the aggregation constraints to provide direct and indirect predictions for the same bottom-level series.
We formalize the notion of candidate forecasts as follows.
\begin{mydefinition}[Candidate forecasts for a bottom-level series]\label{def:candidate-forecasts}
For the $i$-th bottom-level series, define the admissible coefficient set as $\mathcal{C}_i
:=
\{\mathbf{c}_{i}\in\mathbb{R}^n:S^\prime\mathbf{c}_{i}=\mathbf{e}_i\}$. 
The corresponding set of all $h$-step-ahead candidate forecasts at time $t$ is
\[
\mathcal{F}_{t+h|t}^{\mathrm{all}}(i)
:= \bigl\{\mathbf{c}_{i}^{\prime}\hat{\mathbf{y}}_{t+h|t} 
   : \mathbf{c}_{i} \in \mathcal{C}_i\bigr\}.
\]
We refer to any element of this set as a candidate forecast for the $i$-th series.
\end{mydefinition}

We next characterize the structure of the admissible coefficient set
$\mathcal{C}_i$, which determines the maximum number of linearly independent
candidate forecasts for a given bottom-level series.

\begin{mylemma}[Affine structure of admissible coefficient vectors]
  \label{indep-set}
For any bottom-level series $i$, every admissible coefficient vector
$\mathbf{c}_i\in\mathcal{C}_i$ admits the representation
\begin{equation}
\label{coeffs}
\mathbf{c}_i
=
\mathbf{v}_{0}^{(i)}
+
\sum_{j=1}^{n_a}\alpha_j\mathbf{v}^{(j)},
\end{equation}
for some $(\alpha_1,\ldots,\alpha_{n_a})\in\mathbb{R}^{n_a}$,
where $\mathbf{v}_{0}^{(i)}=(\mathbf{0}_{n_a}^{\prime},\mathbf{e}_i^{\prime})^{\prime}\in \mathbb{R}^n$ is the direct-forecast coefficient vector for
  series $i$ that places $1$ on the entry corresponding to the bottom-level series $i$ and $0$ on all other entries (so that $S^{\prime}\mathbf{v}_{0}^{(i)}=\mathbf{e}_i$), and each $\mathbf{v}^{(j)} \in \mathbb{R}^n$ corresponds to the aggregate series $j \in \mathcal{A}$. Using the ordering of the series in $\mathbf{y}_t$ to index the coordinates, the $m$-th component of $\mathbf{v}^{(j)}$, for $m=1,\ldots,n$, is given by \begin{equation} \label{v_jl_compute} [\mathbf{v}^{(j)}]_m = \begin{cases} 1, & m=j,\\ -1, & m\in\mathrm{Desc}_{\mathrm{bot}}(j),\\ 0, & \text{otherwise}. \end{cases} \end{equation}

  Therefore, $\mathcal{C}_i$ is an affine subspace of dimension $n_a$, whose
  linear span has dimension $n_a+1$. Equivalently, no candidate-generation
  matrix whose rows are admissible coefficient vectors for series $i$ can have
  rank larger than $n_a+1$, so at most $n_a+1$ candidate forecasts in
  $\mathcal{F}_{t+h|t}^{\mathrm{all}}(i)$ are linearly independent.
  \end{mylemma}

The proof of this lemma, deferred to \Cref{app:proof-indep-set}, 
characterizes the solution set of the linear system 
$S^\prime\mathbf{c}=\mathbf{e}_i$ for each bottom-level series $i$.
The affine representation in~\eqref{coeffs} implies that infinitely many
admissible candidate forecasts can be generated by varying the coefficients
$\alpha_j$. However, not all admissible choices have clear hierarchical
interpretations. Fractional choices, for instance, mix multiple aggregation
constraints in arbitrary proportions, so the resulting forecast
$\mathbf{c}_i^{\prime}\hat{\mathbf{y}}_{t+h|t}$ is difficult to interpret. A second problematic
construction involves adding an estimated aggregation residual, which is zero for the true
observations, to the direct forecast of a target series. In the running example,
for bottom-level series $AA$, setting $\alpha_1=\alpha_2=0$ and
$\alpha_3=1$ gives
$\hat{y}^{AA}_{t+h|t}
+\hat{y}^{B}_{t+h|t}
-\hat{y}^{BA}_{t+h|t}
-\hat{y}^{BB}_{t+h|t}$,
where
$\hat{y}^{B}_{t+h|t}-\hat{y}^{BA}_{t+h|t}-\hat{y}^{BB}_{t+h|t}$
estimates the aggregation residual
$y^{B}_{t+h}-y^{BA}_{t+h}-y^{BB}_{t+h}=0$.
This uses incoherence from the collateral branch $B$ rather than an indirect
forecast of $AA$ itself. We therefore focus on a structured and interpretable
subset of admissible candidates, ruling out fractional values of $\alpha_1, \dots, \alpha_{n_a}$ and avoiding reliance solely on incoherence from collateral branches. As shown next, this subset yields a maximal
linearly independent candidate set for each bottom-level series and can be
constructed efficiently from the aggregation structure alone.

\subsection{Constructing a Maximal and Structured Candidate Forecast Set}
\label{constructing-algorithm}

The set $\mathcal{F}_{t+h|t}^{\mathrm{all}}(i)$ in Definition~\ref{def:candidate-forecasts} contains all admissible
candidate forecasts and is generally infinite. We now construct a finite,
structured, and interpretable subset
\(\mathcal{F}_{t+h|t}^{\mathrm{str}}(i)\) of \(n_a+1\) candidate forecasts for each
bottom-level series $i \in \mathcal{B}$. For
each bottom-level series \(i\), recall that \(\mathrm{Anc}(i)\) denotes its aggregate
ancestors, including the total series, and \(\mathrm{Col}(i)\) denotes the
aggregate series that are not ancestors of \(i\).

We define
\[
\mathcal{F}_{t+h|t}^{\mathrm{str}}(i)
=
\mathcal{F}_{t+h|t}^{\mathrm{dir}}(i)
\cup
\mathcal{F}_{t+h|t}^{\mathrm{anc}}(i)
\cup
\mathcal{F}_{t+h|t}^{\mathrm{collateral}}(i).
\]
By construction,
$\mathcal{F}_{t+h|t}^{\mathrm{str}}(i)
\subset
\mathcal{F}_{t+h|t}^{\mathrm{all}}(i)$,
where $\mathcal{F}_{t+h|t}^{\mathrm{dir}}(i)
= 
\{\hat{b}^{(i)}_{t+h|t}\}$ contains the \emph{direct} candidate forecast of the target bottom-level series, generated by the coefficient vector $\mathbf{c}_{i}^{(0)} =  \mathbf{v}_{0}^{(i)}$. For example, for series $AA$ in 
Figure~\ref{example_1}, the base forecast $\hat{y}^{AA}_{t+h|t}$ itself is a direct candidate forecast. 
The sets
\(\mathcal{F}_{t+h|t}^{\mathrm{anc}}(i)\) and
\(\mathcal{F}_{t+h|t}^{\mathrm{collateral}}(i)\), defined below, contain \emph{indirect}
candidate forecasts constructed from ancestor and collateral aggregates,
respectively. Each forecast in these two sets is associated with an admissible
coefficient vector \(\mathbf{c}_{i}^{(\cdot)}\in\mathcal{C}_i\). Stacking the
direct, ancestor-based, and collateral-based coefficient vectors row-wise gives
a matrix \(C^{(i)}\), so that \(C^{(i)}\hat{\mathbf{y}}_{t+h|t}\) contains
exactly the forecasts in \(\mathcal{F}_{t+h|t}^{\mathrm{str}}(i)\).

\paragraph{Ancestor-based candidates.}
An ancestor-based candidate uses an aggregate forecast on the path from the
target series to the series in the top level, and subtracts the forecasts of the other
bottom-level descendants of that aggregate. For example, for \(AA\) in
Figure~\ref{example_1}, the ancestor \(A\) yields $\hat{y}^{A}_{t+h|t}-\hat{y}^{AB}_{t+h|t}$,
while the total series yields
$
\hat{y}^{\mathrm{Total}}_{t+h|t}
-\hat{y}^{AB}_{t+h|t}
-\hat{y}^{BA}_{t+h|t}
-\hat{y}^{BB}_{t+h|t}$.
In general, for each ancestor \(j\in\mathrm{Anc}(i)\), construct 
\small\[
\mathcal{F}_{t+h|t}^{\mathrm{anc}}(i)
=
\Bigl\{
\hat{y}^{j}_{t+h|t}
-
\sum_{m\in\mathrm{Desc}_{\mathrm{bot}}(j)\setminus\{i\}}
\hat{y}^{m}_{t+h|t}
:
j\in\mathrm{Anc}(i)
\Bigl\}.
\]\normalsize
The associated coefficient vectors are 
\(\mathbf{c}_{i}^{(j)}=\mathbf{v}^{(i)}_0+\mathbf{v}^{(j)}\) for $j \in \mathrm{Anc}(i)$.

\paragraph{Collateral-based candidates.}
Ancestor-based candidates use only aggregates on the path from the target series
to the series in the top level. Collateral-based candidates in addition use non-ancestor aggregates by
replacing a group of bottom-level forecasts with the forecast of their common
aggregate. For example, for \(AA\) in Figure~\ref{example_1}, start from the
total-level decomposition
$
\hat{y}^{\mathrm{Total}}_{t+h|t}
-\hat{y}^{AB}_{t+h|t}
-\hat{y}^{BA}_{t+h|t}
-\hat{y}^{BB}_{t+h|t}$. 
The bottom-level terms \(\hat{y}^{BA}_{t+h|t}+\hat{y}^{BB}_{t+h|t}\)
correspond to the collateral aggregate \(B\). Replacing this bottom-level sum
by the non-ancestor aggregate forecast \(\hat{y}^{B}_{t+h|t}\) gives the collateral-based
candidate
$
\hat{y}^{\mathrm{Total}}_{t+h|t}
-
\hat{y}^{B}_{t+h|t}
-
\hat{y}^{AB}_{t+h|t}$, 
which is a structurally valid forecast for \(AA\) using non-ancestor aggregate
information.
In general, for each collateral aggregate \(j\in\mathrm{Col}(i)\), construct the candidate forecasts 
\small\[
\mathcal{F}_{t+h|t}^{\mathrm{collateral}}(i)
=
\Bigl\{
\hat{y}^{\mathrm{Total}}_{t+h|t}
-
\hat{y}^{j}_{t+h|t}
-
\sum_{m\in B_i^j}
\hat{y}^{m}_{t+h|t}
:
j\in\mathrm{Col}(i)
\Bigl\},
\]\normalsize
where $B_i^j
=
\{1,\ldots,n_b\}
\setminus
\bigl(\{i\}\cup\mathrm{Desc}_{\mathrm{bot}}(j)\bigr)$ 
collects the bottom-level series that are neither the target \(i\) nor
bottom-level descendants of \(j\). The associated coefficient vectors are  $\mathbf{c}_{i}^{(j)}
=
\mathbf{v}^{(i)}_0+\mathbf{v}^{(\mathrm{Total})}-\mathbf{v}^{(j)}$ for $j \in \mathrm{Col}(i)$, where $\mathbf{v}^{(\mathrm{Total})}$ is the basis vector in \Cref{indep-set} corresponding to the top total series. 

\begin{algorithm}[!t]
  \caption{Efficient construction of candidate-generation matrix $C$}
  \label{algorithm-1}
  \renewcommand{\algorithmicrequire}{\textbf{Input:}}
  \renewcommand{\algorithmicensure}{\textbf{Output:}}
  \begin{algorithmic}[1]
  \Require Summing matrix $S$
  \Ensure Candidate-generation matrix $C$

  \State Compute $\{\mathbf{v}^{(1)},\ldots,\mathbf{v}^{(n_a)}\}$ using~\eqref{v_jl_compute}.
  \For{$i=1$ {\bfseries to} $n_b$} \Comment{Iterate over each bottom-level series $i$}
      \State Determine $\mathrm{Anc}(i)$, $\mathrm{Col}(i)$, and the direct-forecast coefficient vector  $\mathbf{v}_0^{(i)}$. 
      \State Form $C^{(i)}$ by row-wise stacking  $\mathbf{c}_{i}^{(0)} = \mathbf{v}_0^{(i)}$,
      $\mathbf{c}_{i}^{(j)} = \mathbf{v}_0^{(i)}+\mathbf{v}^{(j)}$ for $j\in\mathrm{Anc}(i)$,
      and $\mathbf{c}_{i}^{(j)} = \mathbf{v}_0^{(i)}+\mathbf{v}^{(\mathrm{Total})}-\mathbf{v}^{(j)}$ for $j\in\mathrm{Col}(i)$.
  \EndFor
  \State Stack $C^{(1)},\ldots,C^{(n_b)}$ row-wise to obtain $C$.
  \State \Return $C$
  \end{algorithmic}
  \end{algorithm}

The construction is summarized in Algorithm~\ref{algorithm-1}. For each
bottom-level series \(i\), the algorithm forms a block
\(C^{(i)}\in\mathbb{R}^{(n_a+1)\times n}\) by stacking the direct coefficient
vector together with the ancestor-based and collateral-based coefficient
vectors. 
The block \(C^{(i)}\) has exactly \(n_a+1\) rows: one direct row and one row
for each aggregate series, represented either through an ancestor-based or a
collateral-based construction. Intuitively, these non-direct rows correspond to
distinct aggregation constraints, thus the construction automatically
attains the maximum number of linearly independent candidate forecasts allowed
by the affine representation in~\eqref{coeffs}. The following proposition
formalizes this property.

\begin{myproposition}[Maximality of the structured candidate set]
  \label{maximality}
  For each bottom-level series \(i \in \mathcal{B}\), the candidate-generation matrix \(C^{(i)}\)
  has full row rank, i.e., \(n_a+1\). Accordingly, every admissible candidate
  forecast in \(\mathcal{F}_{t+h|t}^{\mathrm{all}}(i)\) can be expressed as a
  linear combination of the structured candidate forecasts in
  \(\mathcal{F}_{t+h|t}^{\mathrm{str}}(i)\).
  \end{myproposition}

It follows that multiplying \(C^{(i)}\) by the base forecast vector produces a
maximal linearly independent set of candidate forecasts for series \(i\):
\begin{equation}
\label{eq:candidates}
  \hat{\mathbf{y}}_{\mathrm{cand},t+h\mid t}^{(i)}
  =
  C^{(i)}\hat{\mathbf{y}}_{t+h\mid t}
  \in \mathbb{R}^{n_a+1},
  \qquad i=1,\ldots,n_b.
\end{equation}

For the example in Figure~\ref{example_1}, the structured candidate
forecast vector for series \(AA\) is
\(\hat{\mathbf{y}}_{\mathrm{cand},t+h\mid t}^{(AA)}
=
\bigl(
\hat y^{AA}_{t+h\mid t},
\hat y^{A}_{t+h\mid t}-\hat y^{AB}_{t+h\mid t},
\hat y^{\mathrm{Total}}_{t+h\mid t}
-\hat y^{AB}_{t+h\mid t}
-\hat y^{BA}_{t+h\mid t}
-\hat y^{BB}_{t+h\mid t},
\hat y^{\mathrm{Total}}_{t+h\mid t}
-\hat y^{B}_{t+h\mid t}
-\hat y^{AB}_{t+h\mid t}
\bigr)^{\prime}\), which consists of the direct
forecast for \(AA\), two ancestor-based indirect forecasts generated from
\(A\) and \(\mathrm{Total}\), and one collateral-based indirect forecast
generated from \(B\). Although the hierarchy in Figure~\ref{example_1} is balanced, the
candidate construction is not restricted to balanced structures.
Section~\ref{unbalanced-example} provides a step-by-step
application of Algorithm~\ref{algorithm-1} to an unbalanced hierarchy.

Stacking all candidate forecast vectors across all bottom-level series gives
\small\begin{equation}
\label{C}
  \hat{\mathbf{y}}_{\mathrm{cand},t+h\mid t}
   = 
  \begin{bmatrix}
     \hat{\mathbf{y}}_{\mathrm{cand},t+h\mid t}^{(1)}\\
     \vdots\\
     \hat{\mathbf{y}}_{\mathrm{cand},t+h\mid t}^{(n_b)}
   \end{bmatrix} = C\hat{\mathbf{y}}_{t+h\mid t}
  , \text{ where }   C=\begin{bmatrix}C^{(1)}\\\vdots\\C^{(n_b)}\end{bmatrix}.
\end{equation}\normalsize

In summary, Algorithm~\ref{algorithm-1} provides an efficient way to translate the
aggregation structure into candidate forecasts for every bottom-level series.
For each targeted bottom-level series, it constructs a maximal linearly independent set of
\(n_a+1\) candidates while preserving a clear interpretation: the first
candidate is the direct base forecast, the ancestor-based candidates are
indirect forecasts obtained from aggregates on the target's ancestral path, and
the collateral-based candidates are indirect forecasts that additionally use
non-ancestor aggregate information. This last class is absent
from the LCC construction discussed in Section~\ref{sec: connection}: LCC only uses
ancestor-based local subsystems and therefore omits valid candidates generated
from collateral aggregates. The proposed construction fills this gap and
retains the full structurally available information for the subsequent
combination step.

\subsection{A Combination Representation of Linear Reconciliation}
\label{equivalence}

The construction above provides, for each bottom-level series, several valid
candidate forecasts for the same target: the direct base forecast and a set of
indirect forecasts generated from the aggregation structure. Adopting the
widely-used \cite{bates1969combination} linear forecast combination approach, we combine these
candidate forecasts using series-specific weights. Specifically, for
bottom-level series \(i\), the combined forecast
\(\tilde{b}^{(i)}_{t+h \mid t}\in\mathbb{R}\) is given by 
\begin{equation*}
\tilde{b}^{(i)}_{t+h \mid t}
=
\mathbf{w}_i^\prime
\hat{\mathbf{y}}_{\mathrm{cand},t+h \mid t}^{(i)},
\quad i=1,\ldots,n_b,\ t=1,\ldots,T,
\end{equation*}
where \(\mathbf{w}_i\in\mathbb{R}^{n_a+1}\) is the combination-weight vector, which satisfies \(\mathbf{1}_{n_a+1}^\prime\mathbf{w}_i=1\).

Implementing this combination for all bottom-level series simultaneously, and
recalling that
$\hat{\mathbf{y}}_{\mathrm{cand},t+h\mid t} = C\hat{\mathbf{y}}_{t+h\mid t}$
from~\eqref{C}, the resulting   forecast vector
$\tilde{\mathbf{b}}_{t+h\mid t} =
(\tilde{b}^{(1)}_{t+h\mid t},\ldots,\tilde{b}^{(n_b)}_{t+h\mid t})^{\prime}$ is
\begin{equation*}
\label{eq:bottom-stack}
\tilde{\mathbf{b}}_{t+h \mid t} = \Phi(\mathbf{w}) C \hat{\mathbf{y}}_{t+h \mid t},
\quad t=1,\ldots,T,
\end{equation*}
where $\Phi(\mathbf{w}) \in \mathbb{R}^{n_b \times n_b(n_a+1)}$ is a
 matrix collecting all weight vectors:
$\mathbf{w} = (\mathbf{w}_1^{\prime}, \ldots, \mathbf{w}_{n_b}^{\prime})^{\prime}$:
\begin{equation}
\label{eq:phi}
\begin{aligned}
&\Phi(\mathbf{w})
=
\begin{bmatrix}
\mathbf{w}_1^{\prime} & \mathbf{0}_{n_a+1}^{\prime} & \cdots & \mathbf{0}_{n_a+1}^{\prime} \\
\mathbf{0}_{n_a+1}^{\prime} & \mathbf{w}_2^{\prime} & \cdots & \mathbf{0}_{n_a+1}^{\prime} \\
\vdots & \vdots & \ddots & \vdots \\
\mathbf{0}_{n_a+1}^{\prime} & \mathbf{0}_{n_a+1}^{\prime} & \cdots & \mathbf{w}_{n_b}^{\prime}
\end{bmatrix}, \\
&\text{where } \mathbf{1}_{n_a+1}^{\prime}\mathbf{w}_i = 1,
\qquad i=1,\ldots,n_b.
\end{aligned}
\end{equation}
where $\mathbf{0}_{n_a+1}^\prime$ denotes the transpose of a zero vector 
of length $n_a+1$.
The block-diagonal structure reflects the fact that each bottom-level series
draws exclusively on its own candidate forecasts.

Finally, applying the summing matrix $S$ recovers reconciled forecasts at all
levels of the hierarchy:
\begin{equation}
\label{eq:mc-recon}
\tilde{\mathbf{y}}_{t+h \mid t} 
= S \tilde{\mathbf{b}}_{t+h \mid t}
= S \underbrace{\Phi(\mathbf{w}) C}_{=:\, P} \hat{\mathbf{y}}_{t+h \mid t},
\quad t=1,\ldots,T.
\end{equation}
Thus, the combination-based reconciliation is a linear reconciliation in the
sense of Definition~\ref{def:linear-reconciliation}, with reconciliation matrix
\(P=\Phi(\mathbf{w})C\). This factorization has a simple interpretation:
\(C\) first maps the original base forecasts into candidate forecasts for each
bottom-level target series, and \(\Phi(\mathbf{w})\) then combines the
candidates for each target using the corresponding series-specific weights.

Although this combination-based reconciliation matrix \(P=\Phi(\mathbf{w})C\) appears to impose additional structure,
Theorem~\ref{Equivalence} shows that it entails no loss of generality: every
unbiased linear reconciliation matrix can be represented in this form, and vice versa.

\begin{mytheorem}[Equivalence with standard linear reconciliation in Definition~\ref{def:linear-reconciliation}]
  \label{Equivalence}
  Let \(C\) be the candidate-generation matrix constructed by
  Algorithm~\ref{algorithm-1}. A matrix \(P\in\mathbb{R}^{n_b\times n}\) satisfies
  the unbiasedness condition \(PS=I_{n_b}\) if and only if there exists a  
  block-diagonal combination matrix \(\Phi(\mathbf{w})\) of the form~\eqref{eq:phi}
  such that $P=\Phi(\mathbf{w})C$. 
  Thus, the proposed forecast-combination reconciliation in~\eqref{eq:mc-recon}
  and the standard unbiased linear reconciliation in~\eqref{linear-recon-1}
  represent the same class of reconciled forecasts.
  \end{mytheorem}

Theorem~\ref{Equivalence} shows that the proposed candidate-forecast
combination framework is not merely an interpretation of a particular new
reconciliation method, but an equivalent representation of standard unbiased linear reconciliation. In contrast, the LCC interpretation in Section~\ref{sec: connection}
is limited to ancestor-based local subsystems and therefore omits valid
candidate forecasts generated from non-ancestor aggregates. By constructing a
maximal candidate set, our framework makes it possible to study reconciliation
directly through combination weights without loss of any information. We next
characterize the optimal combination weights within this same representation.

\subsection{Optimal Combination Weights and MinT}
\label{new-approach}

We now characterize the optimal combination weights based on the proposed
candidate-forecast representation. 
We derive the closed-form optimal weights, analyze their structural properties, and establish the equivalence between the induced reconciliation matrix and MinT.

Recall that
$\mathbf{w}=(\mathbf{w}_1^\prime,\ldots,\mathbf{w}_{n_b}^\prime)^\prime$
stacks the series-specific combination-weight vectors, where
$\mathbf{w}_i\in\mathbb{R}^{n_a+1}$ contains the weights applied to the
candidate forecasts for bottom-level series $i$. With this notation, the
optimal combination weights are defined as the solution to the conditional
MSE minimization problem:

\begin{equation}
\label{obj1}
\begin{aligned}
\min_{\mathbf w}\quad
&\mathbb E\![
\|
\mathbf y_{t+h}
-
S\Phi(\mathbf w)C\hat{\mathbf y}_{t+h|t}\|_2^2
\mid
\mathcal I_t] \\
\text{s.t.}\quad
&
\mathbf1_{n_a+1}'\mathbf w_i=1,
\qquad
i=1,\ldots,n_b.
\end{aligned}
\end{equation}

Equivalently, the sum-to-one constraints can be written compactly as follows: 
\begin{equation}
\label{A}
\Pi\mathbf{w}=\mathbf{1}_{n_b}, \text{ where }
\Pi = \left[\begin{array}{cccc}
\mathbf{1}_{n_a+1}^{\prime} & \mathbf{0}_{n_a+1}^{\prime} & \cdots & \mathbf{0}_{n_a+1}^{\prime} \\
\mathbf{0}_{n_a+1}^{\prime} & \mathbf{1}_{n_a+1}^{\prime} & \cdots & \mathbf{0}_{n_a+1}^{\prime} \\
\vdots & \vdots & \ddots & \vdots \\
\mathbf{0}_{n_a+1}^{\prime} & \mathbf{0}_{n_a+1}^{\prime} & \cdots & \mathbf{1}_{n_a+1}^{\prime}
\end{array}\right].
\end{equation}

 We now further characterize the optimal weights that solve the problem in \eqref{obj1}. To this end, we follow MinT to assume that the base forecast errors are conditionally unbiased, $\mathbb{E}[\hat{\mathbf{e}}_{t+h|t}\mid\mathcal{I}_t]=\mathbf{0}$, and that $\Sigma_h=\mathrm{Var}(\hat{\mathbf{e}}_{t+h|t}\mid\mathcal{I}_t)$ for the base forecast error $\hat{\mathbf{e}}_{t+h|t}=\mathbf{y}_{t+h}-\hat{\mathbf{y}}_{t+h|t}$ is invariant across $t$. Then \eqref{obj1} has a useful convex quadratic
 program reformulation as follows.

\begin{mytheorem}
  \label{quadratic-theorem}
A vector $\mathbf{w}^\ast$ solves~\eqref{obj1} if and only if it solves the convex quadratic program
\begin{equation}
\label{quadratic-formulation}
\begin{aligned}
    & \min_{\mathbf{w}} \ \mathbf{w}^{\prime} Q_h \mathbf{w},
    \quad \mathrm{s.t.} \quad \Pi\mathbf{w} = \mathbf{1}_{n_b},
\end{aligned}
\end{equation}
where $Q_h \in \mathbb{R}^{n_b(n_a+1) \times n_b(n_a+1)}$ is the matrix $Q_h
=
\bigl\{(S^\prime S)\otimes \mathbf{1}_{n_a+1}\mathbf{1}_{n_a+1}^\prime\bigr\}
\circ
(C\Sigma_h C^\prime)$.
Here $\otimes$ denotes the Kronecker product and $\circ$ denotes the
Hadamard product. Equivalently, for $i,i^\prime = 1,\ldots,n_b$, the $(i,i')$-th
block of $Q_h$ is
$[Q_h]_{i,i'}=z_{i,i'}C^{(i)}\Sigma_h(C^{(i')})^\prime \in \mathbb{R}^{(n_a+1) \times (n_a+1)}$, with
$z_{i,i'}$ as the $(i,i')$-entry of $S^\prime S$.
\end{mytheorem}

In fact, when the base forecast error covariance matrix is strictly positive definite, i.e., $\Sigma_h \succ {0}$, it can be proved that the matrices $Q_h$ and $\Pi Q_h^{-1} \Pi^\prime$ are also strictly positive 
definite. As a result, the quadratic program in \eqref{quadratic-formulation} is strictly convex, admitting a unique closed-form solution.

\begin{myproposition}[Closed-form solution]
  \label{closed_solution}
  Under the candidate-generation construction in Algorithm~\ref{algorithm-1}, if the base forecast
  error covariance matrix satisfies $\Sigma_h \succ 0$, then the quadratic
  program in~\eqref{quadratic-formulation} has a unique  closed-form solution: the optimal
  combination weights are
  \begin{equation}
  \label{solution-w}
  \mathbf{w}^\ast
  =
  Q_h^{-1}\Pi^{\prime}
  (\Pi Q_h^{-1}\Pi^{\prime})^{-1}
  \mathbf{1}_{n_b}.
  \end{equation}
  \end{myproposition}

  Although Proposition~\ref{closed_solution} gives a closed-form expression,
  the formula in~\eqref{solution-w} is difficult to use directly in large
  hierarchies. It requires inverting the full matrix $Q_h$, whose dimension is
  $n_b(n_a+1)\times n_b(n_a+1)$, and also the matrix
  $\Pi Q_h^{-1}{\Pi}^\prime$. Since the off-diagonal blocks of $Q_h$ encode
  cross-series dependence among candidate forecast errors, estimating and
  inverting the full matrix can be costly and unstable when the hierarchy is
  large.

  Surprisingly, we discover that these off-diagonal blocks do not affect the optimal combination
  weights. To state this result, define the block-diagonal counterpart of $Q_h$ as $  Q_h^{\mathrm{bd}}
  =
  \operatorname{blkdiag}\!\left(
  [Q_h]_{1,1},\ldots,[Q_h]_{n_b,n_b}
  \right)$,
  which retains only the within-series candidate-error covariance blocks and sets
  all cross-series blocks to zero. 
  Consider the corresponding block-diagonal quadratic program
\begin{equation}
\label{quadratic-formulation-bd}
\begin{aligned}
    & \min_{\mathbf{w}} \ \mathbf{w}^{\prime} Q_h^{\mathrm{bd}} \mathbf{w},
    \quad \mathrm{s.t.} \quad \Pi \mathbf{w} = \mathbf{1}_{n_b}.
\end{aligned}
\end{equation}

Because $Q_h^{\mathrm{bd}}$ is block diagonal and the constraints
$\Pi \mathbf{w}=\mathbf{1}_{n_b}$ impose only the series-wise restrictions
$\mathbf{1}_{n_a+1}^{\prime}\mathbf{w}_i=1$,
problem~\eqref{quadratic-formulation-bd} decouples into $n_b$ independent
subproblems. Specifically, for each $i=1,\ldots,n_b$,
\begin{equation}
\label{eq:subproblem}
\min_{\mathbf{w}_i}\quad
\mathbf{w}_i^{\prime}[Q_h]_{i,i}\mathbf{w}_i
\qquad
\text{s.t.}\quad
\mathbf{1}_{n_a+1}^{\prime}\mathbf{w}_i=1.
\end{equation}
Let
$\mathbf{w}^{\mathrm{sep}}
=
(\mathbf{w}_1^{\mathrm{sep}\prime},\ldots,
\mathbf{w}_{n_b}^{\mathrm{sep}\prime})^\prime$
denote the solution obtained by solving these independent subproblems. 

\begin{myproposition}[Separability of optimal weights]
  \label{prop:equivalence}
  Under the conditions in \Cref{closed_solution}, for each bottom-level series $i$, the solution to the $i$-th subproblem of
  \eqref{quadratic-formulation-bd} is
  \begin{equation}
  \label{eq:BG}
  \mathbf{w}_i^{\mathrm{sep}} =
  \frac{(C^{(i)}\Sigma_h C^{(i)\prime})^{-1}\mathbf{1}_{n_a+1}}
  {\mathbf{1}_{n_a+1}^\prime
  (C^{(i)}\Sigma_h C^{(i)\prime})^{-1}\mathbf{1}_{n_a+1}},
  \quad i=1,\ldots,n_b.
  \end{equation}
  Moreover, the stacked solution
  $\mathbf{w}^{\mathrm{sep}}$ coincides with the solution to the full
  quadratic program~\eqref{quadratic-formulation}:
  \[
  \mathbf{w}^\ast=\mathbf{w}^{\mathrm{sep}}.
  \]
  \end{myproposition}

  Thus the block-diagonal program in \eqref{quadratic-formulation-bd} is not an approximation to the full
problem: it gives exactly the same optimal combination weights. The practical
benefit is substantial. Instead of solving one problem of dimension
$n_b(n_a+1)$, the weights can be obtained from $n_b$ independent problems of
dimension $n_a+1$, each involving a much smaller covariance matrix
$C^{(i)}\Sigma_h C^{(i)\prime}$ of the candidate forecast errors for series
$i$. Moreover, formula~\eqref{eq:BG} is precisely the Bates--Granger optimal combination
rule applied separately to the candidate forecasts of each bottom-level series.

We next connect the reconciliation induced by these optimal combination weights
to MinT and establish their equivalence.
\begin{myproposition}[MinT as a combination procedure]
\label{equ-mint}
The optimal reconciliation matrices induced by $\mathbf{w}^\ast$ 
in~\eqref{solution-w} and by $\mathbf{w}^{\mathrm{sep}}$ 
in~\eqref{eq:BG} coincide with the optimal MinT reconciliation matrix in~\eqref{mint-solution}:
\[
\Phi(\mathbf{w}^\ast)C = \Phi(\mathbf{w}^{\mathrm{sep}})C = P_{\mathrm{mint}}^\ast.
\]
\end{myproposition}

This equivalence gives MinT a forecast-combination interpretation. In
particular, MinT can be viewed as first constructing, for each bottom-level
series, a set of structurally valid candidate forecasts and then applying the
optimal combination rule in~\eqref{eq:BG}. To our knowledge, this provides the
first explicit link between MinT reconciliation and the Bates--Granger
forecast-combination perspective.

\subsection{Summary of the New Reconciliation Framework}

This section develops the proposed reconciliation framework based on forecast
combination. For each bottom-level series, we construct a maximal set of structurally valid
candidate forecasts and collect them through the candidate-generation matrix
$C$. Reconciliation is then performed by choosing series-specific combination
weights, encoded in $\Phi(\mathbf{w})$, to combine these candidates at the
bottom level. The resulting bottom-level combined forecasts are aggregated by
$S$, yielding coherent forecasts
$S\Phi(\mathbf{w})C\hat{\mathbf{y}}_{t+h|t}$ for the full hierarchy. The equivalence theorem shows that this representation spans the same class as
standard unbiased linear reconciliation.

The subsequent results characterize the mean squared error (MSE)-optimal combination weights within
this framework. Although the full optimization program involves joint optimization over high dimensional weights, its optimal solution is separable: the weights can be obtained
from independent per-series problems, each taking the Bates--Granger form
in~\eqref{eq:BG}. The induced reconciliation matrix is exactly the MinT matrix. This 
 shows that the well-known MinT approach can be reinterpreted as an optimal forecast-combination
procedure over the candidate forecasts generated by the hierarchy.

These results establish the theoretical equivalence between the proposed
combination framework and standard linear reconciliation. In finite samples,
however, the covariance matrices entering the weight optimization problems must be estimated, potentially in
high-dimensional settings where the number of series is large relative to the
available forecast error history. The combination perspective provides a useful
way to address this problem, because it exposes the covariance structure and the
combination weights as objects that can be stabilized or regularized directly.
Moreover, the solution separability property motivates a series-wise separate estimation procedure, which  not only improves computation speed but also offers additional stability. 
The next section develops finite-sample estimation methods based on our combination perspective.

\section{Finite-Sample Estimation with  Regularization}
\label{extensions}
\label{framework}

The previous section developed the forecast-combination reconciliation
framework and characterized the theoretically optimal combination weights.
We now turn to the important practical issue of finite-sample estimation: given a finite training sample,
how should these weights be estimated and used to produce reconciled
forecasts? Section~\ref{reg_two_components} describes how we address this through the use of regularization to stabilize the estimation of forecast error covariances and combining weights. Section~\ref{shrink_estimator} discusses covariance-side regularization through shrinkage, Section~\ref{e-penalty} introduces weight-side regularization through  penalization, and Section~\ref{joint-sep-implementation} describes different options for implementing the two together.

\subsection{Regularized Estimation of Two Key Components}
\label{reg_two_components}

A central challenge with the forecast-combination reconciliation framework is that the optimal weights depend on the forecast-error
covariance structure, which is unknown and must be estimated from data. When
the hierarchy is large relative to the available training sample, the covariance estimates and the resulting
weight estimates can be unstable. Our implementation addresses this finite-sample
problem from two complementary directions. First, we stabilize the estimated
forecast error covariance matrix through shrinkage estimation.
Second, we regularize the combination weights directly, penalizing them toward
stable benchmarks. Together, these two components yield a practical
penalized reconciliation procedure that preserves the interpretation of the
combination framework while improving finite-sample stability.

To see where these two components enter, recall that the theoretical optimal
weights solve the quadratic program in~\eqref{quadratic-formulation}, whose
objective is governed by the matrix $Q_h$. This matrix summarizes the
candidate forecast error covariance after accounting for the
aggregation geometry of the hierarchy; equivalently, its blocks are
$[Q_h]_{i,i'}=z_{i,i'}C^{(i)}\Sigma_h C^{(i')\prime}$, where $z_{i,i'}$ is the
corresponding entry of $S^\prime S$, $\Sigma_h$ is the covariance matrix of the $h$-step base forecast errors, and $C^{(i)}$ is the candidate generation matrix given by Algorithm~\ref{algorithm-1}. In finite samples, $Q_h$ is replaced by
an empirical estimator $\hat Q_h$, and in Section~\ref{shrink_estimator}, we explain how we use shrinkage to improve the estimation.

Weight-side stabilization can be represented by adding a generic penalty
$\lambda \mathcal P(\mathbf w)$ to the empirical quadratic objective. 
The penalty term $\mathcal P(\mathbf w)$ is used to encourage desirable
structure in the estimated combination weights, such as shrinkage toward
stable benchmark weights, sparsity among candidate forecasts, or other
interpretable weighting patterns. 
The tuning parameter $\lambda\ge 0$ controls the strength of this
regularization: $\lambda=0$ recovers the unpenalized estimator, whereas
larger values place more emphasis on the structure encouraged by
$\mathcal P(\mathbf w)$. 
A general stabilized weight estimator can therefore be written
schematically as
\begin{equation}
  \label{empirical-general-penalized}
  \hat{\mathbf{w}}  \in \arg\min_{\mathbf{w}}~
\mathbf{w}'\hat Q_h\mathbf{w}
+ \lambda\mathcal P(\mathbf{w}),
\qquad
\mathrm{s.t.}\ \Pi \mathbf{w}=\mathbf{1}_{n_b}.
\end{equation}
For practical implementation, the equality constraints in
\eqref{empirical-general-penalized} can be eliminated exactly through
an affine reparameterization. The resulting unconstrained formulation
is provided in Section~\ref{sec:unconstrained}.

The quadratic term in~\eqref{empirical-general-penalized} is controlled by how $\hat Q_h$ is estimated; the penalty
term controls how additional structure is imposed on the combination weights. After discussing the estimation of $\hat Q_h$  in Section~\ref{shrink_estimator}, we consider specifications for the penalty term $\mathcal P(\mathbf w)$ in Section~\ref{e-penalty}.

\subsection{Covariance Shrinkage}
\label{shrink_estimator}

The population matrix $Q_h$ is determined by the base forecast error covariance matrix $\Sigma_h$, together with the candidate-generation matrix $C$ and the aggregation structure. 
Thus, in finite samples, the central estimation task is to estimate $\Sigma_h$. 
This is the same problem encountered in MinT: as reviewed in Section~\ref{review} and Table~\ref{tab:recon-methods}, standard implementations usually adopt the simplification $\Sigma_h=\kappa_h\Sigma_1$ with $\kappa_h=1$, so that different estimators of $\Sigma_1$ give rise to OLS, WLSs, WLSv, MinT(Sample), and MinT(Shrink).
Our framework can use these same covariance estimators directly: given any plug-in estimator $\hat{\Sigma}_h$, we form $\hat Q_h
=
\bigl\{(S^\prime S)\otimes \mathbf{1}_{n_a+1}\mathbf{1}_{n_a+1}^\prime\bigr\}
\circ
(C\hat \Sigma_h C^\prime)$, 
that is, by replacing $\Sigma_h$ with $\hat{\Sigma}_h$ in the population expression for $Q_h$. 
When no weight penalty is imposed, this plug-in construction recovers the corresponding MinT variant.

\begin{mycorollary}[Equivalence under plug-in covariance estimation]
  \label{equ-mint-empirical}
  For any symmetric positive definite estimator $\hat{\Sigma}_h$ of $\Sigma_h$,
  construct $\hat{Q}_h$ by plugging $\hat{\Sigma}_h$ into $Q_h$. 
  Then the unpenalized combination-based reconciliation matrix (i.e., $\lambda = 0$) coincides with
  the corresponding plug-in MinT matrix: $  \Phi(\hat{\mathbf w}^{\ast})C
  =
  \hat{P}_{\mathrm{mint}}^\ast$,
  where $\hat{P}_{\mathrm{mint}}^\ast$ is obtained from~\eqref{mint-solution}
  with $\Sigma_h$ replaced by $\hat{\Sigma}_h$.
  \end{mycorollary}

\begin{proof}
Proposition~\ref{equ-mint} holds for every positive-definite forecast
error covariance matrix. Therefore, replacing $\Sigma_h$ by any
symmetric positive-definite estimator $\hat{\Sigma}_h$ yields 
$
\Phi(\hat{\mathbf w}^{\ast})C
=
\hat P_{\mathrm{mint}}^\ast.
$
\end{proof}

Thus, Corollary~\ref{equ-mint-empirical} shows that the covariance estimator is a modular input in the proposed framework. 
If $\hat{\Sigma}_h$ is chosen according to the standard covariance specifications reviewed in Table~\ref{tab:recon-methods}, the unpenalized combination-based estimator recovers the corresponding standard reconciliation methods, including OLS, WLSs, WLSv, MinT(Sample), and MinT(Shrink). 
In this sense, the proposed framework retains the full covariance-estimation flexibility of MinT.

At the same time, the framework is not restricted to covariance estimators previously used in the MinT literature. 
Because any positive definite estimator $\hat{\Sigma}_h$ can be plugged into the construction of $\hat{Q}_h$, one can introduce alternative covariance specifications tailored to the empirical setting. 
In what follows, we use the common one-step proxy by setting $\hat{\Sigma}_h=\hat{\Sigma}_1$ in the empirical construction, and borrow a factor-based shrinkage estimator $\hat{\Sigma}_1^{\mathrm{shrink}}$ for $\Sigma_1$ from \cite{ledoit2003improved}.
The idea parallels MinT(Shrink): instead of relying directly on the raw sample covariance, we shrink it toward a structured target. 
Unlike MinT(Shrink), which uses a diagonal matrix target, our estimator uses a factor-based covariance target that captures common forecast-error dependence parsimoniously. 
The resulting shrinkage estimator combines information from the raw sample covariance with the stability of the factor structure, improving performance in high-dimensional settings~\citep{ledoit2003improved,ledoit2004wellconditioned}. 
Factor-based estimators have also proven useful in forecast combination~\citep{chen2025combining} and hierarchical forecasting~\citep{hollyman2024scaleable}. 
In hierarchical forecasting, however, they have primarily been used within Bayesian dynamic models for covariance modeling of the underlying series. Here, by contrast, we use a factor-based estimator as a stabilized plug-in estimate of the base forecast error covariance. To our knowledge, such a factor-based covariance estimator has not previously been used for MinT-style reconciliation.

We now describe the factor-based covariance estimator in \cite{ledoit2003improved}. 
For each series $m=1,\ldots,n$ and each forecast origin
$t=1,\ldots,T-1$, let $\hat{e}_{t+1\mid t}^{(m)}$ denote the
one-step-ahead forecast error of the $m$-th series. We model these errors
using the single-factor representation
\begin{equation*}
\hat{e}_{t+1\mid t}^{(m)}
=
\gamma_m+\beta_m f_{t+1}+\epsilon_{t+1}^{(m)},
\end{equation*}
where $\gamma_m$ and $\beta_m$ are series-specific parameters,
$f_{t+1}$ denotes a common factor at $t+1$ shared by all series,
and $\epsilon_{t+1}^{(m)}$ is an idiosyncratic error term satisfying
$\mathbb{E}[\epsilon_{t+1}^{(m)} \mid f_{t+1}]=0$ and
$\mathbb{E}[\epsilon_{t+1}^{(m)}\epsilon_{t+1}^{(m^{\prime})}]=0$
for $m\neq m^{\prime}$.
We further assume $\operatorname{Var}(f_{t+1}) = \tau^2$ and
$\operatorname{Var}(\epsilon_{t+1}^{(m)}) = \sigma_m^2$,
which leads to the following covariance structure for the one-step-ahead
base forecast errors:
\[
\Sigma_1^{\mathrm{factor}}
= \tau^2\,\boldsymbol{\beta}\boldsymbol{\beta}^{\prime} + D,
\]
where $\boldsymbol{\beta}=(\beta_1,\dots,\beta_n)^{\prime}$ and
$D=\operatorname{diag}(\sigma_1^2,\ldots,\sigma_n^2)$.
In practice, the common factor is estimated by the cross-sectional mean
$\hat{f}_{t+1}=n^{-1}\sum_{m=1}^n \hat{e}_{t+1|t}^{(m)}$, and the loadings and
idiosyncratic variances are obtained from series-level regressions on this
estimated factor. We denote the resulting covariance estimator as $\hat{\Sigma}_1^{\mathrm{factor}}$. 

Finally, the shrinkage estimator of $\Sigma_1$ is
$
\hat{\Sigma}_1^{\mathrm{shrink}}
=
\hat{\delta}\hat{\Sigma}_1^{\mathrm{factor}}
+
(1-\hat{\delta})\hat{\Sigma}_1^{\mathrm{raw}},$
where $\hat{\delta}$ is a data-dependent plug-in shrinkage coefficient
estimated under the quadratic loss induced by the Frobenius norm,
following the one-factor shrinkage construction of
\citet{ledoit2003improved}. The estimator combines the parsimonious
dependence structure of the one-factor target with the information
contained in the raw sample covariance matrix, thereby stabilizing
covariance estimation when the number of series is large relative to the
available forecast-error history. In the empirical section, \textit{Factor} denotes the specification that sets
$\hat{\Sigma}_h=\hat{\Sigma}_1^{\mathrm{shrink}}$ when constructing
$\hat Q_h$, which is then used in the penalized quadratic program
\eqref{empirical-general-penalized}.

\subsection{Combination Weight Penalization}
\label{e-penalty}

The second implementation choice concerns how to penalize the combination weights. 
In finite samples, the weight vector itself can be a high-dimensional object, since each bottom-level series is associated with one direct forecast of itself and multiple indirect forecasts constructed from  aggregate-level forecasts. 
Estimated optimal weights may therefore have high variance. 
The combination framework makes this source of instability explicit and creates the possibility of regularizing the weights directly, rather than stabilizing only the covariance estimator. 
The key question is then how to design the penalty: that is, what structure should be encouraged in the estimated combination weights.

A convenient choice is to shrink the indirect-forecast weights toward zero. 
Let $\mathbf{w}_{i,2:(n_a + 1)}$ denote the subvector of $\mathbf{w}_i$ that collects the weights assigned to the $n_a$ indirect candidate forecasts for each bottom-level series $i$.
A zero-shrinkage penalty can be written as $\mathcal{P}^0(\mathbf{w})
=
\sum_{i=1}^{n_b}\|\mathbf{w}_{i,2:(n_a + 1)}\|_p^p$, 
where $p=2$ gives a ridge-type penalty and $p=1$ gives a LASSO-type penalty. The ridge penalty smoothly shrinks indirect weights toward zero, whereas the LASSO penalty can set some indirect weights exactly to zero, thereby selecting among indirect candidate forecasts. 
In the limit $\lambda\to\infty$, all indirect weights are forced to zero. The limiting reconciled bottom-level vector therefore equals the original bottom-level base forecasts, and the final multiplication by the summing matrix $S$ aggregates these forecasts to all higher levels. 
Thus, the limiting procedure is exactly the bottom-up method.

In view of the success of equal weighting in forecast combining, \citet{diebold2019machine} propose egalitarian penalization, which shrinks
estimated combination weights toward a simple average rather than toward zero. To implement this, we can set the penalty in~\eqref{empirical-general-penalized} as follows:
\begin{equation}
\label{eq:egalitarian-penalty}
\mathcal{P}^{\mathrm{eq}}(\mathbf{w})
=  \sum_{i=1}^{n_b}
  \|\mathbf{w}_{i,2:(n_a + 1)} - \frac{1}{n_a+1}\mathbf{1}_{n_a}\|_p^p,
\end{equation}
where, again, $p=2$ gives a ridge-type penalty and $p=1$ gives a LASSO-type penalty; we refer to the corresponding egalitarian versions as eRidge and eLASSO, respectively.
For later use, denote the contribution of bottom-level series $i$ to
this penalty by
$
\mathcal{P}_{i}^{\mathrm{eq}}(\mathbf{w}_i)
=\|
\mathbf{w}_{i,2:(n_a+1)}
-\frac{1}{n_a+1}\mathbf{1}_{n_a}\|_p^p.
$
Thus,
$\mathcal{P}^{\mathrm{eq}}(\mathbf{w})
=\sum_{i=1}^{n_b}\mathcal{P}_{i}^{\mathrm{eq}}(\mathbf{w}_i)$.
The penalty shrinks the indirect weights toward $\frac{1}{n_a+1}\mathbf{1}_{n_a}$. 
Through the sum-to-one constraint $\mathbf{1}_{n_a+1}^{\prime}\mathbf{w}_i=1$, this target corresponds exactly to the equal-weight allocation $\frac{1}{n_a+1}\mathbf{1}_{n_a+1}$ over all $n_a+1$ candidate forecasts, including the direct forecast. 
Thus, as $\lambda$ increases, the estimated weights are pulled toward a genuine combination benchmark rather than toward bottom-up reconciliation. 
In empirical implementation, $\lambda$ is treated as a tuning parameter rather than fixed a priori, and selected using holdout validation.

\citet{diebold2019machine} show that shrinkage toward equal weights can perform well in forecast combination, and further develop extensions that combine selection and egalitarian shrinkage, such as partially egalitarian LASSO and related trim-and-average rules. 
For simplicity, we focus here on the egalitarian penalty in~\eqref{eq:egalitarian-penalty}, which is easy to implement and isolates the role of weight-level regularization in reconciliation. 
More broadly, the key point is that the forecast-combination representation makes reconciliation weights explicit objects that can be regularized directly. 
This opens a route for incorporating  successful methods from the forecast combination literature into hierarchical and grouped forecast reconciliation.

\subsection{Joint and Separate Penalized Implementation}
\label{joint-sep-implementation}

We now describe the implementation of both the covariance-side choice in Section~\ref{shrink_estimator}
and the weight-side penalty in Section~\ref{e-penalty}. The covariance
estimator determines the empirical candidate-error matrix $\hat Q_h$, while
the penalty determines how the combination weights are regularized. For example, using the egalitarian penalty
in~\eqref{eq:egalitarian-penalty}, the full finite-sample problem is
\begin{equation}
\label{eq:pool-pen}
\hat{\mathbf w}^{\mathrm{joint}}
\in
\operatorname*{argmin}_{\mathbf w:\,\Pi \mathbf w=\mathbf 1_{n_b}}
\mathbf w^{\prime}\hat Q_h\mathbf w
+
\lambda\mathcal P^{\mathrm{eq}}(\mathbf w).
\end{equation}

We refer to~\eqref{eq:pool-pen} as the \textit{Joint} implementation. It
retains the full estimated  matrix $\hat Q_h$, including both the within-series
blocks $[\hat Q_h]_{i,i}$ and the cross-series blocks $[\hat Q_h]_{i,i'}$.
Thus, \emph{Joint} implementation uses the complete estimated covariance of the
candidate forecast errors.

A second implementation choice is to simplify $\hat Q_h$ before solving the
penalized problem. Motivated by the separability result in
Proposition~\ref{prop:equivalence}, we may retain only the within-series
candidate-error blocks and set all cross-series blocks to zero:
\begin{equation*}
\label{simp}
\hat{Q}_h^{\mathrm{bd}} =
\operatorname{blkdiag}([\hat{Q}_h]_{1,1},\ldots,[\hat{Q}_h]_{n_b,n_b}).
\end{equation*}
Replacing $\hat Q_h$ by $\hat Q_h^{\mathrm{bd}}$ yields the
\textit{Separate} implementation. Because the quadratic objective, the
egalitarian penalty, and the sum-to-one constraints all decompose across
bottom-level series, the problem separates into $n_b$ independent
subproblems. Specifically, for each $i=1,\ldots,n_b$,
\begin{equation}
\label{eq:sep-pen}
\hat{\mathbf{w}}_i^{\mathrm{sep}}
\in
\operatorname*{argmin}_{
\mathbf{w}_i:\,
\mathbf{1}_{n_a+1}^{\prime}\mathbf{w}_i=1}
\mathbf{w}_i^{\prime}[\hat Q_h]_{i,i}\mathbf{w}_i
+
\lambda\mathcal{P}_{i}^{\mathrm{eq}}(\mathbf{w}_i).
\end{equation}
The stacked \emph{Separate} solution obtained from \eqref{eq:sep-pen} is denoted
$\hat{\mathbf{w}}^{\mathrm{sep}}
=((\hat{\mathbf{w}}_1^{\mathrm{sep}})^{\prime},\ldots,
(\hat{\mathbf{w}}_{n_b}^{\mathrm{sep}})^{\prime})^{\prime}$.

The relationship between \emph{Joint} and \emph{Separate} depends on whether penalization is active. 
When $\lambda=0$, Proposition~\ref{prop:equivalence} and the plug-in principle imply that replacing $\hat Q_h$ by its block-diagonal counterpart does not change the solution: \emph{Joint} and \emph{Separate} coincide and recover the same plug-in MinT solution. 
When $\lambda>0$, this equivalence no longer holds in general. 
The penalty changes the optimization problem, and the cross-series blocks in the full $\hat Q_h$ can affect the \emph{Joint} optimum. 
Consequently, \emph{Joint} and \emph{Separate} should be viewed as two distinct penalized implementations rather than merely two computational routes to the same estimator. 
This distinction is made explicit in Section~\ref{divergence}, which provides examples where the two implementations differ under penalization. 
Section~\ref{solut-eRidge} further shows that, for the eRidge penalty, the resulting estimators admit closed-form solutions.

The \emph{Separate} implementation is particularly useful for both computational and statistical reasons. 
Computationally, it replaces a single constrained problem of dimension $n_b(n_a+1)$ with $n_b$ independent constrained problems of dimension $n_a+1$, each involving only the candidate forecasts for one bottom-level target series. These subproblems can be solved separately and in parallel, hence speeding up the computation and improving the scalability. 
Statistically, \emph{Separate} can be interpreted as a covariance-stabilized implementation from the combination perspective. 
It preserves the within-series candidate-error covariance matrix needed to choose the weights for each target series, while discarding cross-series covariance blocks that may be very noisy in finite samples. 
Thus, \emph{Separate} can reduce the computational burden while potentially improving statistical estimation stability.

Figure~\ref{fig:flowchart} summarizes the implementation pipeline: estimate the
base forecast error covariance $\hat{\Sigma}_h = \hat{\Sigma}_1$, construct either the full or
block-diagonal $\hat Q_h$, choose whether to impose weight penalty, and obtain the final coherent forecasts. It also highlights the relationship between \emph{Joint} and
\emph{Separate}: they coincide if there is no weight penalization, but can diverge once penalization is active.

\begin{figure}[htbp]
\centering
\begin{adjustbox}{max width=\textwidth, center}
\begin{tikzpicture}[
  font=\sffamily\footnotesize,
  >=Latex,
  arrow/.style={->, thick, line cap=round, line join=round},
  connector/.style={thick, line cap=round, line join=round},
  stage/.style={draw, rounded corners=2pt, thick, fill=gray!12,
                align=center, minimum width=12.8cm, minimum height=0.72cm},
  cov/.style={draw, rounded corners=2pt, thick, fill=gray!12,
              align=center, minimum width=4.65cm, minimum height=0.90cm},
  newcov/.style={draw, rounded corners=2pt, thick, fill=gray!12,
              align=center, minimum width=4.65cm, minimum height=0.90cm},
  qbox/.style={draw, rounded corners=2pt, thick, fill=gray!12,
              align=center, minimum width=4.30cm, minimum height=0.90cm},
  decision/.style={diamond, draw, thick, fill=gray!16,
              align=center, aspect=2.05, inner sep=2pt, minimum width=1.95cm},
  alg/.style={draw, rounded corners=2pt, thick, fill=gray!12,
              align=center, minimum width=3.55cm, minimum height=0.90cm},
  result/.style={draw, rounded corners=2pt, thick, fill=gray!12,
              align=center, minimum width=4.85cm, minimum height=0.92cm},
  final/.style={draw, rounded corners=2pt, very thick, fill=gray!20,
              align=center, minimum width=7.40cm, minimum height=0.98cm},
  note/.style={font=\sffamily\scriptsize, align=center}
]
\def\enterh{0.45}
\node[stage] (s1) at (0,0)
{\textbf{Step 1: Estimate one-step covariance proxy} $\widehat\Sigma_h = \widehat\Sigma_1$};
\node[cov] (existing) at (-3.55,-1.65)
{\textbf{Existing MinT variants}\\[-1pt]
OLS / WLSs / WLSv / MinT(Sample) / MinT(Shrink)};
\node[newcov] (factor) at (3.55,-1.65)
{\textbf{Proposed factor-based}\\[-1pt]
covariance estimation};
\draw[arrow] (s1.south) |- ($(existing.north)+(0,0.40)$) -- (existing.north);
\draw[arrow] (s1.south) |- ($(factor.north)+(0,0.40)$) -- (factor.north);
\node[stage] (s2) at (0,-3.55)
{\textbf{Step 2: Construct plug-in candidate forecast error covariance estimator}
$\widehat Q_h$};
\draw[arrow] (existing.south) |- ($(s2.north)+(-2.45,0.40)$) -- ($(s2.north)+(-2.45,0)$);
\draw[arrow] (factor.south) |- ($(s2.north)+(2.45,0.40)$) -- ($(s2.north)+(2.45,0)$);
\node[qbox] (fullQ) at (-3.55,-5.20)
{\textbf{Full} $\widehat Q_h$\\[-1pt]
with cross-series blocks};
\node[qbox] (bdQ) at (3.55,-5.20)
{\textbf{Block-diagonal} $\widehat Q_h^{\mathrm{bd}}$\\[-1pt]
per-series blocks};
\draw[arrow] (s2.south) |- ($(fullQ.north)+(0,0.40)$) -- (fullQ.north);
\draw[arrow] (s2.south) |- ($(bdQ.north)+(0,0.40)$) -- (bdQ.north);
\node[decision] (dec) at (0,-6.90) {$\lambda=0?$};
\draw[arrow] (fullQ.south) |- (dec.west);
\draw[arrow] (bdQ.south) |- (dec.east);
\node[alg] (joint0) at (-5.95,-9.45)
{Joint unpenalized QP\\using full $\widehat Q_h$};
\node[alg] (sep0) at (-2.00,-9.45)
{Separate unpenalized QPs\\Bates--Granger closed form};
\node[alg] (jointp) at (2.00,-9.45)
{Joint penalized QP\\full $\widehat Q_h$; eRidge/eLasso};
\node[alg] (sepp) at (5.95,-9.45)
{Separate penalized QPs\\$\widehat Q_h^{\mathrm{bd}}$; eRidge/eLasso};
\coordinate (split)  at (0,-7.45);
\coordinate (yleft)  at (-4.00,-7.45);
\coordinate (yright) at (4.00,-7.45);
\coordinate (algtopL) at ($(joint0.north)+(0,0.40)$);
\coordinate (ydownL) at (-4.00,0 |- algtopL);
\coordinate (algtopR) at ($(jointp.north)+(0,0.40)$);
\coordinate (ydownR) at (4.00,0 |- algtopR);
\draw[connector] (dec.south) -- (split);
\draw[connector] (split) -| (ydownL);
\draw[connector] (split) -| (ydownR);
\node[note] at (-2.05,-7.20) {Yes};
\node[note] at (2.05,-7.20) {No};
\draw[arrow] (ydownL) -| (joint0.north);
\draw[arrow] (ydownL) -| (sep0.north);
\draw[arrow] (ydownR) -| (jointp.north);
\draw[arrow] (ydownR) -| (sepp.north);
\node[result] (same) at (-3.95,-11.26)
{\textbf{Same unpenalized solution}};
\node[result] (diff) at (3.95,-11.26)
{\textbf{No guarantee of the same solution}};
\draw[arrow] (joint0.south) |- ($(same.north)+(-1.25,0.40)$) -- ($(same.north)+(-1.25,0)$);
\draw[arrow] (sep0.south) |- ($(same.north)+(1.25,0.40)$) -- ($(same.north)+(1.25,0)$);
\draw[arrow] (jointp.south) |- ($(diff.north)+(-1.25,0.40)$) -- ($(diff.north)+(-1.25,0)$);
\draw[arrow] (sepp.south) |- ($(diff.north)+(1.25,0.40)$) -- ($(diff.north)+(1.25,0)$);
\node[final] (final) at (0,-13.11)
{\textbf{Coherent reconciled forecasts}\\[-1pt]
$\widetilde{\mathbf y}_{t+h\mid t}
= S\Phi(\widehat{\mathbf w})C\widehat{\mathbf y}_{t+h\mid t}$};
\draw[arrow] (same.south) |- ($(final.north)+(-2.00,0.40)$) -- ($(final.north)+(-2.00,0)$);
\draw[arrow] (diff.south) |- ($(final.north)+(2.00,0.40)$) -- ($(final.north)+(2.00,0)$);
\end{tikzpicture}
\end{adjustbox}
\caption{Workflow of the proposed forecast-combination reconciliation framework.}
\label{fig:flowchart}
\end{figure}
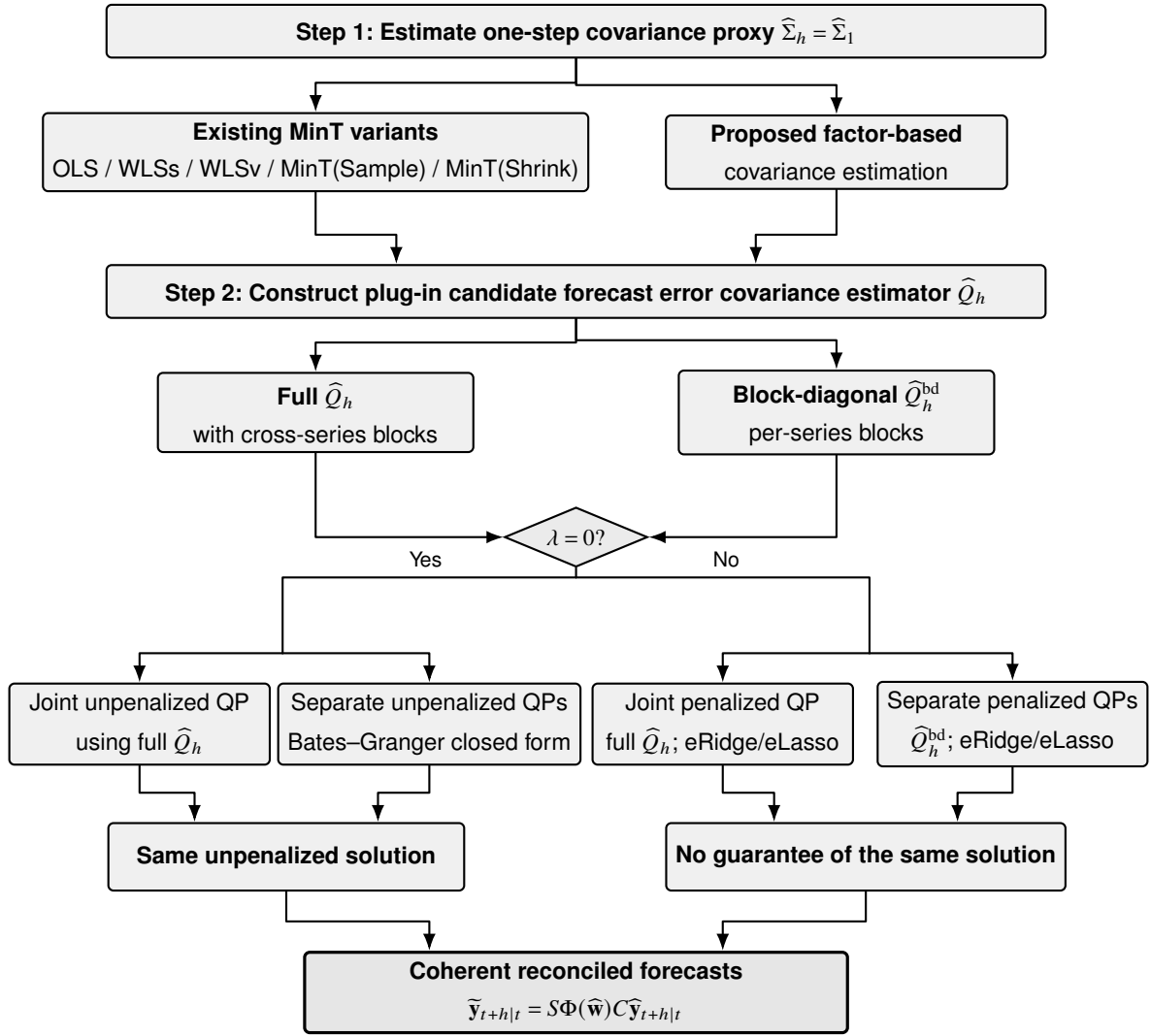

\section{Empirical Evaluation}
\label{sec:empirical}

This section evaluates the proposed  combination-based
reconciliation framework in Section~\ref{extensions} on two aggregation
structures: an Australian electricity generation hierarchy and a grouped
Australian labor force dataset. The empirical analysis has two aims.
First, within the proposed framework, we examine how forecast accuracy is
affected by the covariance specification, including standard MinT-style
estimators and the proposed {factor}-based estimator; by direct
egalitarian penalization of the combination weights; and by the choice
between \textit{Joint} and \textit{Separate} estimation when
penalization is active. Second, we compare the resulting methods with
standard and combination-based benchmarks, including unreconciled
\textit{Base} forecasts, \textit{Bottom-up}, MinT-style reconciliation,
\textit{LCC}, and its variant.

\subsection{Data, Design, Methods, and Evaluation}
\label{sec:empirical-design}
The first application uses the hierarchy of  daily Australian electricity
generation by energy source of~\citet{panagiotelis2023probabilistic}.
The data span June 11, 2019, to June 10, 2020, giving 366 daily observations.
Short-term forecasts of generation by energy source are important for grid
operation, reserve planning, and balancing intermittent renewable supply.
The hierarchy has four aggregation levels. Total generation is split
into renewable and non-renewable sources; renewable generation is
further divided into Batteries, Hydro (inc.\ Pumps), Solar, Wind, and
Biomass, while non-renewable generation is divided into Coal, Gas, and
Distillate. Several source categories are further decomposed into
detailed components, such as Solar (Rooftop) and Solar (Utility), Black
and Brown Coal, and four Gas technologies. Overall, the hierarchy
contains $n=23$ series, of which $n_b=15$ are bottom-level series and
$n_a=8$ are aggregated series. Figure~\ref{elct} displays the
series across hierarchy levels and illustrates the heterogeneous
seasonal and volatility patterns. 

The second application uses the Australian labor force dataset of~\citet{wang2025optimal}. It contains monthly unemployment series from January 2010 to
July 2023, giving 163 observations per series. Unlike the electricity
dataset, this is a grouped time-series structure formed by two
cross-classifying attributes: duration of job search and state or
territory (STT). Duration is divided into six groups, from under one
month to two years or more, and STT covers the eight Australian states
and territories: NSW, VIC, QLD, SA, WA, TAS, NT, and ACT. The resulting
grouped hierarchy contains one top-level series, 6 duration-level
series, 8 STT-level series, and  $n_b = 48$ duration~$\times$~STT bottom-level
series, for a total of $n = 63$ series. Figure~\ref{labour}
illustrates selected series, including the strong top-level seasonality and the more heterogeneous
lower-level grouped dynamics.

For both applications, we first fit ARIMA models from the
\pkg{forecast} package~\citep{Rforecast} separately to each series and
use the resulting forecasts as base forecasts. While this package selects among ARIMA models using an information criterion, in practice we acknowledge that model selection often also involves expert judgment \citep{petropoulos2023forecast}. Because the ARIMA models are estimated independently across the hierarchy or grouped structure, the base forecasts are not coherent. Reconciliation is therefore needed to achieve coherent forecasts: it takes the independently generated base forecasts as inputs, exploits the known aggregation constraints, and produces coherent forecasts across all levels. Within our framework, this leads to a common empirical
design in which different reconciliation procedures are compared through
their choices of covariance estimator, weight-level regularization, and
\textit{Joint} versus \textit{Separate} estimation.

We organize the reconciliation methods into two categories. The first contains unpenalized combining specifications, i.e.,
$\lambda=0$ in~\eqref{empirical-general-penalized}. For this category, the
implementation choice is the covariance estimator used for the one-step proxy
$\hat{\Sigma}_1$, which is then used as $\hat{\Sigma}_h$.
We consider three standard baseline covariance estimators from standard
MinT-style reconciliation in Table~\ref{tab:recon-methods}: the raw sample covariance for
\textit{MinT(Sample)}, the diagonal-shrinkage covariance for
\textit{MinT(Shrink)}, and the identity-based covariance specification for
\textit{OLS}. We label these specifications \textit{Raw Sample},
\textit{Diag Shrink}, and \textit{OLS}, respectively. By
Corollary~\ref{equ-mint-empirical}, their unpenalized combination-based versions
recover the corresponding MinT-style methods exactly. In addition, we include
the factor-based shrinkage estimator introduced in
Section~\ref{shrink_estimator}, labeled \textit{Factor}, as a new covariance
specification.
Because no penalty is imposed on the
combination weights, Proposition~\ref{prop:equivalence} implies that
the \textit{Joint} and \textit{Separate} formulations yield identical
weights. We therefore report a single unpenalized result for each covariance specification, labelled \textit{Joint} in the tables, with the understanding that it is identical to \textit{Separate} because $\lambda=0$.

The second category contains penalized combining weight specifications, which incorporates egalitarian combining weight regularization with each of
the four covariance specifications. For each of \textit{Raw Sample},
\textit{Diag Shrink}, \textit{OLS}, and \textit{Factor}, we consider
both eRidge and eLASSO penalties for the combining weights. With these penalties,
\textit{Joint} and \textit{Separate} are no longer equivalent in
general, as explained in Section~\ref{joint-sep-implementation}. We therefore
evaluate both strategies, yielding four penalized variants for each
covariance specification: \textit{Joint+eRidge},
\textit{Joint+eLASSO}, \textit{Separate+eRidge}, and
\textit{Separate+eLASSO}. 

In addition to these  reconciliation methods within our general framework, we report
four external benchmark methods for comparison: unreconciled \textit{Base}
forecasts, \textit{Bottom-up} forecasts, and two local-combination
methods, \textit{LCC}~\citep{hollyman2021understanding} and its variant labeled as  
\textit{LCC-Variant}~\citep{di2024forecast}, which are implemented
separately.
We evaluate forecast accuracy using the root mean squared error (RMSE) and RMSE skill score, which is the percentage improvement in RMSE
relative to the unreconciled \textit{Base} forecasts. Higher skill scores are therefore better. We report results averaged across forecast horizons and across series within each hierarchy level. Both applications
reserve the last 20\% of observations for rolling origin evaluation:
the electricity study uses 67 rolling windows with horizons
$h=1,\ldots,7$, and the labor force study uses 21 rolling windows with
horizons $h=1,\ldots,12$. Detailed implementation details are reported in Section~\ref{app:empirical-details}.

\subsection{Electricity Generation Results}
\label{sec:energy-results}

\begin{table}[ht]
\refstepcounter{table}
\label{energy_result}

\noindent
\begin{minipage}{\textwidth}
{\small
\renewcommand{\baselinestretch}{1}\selectfont
\textbf{Table~\thetable}\quad
Electricity generation RMSE and RMSE skill score
(\% relative to \textit{Base}) by hierarchy level and overall.
Gray rows denote classical and MinT-style benchmark methods.
Underlining indicates matching or improving on the corresponding
unpenalized \textit{Joint} baseline, boldface marks matching or
improving on all gray-row benchmarks, and $^{\dagger}$ marks
\textit{Separate} penalized variants that match or improve on their
\textit{Joint} penalized counterparts.
\par}
\end{minipage}

\vspace{4pt}

\centering
\scriptsize

  \centering
  \scriptsize
  \begin{adjustbox}{width=\linewidth,center}
  \begin{tabular}{crcccccrccccc}
  \toprule
  \multicolumn{1}{l}{} & \multicolumn{1}{l}{} & \multicolumn{4}{c}{RMSE} & & \multicolumn{1}{c}{} & \multicolumn{4}{c}{RMSE skill score (\%)} & \\
  \cline{3-6} \cline{9-12}
  \multicolumn{1}{l}{} & \multicolumn{1}{l}{} & Level 1 & Level 2 & Level 3 & Level 4 & & \multicolumn{1}{c}{} & Level 1 & Level 2 & Level 3 & Level 4 & \\
  \multicolumn{2}{c}{Method} & Total & \begin{tabular}[c]{@{}c@{}}Source\\type\end{tabular} & \begin{tabular}[c]{@{}c@{}}Energy\\category\end{tabular} & \begin{tabular}[c]{@{}c@{}}Detailed\\energy\\source\end{tabular} & All & \multicolumn{1}{c}{} & Total & \begin{tabular}[c]{@{}c@{}}Source\\type\end{tabular} & \begin{tabular}[c]{@{}c@{}}Energy\\category\end{tabular} & \begin{tabular}[c]{@{}c@{}}Detailed\\energy\\source\end{tabular} & All \\
  \midrule
  \rowcolor[HTML]{E7E6E6}
  \multicolumn{2}{r}{\cellcolor[HTML]{FFFFFF}\textit{Base}} & 18.24 & 21.47 & 11.63 & 8.48 & 11.58 & \cellcolor[HTML]{FFFFFF} & - &-  &- &- &- \\
  \rowcolor[HTML]{E7E6E6}
  \multicolumn{2}{r}{\cellcolor[HTML]{FFFFFF}\textit{Bottom-up}} & 23.50 & 22.54 & 11.52 & 8.48 & 12.12 & \cellcolor[HTML]{FFFFFF} & -28.8 & -5.0 & 0.9 & 0.0 & -4.7 \\
  \rowcolor[HTML]{E7E6E6}

  \rowcolor[HTML]{E7E6E6}
  \multicolumn{2}{r}{\cellcolor[HTML]{FFFFFF}\textit{LCC}} & 20.48 & 21.74 & 11.24 & 8.43 & 11.66 & \cellcolor[HTML]{FFFFFF} & -12.2 & -1.3 & 3.4 & 0.6 & -0.7 \\
   \rowcolor[HTML]{E7E6E6}
    \multicolumn{2}{r}{\cellcolor[HTML]{FFFFFF}\textit{LCC-Variant}} & 19.22 & 21.33 & 11.10 & 8.44 & 11.46 & \cellcolor[HTML]{FFFFFF} & -5.3 & 0.6 & 4.6 & 0.5 & 1.0 \\
  \midrule
  \rowcolor[HTML]{E7E6E6}
  \cellcolor[HTML]{FFFFFF} & \makecell[l]{\textit{Joint}\\\textit{(MinT-Sample)}} & 18.33 & 20.77 & 10.77 & 8.30 & 11.16 & \cellcolor[HTML]{FFFFFF} & -0.5 & 3.3 & 7.4 & 2.2 & 3.6 \\
  \rowcolor[HTML]{FFFFFF}
  \cellcolor[HTML]{FFFFFF} & \textit{Joint + eRidge} & \underline{18.25} & \underline{20.67} & 10.78 & \underline{8.26} & \underline{11.12} & \cellcolor[HTML]{FFFFFF} & \underline{-0.0} & \underline{3.7} & 7.3 & \underline{2.6} & \underline{3.9} \\
  \rowcolor[HTML]{FFFFFF}
  \cellcolor[HTML]{FFFFFF} & \textit{Joint + eLASSO} & \underline{18.25} & 20.79 & 10.78 & \underline{8.28} & \underline{11.15} & \cellcolor[HTML]{FFFFFF} & \underline{-0.0} & 3.2 & 7.3 & \underline{2.3} & \underline{3.7} \\
  \rowcolor[HTML]{FFFFFF}
  \cellcolor[HTML]{FFFFFF} & \textit{Separate + eRidge} & \underline{\textbf{18.07}}$^{\dagger}$ & \underline{20.62}$^{\dagger}$ & \underline{\textbf{10.74}}$^{\dagger}$ & \underline{\textbf{8.25}}$^{\dagger}$ & \underline{\textbf{11.08}}$^{\dagger}$ & \cellcolor[HTML]{FFFFFF} & \underline{\textbf{0.9}}$^{\dagger}$ & \underline{4.0}$^{\dagger}$ & \underline{\textbf{7.6}}$^{\dagger}$ & \underline{\textbf{2.7}}$^{\dagger}$ & \underline{\textbf{4.2}}$^{\dagger}$ \\
  \rowcolor[HTML]{FFFFFF}
  \multirow{-5}{*}{\cellcolor[HTML]{FFFFFF}Raw Sample} & \textit{Separate + eLASSO} & \underline{\textbf{18.00}}$^{\dagger}$ & \underline{20.68}$^{\dagger}$ & \underline{\textbf{10.74}}$^{\dagger}$ & \underline{8.27}$^{\dagger}$ & \underline{11.10}$^{\dagger}$ & \cellcolor[HTML]{FFFFFF} & \underline{\textbf{1.3}}$^{\dagger}$ & \underline{3.7}$^{\dagger}$ & \underline{\textbf{7.6}}$^{\dagger}$ & \underline{2.5}$^{\dagger}$ & \underline{4.1}$^{\dagger}$ \\
  \midrule
  \rowcolor[HTML]{E7E6E6}
  \cellcolor[HTML]{FFFFFF} & \makecell[l]{\textit{Joint}\\\textit{(MinT-Shrink)}} & 18.18 & 20.54 & 10.76 & 8.25 & 11.09 & \cellcolor[HTML]{FFFFFF} & 0.3 & 4.3 & 7.5 & 2.7 & 4.2 \\
  \rowcolor[HTML]{FFFFFF}
  \cellcolor[HTML]{FFFFFF} & \textit{Joint + eRidge} & 18.19 & \underline{\textbf{20.53}} & \underline{\textbf{10.75}} & \underline{\textbf{8.25}} & \underline{\textbf{11.09}} & \cellcolor[HTML]{FFFFFF} & \underline{\textbf{0.3}} & \underline{\textbf{4.4}} & \underline{\textbf{7.6}} & \underline{\textbf{2.7}} & \underline{\textbf{4.2}} \\
  \rowcolor[HTML]{FFFFFF}
  \cellcolor[HTML]{FFFFFF} & \textit{Joint + eLASSO} & 18.19 & \underline{\textbf{20.53}} & \underline{\textbf{10.75}} & \underline{\textbf{8.25}} & \underline{\textbf{11.09}} & \cellcolor[HTML]{FFFFFF} & \underline{\textbf{0.3}} & \underline{\textbf{4.4}} & \underline{\textbf{7.6}} & \underline{\textbf{2.7}} & \underline{\textbf{4.2}} \\
  \rowcolor[HTML]{FFFFFF}
  \cellcolor[HTML]{FFFFFF} & \textit{Separate + eRidge} & \underline{\textbf{18.14}}$^{\dagger}$ & \underline{\textbf{20.53}}$^{\dagger}$ & \underline{\textbf{10.76}} & \underline{\textbf{8.25}}$^{\dagger}$ & \underline{\textbf{11.09}}$^{\dagger}$ & \cellcolor[HTML]{FFFFFF} & \underline{\textbf{0.5}}$^{\dagger}$ & \underline{\textbf{4.4}}$^{\dagger}$ & \underline{\textbf{7.5}} & \underline{\textbf{2.7}}$^{\dagger}$ & \underline{\textbf{4.2}}$^{\dagger}$ \\
  \rowcolor[HTML]{FFFFFF}
  \multirow{-5}{*}{\cellcolor[HTML]{FFFFFF}Diag Shrink} & \textit{Separate + eLASSO} & \underline{\textbf{18.16}}$^{\dagger}$ & \underline{\textbf{20.53}}$^{\dagger}$ & \underline{\textbf{10.76}} & \underline{\textbf{8.25}}$^{\dagger}$ & \underline{\textbf{11.09}}$^{\dagger}$ & \cellcolor[HTML]{FFFFFF} & \underline{\textbf{0.4}}$^{\dagger}$ & \underline{\textbf{4.4}}$^{\dagger}$ & \underline{\textbf{7.5}} & \underline{\textbf{2.7}}$^{\dagger}$ & \underline{\textbf{4.2}}$^{\dagger}$ \\
  \midrule
  \rowcolor[HTML]{E7E6E6}
  \cellcolor[HTML]{FFFFFF} & \makecell[l]{\textit{Joint}\\\textit{(Standard OLS)}} & 18.20 & 21.03 & 11.26 & 8.46 & 11.39 & \cellcolor[HTML]{FFFFFF} & 0.2 & 2.0 & 3.2 & 0.2 & 1.6 \\
  \rowcolor[HTML]{FFFFFF}
  \cellcolor[HTML]{FFFFFF} & \textit{Joint + eRidge} & 18.26 & \underline{20.97} & \underline{11.19} & \underline{8.42} & \underline{11.35} & \cellcolor[HTML]{FFFFFF} & -0.1 & \underline{2.3} & \underline{3.8} & \underline{0.7} & \underline{2.0} \\
  \rowcolor[HTML]{FFFFFF}
  \cellcolor[HTML]{FFFFFF} & \textit{Joint + eLASSO} & \underline{\textbf{18.16}} & 21.04 & \underline{11.22} & \underline{8.44} & \underline{11.36} & \cellcolor[HTML]{FFFFFF} & \underline{\textbf{0.5}} & \underline{2.0} & \underline{3.6} & \underline{0.5} & \underline{1.8} \\
  \rowcolor[HTML]{FFFFFF}
  \cellcolor[HTML]{FFFFFF} & \textit{Separate + eRidge} & 18.28 & \underline{20.92}$^{\dagger}$ & \underline{11.23} & \underline{8.43} & \underline{11.35}$^{\dagger}$ & \cellcolor[HTML]{FFFFFF} & -0.2 & \underline{2.5}$^{\dagger}$ & \underline{3.5} & \underline{0.6} & \underline{1.9} \\
  \rowcolor[HTML]{FFFFFF}
  \multirow{-5}{*}{\cellcolor[HTML]{FFFFFF}OLS} & \textit{Separate + eLASSO} & \underline{\textbf{18.15}}$^{\dagger}$ & \underline{21.00}$^{\dagger}$ & \underline{11.22}$^{\dagger}$ & \underline{8.45} & \underline{11.36}$^{\dagger}$ & \cellcolor[HTML]{FFFFFF} & \underline{\textbf{0.5}}$^{\dagger}$ & \underline{2.2}$^{\dagger}$ & \underline{3.5} & \underline{0.4} & \underline{1.8}$^{\dagger}$ \\
  \midrule
  \rowcolor[HTML]{FFFFFF}
  \cellcolor[HTML]{FFFFFF} & \cellcolor[HTML]{FFFFFF}\textit{Joint} & \textbf{18.18} & 20.57 & \textbf{10.72} & \textbf{8.24} & \textbf{11.08} & \cellcolor[HTML]{FFFFFF} & \textbf{0.3} & 4.2 & \textbf{7.8} & \textbf{2.9} & \textbf{4.3} \\
  \rowcolor[HTML]{FFFFFF}
  \cellcolor[HTML]{FFFFFF} & \textit{Joint + eRidge} & \underline{\textbf{18.17}} & \underline{20.55} & \underline{\textbf{10.71}} & \underline{\textbf{8.23}} & \underline{\textbf{11.07}} & \cellcolor[HTML]{FFFFFF} & \underline{\textbf{0.4}} & \underline{\textbf{4.3}} & \underline{\textbf{7.9}} & \underline{\textbf{2.9}} & \underline{\textbf{4.4}} \\
  \rowcolor[HTML]{FFFFFF}
  \cellcolor[HTML]{FFFFFF} & \textit{Joint + eLASSO} & 18.19 & \underline{20.55} & \underline{\textbf{10.71}} & \underline{\textbf{8.24}} & \underline{\textbf{11.07}} & \cellcolor[HTML]{FFFFFF} & \underline{\textbf{0.3}} & \underline{\textbf{4.3}} & \underline{\textbf{7.9}} & \underline{\textbf{2.9}} & \underline{\textbf{4.4}} \\
  \rowcolor[HTML]{FFFFFF}
  \cellcolor[HTML]{FFFFFF} & \textit{Separate + eRidge} & \underline{\textbf{18.10}}$^{\dagger}$ & \underline{\textbf{20.53}}$^{\dagger}$ & \underline{\textbf{10.72}} & \underline{\textbf{8.23}}$^{\dagger}$ & \underline{\textbf{11.06}}$^{\dagger}$ & \cellcolor[HTML]{FFFFFF} & \underline{\textbf{0.8}}$^{\dagger}$ & \underline{\textbf{4.4}}$^{\dagger}$ & \underline{\textbf{7.8}} & \underline{\textbf{3.0}}$^{\dagger}$ & \underline{\textbf{4.4}}$^{\dagger}$ \\
  \rowcolor[HTML]{FFFFFF}
  \multirow{-5}{*}{\cellcolor[HTML]{FFFFFF}Factor} & \textit{Separate + eLASSO} & \underline{\textbf{18.08}}$^{\dagger}$ & \underline{\textbf{20.54}}$^{\dagger}$ & \underline{\textbf{10.72}} & \underline{\textbf{8.24}}$^{\dagger}$ & \underline{\textbf{11.07}}$^{\dagger}$ & \cellcolor[HTML]{FFFFFF} & \underline{\textbf{0.9}}$^{\dagger}$ & \underline{\textbf{4.3}}$^{\dagger}$ & \underline{\textbf{7.8}} & \underline{\textbf{2.9}}$^{\dagger}$ & \underline{\textbf{4.4}}$^{\dagger}$ \\
  \bottomrule
  \end{tabular}
  \end{adjustbox}

\end{table}

Table~\ref{energy_result} reports RMSE values and skill
scores by hierarchy level and overall. Shaded rows denote classical
benchmarks and MinT-style unpenalized baselines; underlining indicates
performance matching or improving on the corresponding unpenalized
\textit{Joint} baseline; boldface indicates performance matching or
exceeding all gray-row benchmarks; and $^{\dagger}$ marks
\textit{Separate} variants that match or improve on their corresponding
\textit{Joint} penalized result.

\paragraph{Bottom-up and LCC-type benchmarks.}
The proposed methods improve clearly on the non-MinT benchmarks in the first four rows of the table. Interestingly, \textit{LCC} is worse
than \textit{Base} at Levels~1--2, while \textit{LCC-Variant} only
moves ahead of \textit{Base} at Levels~2--4 and remains weak at Level~1.
This comparison is particularly meaningful because LCC-type methods are
the closest existing approaches to a forecast-combination interpretation
of reconciliation. As discussed in
Sections~\ref{sec: connection} and~\ref{sec:LCC}, their local two-level
reconciliation and restricted candidate sets can leave useful
cross-level information unexploited. Our proposed global combination
framework avoids this restriction and delivers uniformly lower overall
RMSE in Table~\ref{energy_result}.

\paragraph{Factor covariance specification.}
Among the unpenalized configurations, \textit{Factor+Joint} gives the
lowest overall RMSE, $11.08$, slightly improving on the
strongest gray-row benchmark, \textit{MinT(Shrink)} at $11.09$. With egalitarian penalization,
\textit{Factor+Separate+eRidge} achieves the best overall result in
Table~\ref{energy_result}, with RMSE $11.06$. It also
matches or improves on the strongest benchmark at every hierarchy level. 
Although the best Level~1 result among all proposed variants is attained
by \textit{Raw Sample+Separate+eLASSO} ($18.00$), the \textit{Factor}
variants are consistently among the strongest methods. Indeed, the table suggests that the factor-based shrinkage covariance
estimation provides a stable basis for reconciliation.

\paragraph{Egalitarian penalization for the combining weights.}
At the overall level, egalitarian penalization matches or improves the
corresponding unpenalized \textit{Joint} baseline in all  configurations. The largest gain occurs under
\textit{Raw Sample}, where the best penalized result reduces overall
RMSE from $11.16$ to $11.08$. Under \textit{OLS}, the best penalized
result improves RMSE from $11.39$ to $11.35$. By contrast,
\textit{Diag Shrink} is already strong and remains essentially
unchanged at the displayed precision ($11.09$ throughout), while
\textit{Factor} improves only slightly from $11.08$ to $11.06$. Thus,
imposing the penalty on the combining weights is most useful when the covariance estimator is less
regularized, and mainly provides small stabilization gains when the
baseline covariance estimator is already strong.

\paragraph{Joint versus Separate under combining weight penalization.}
Under penalization, the \textit{Separate} strategy performs very similarly
to, and often slightly better than, its \textit{Joint} counterpart. This similarity is
consistent with the relatively compact bottom level ($n_b=15$), where
the cross-block components of $\hat{Q}_h$ are present but not sufficiently difficult to estimate. 
Hence, in this application the benefit of \textit{Separate}  
is mainly computational:
because \textit{Separate} decomposes the penalized estimation into
independent bottom-series problems, it is  simpler to solve than the
full \textit{Joint} optimisation over all combination weights.
Table~\ref{energy_result} shows that this computational simplification
does not come at the cost of forecast accuracy.

\subsection{Grouped Labor Force Results}
\label{sec:labour-results}

\begin{table}[ht]
\refstepcounter{table}
\label{labour_result}

\noindent
\begin{minipage}{\textwidth}
{\small
\renewcommand{\baselinestretch}{1}\selectfont
\textbf{Table~\thetable}\quad
Labor force RMSE and RMSE skill score
(\% relative to \textit{Base}) by hierarchy level and overall.
STT denotes State or Territory.
Shading, underlining, boldface, and $^{\dagger}$ follow the same
conventions as in Table~\ref{energy_result}.
\par}
\end{minipage}

\vspace{4pt}

\centering
\scriptsize
    \centering
    \scriptsize
    \begin{adjustbox}{width=\linewidth,center}
    \begin{tabular}{crcccccrccccc}
    \toprule
    \multicolumn{1}{l}{} & \multicolumn{1}{l}{} & \multicolumn{4}{c}{RMSE} & & \multicolumn{1}{c}{} & \multicolumn{4}{c}{RMSE skill score (\%)} & \\
    \cline{3-6} \cline{9-12}
    \multicolumn{1}{l}{} & \multicolumn{1}{l}{} & Level 1 & Level 2 & Level 3 & Level 4 & & \multicolumn{1}{c}{} & Level 1 & Level 2 & Level 3 & Level 4 & \\
    \multicolumn{2}{c}{Method} & Total & \begin{tabular}[c]{@{}c@{}}Duration\end{tabular} & \begin{tabular}[c]{@{}c@{}}STT\end{tabular} & \begin{tabular}[c]{@{}c@{}}Duration\\$\times$ STT \end{tabular} & All & \multicolumn{1}{c}{} & Total & \begin{tabular}[c]{@{}c@{}}Duration\end{tabular} & \begin{tabular}[c]{@{}c@{}}STT\end{tabular} & \begin{tabular}[c]{@{}c@{}}Duration\\$\times$ STT\end{tabular} & All \\
    \midrule
    \rowcolor[HTML]{E7E6E6}
    \multicolumn{2}{r}{\cellcolor[HTML]{FFFFFF}\textit{Base}} & 115.57 & 46.41 & 28.85 & 8.19 & 24.84 & \cellcolor[HTML]{FFFFFF} & - & - &- &- &- \\
    \rowcolor[HTML]{E7E6E6}
    \multicolumn{2}{r}{\cellcolor[HTML]{FFFFFF}\textit{Bottom-up}} & 197.33 & 42.45 & 34.34 & 8.19 & 31.56 & \cellcolor[HTML]{FFFFFF} & -70.7 & 8.5 & -19.0 & 0.0 & -27.1 \\
      \rowcolor[HTML]{E7E6E6}
    \multicolumn{2}{r}{\cellcolor[HTML]{FFFFFF}\textit{LCC}} & 177.10 & 39.45 & 30.98 & 7.74 & 28.57 & \cellcolor[HTML]{FFFFFF} & -53.2 & 15.0 & -7.4 & 5.4 & -15.0 \\
    \rowcolor[HTML]{E7E6E6}
    \multicolumn{2}{r}{\cellcolor[HTML]{FFFFFF}\textit{LCC-Variant}} & 155.80 & 37.20 & 27.72 & 7.47 & 25.69 & \cellcolor[HTML]{FFFFFF} & -34.8 & 19.9 & 3.9 & 8.8 & -3.4 \\

    \midrule
    \rowcolor[HTML]{E7E6E6}
    \cellcolor[HTML]{FFFFFF} & \makecell[l]{\textit{Joint}\\\textit{(MinT-Sample)}} & 139.84 & 36.71 & 26.37 & 8.00 & 24.14 & \cellcolor[HTML]{FFFFFF} & -21.0 & 20.9 & 8.6 & 2.2 & 2.8 \\
    \rowcolor[HTML]{FFFFFF}
    \cellcolor[HTML]{FFFFFF} & \textit{Joint + eRidge} & \underline{134.96} & \underline{35.86} & 26.92 & \underline{7.54} & \underline{23.51} & \cellcolor[HTML]{FFFFFF} & \underline{-16.8} & \underline{22.7} & 6.7 & \underline{7.9} & \underline{5.3} \\
    \rowcolor[HTML]{FFFFFF}
    \cellcolor[HTML]{FFFFFF} & \textit{Joint + eLASSO} & \underline{137.32} & \underline{36.17} & 26.47 & \underline{7.72} & \underline{23.75} & \cellcolor[HTML]{FFFFFF} & \underline{-18.8} & \underline{22.1} & 8.2 & \underline{5.7} & \underline{4.4} \\
    \rowcolor[HTML]{FFFFFF}
    \cellcolor[HTML]{FFFFFF} & \textit{Separate + eRidge} & \underline{129.97}$^{\dagger}$ & \underline{35.57}$^{\dagger}$ & 26.86$^{\dagger}$ & \underline{7.73} & \underline{23.13}$^{\dagger}$ & \cellcolor[HTML]{FFFFFF} & \underline{-12.5}$^{\dagger}$ & \underline{23.4}$^{\dagger}$ & 6.9$^{\dagger}$ & \underline{5.5} & \underline{6.9}$^{\dagger}$ \\
    \rowcolor[HTML]{FFFFFF}
    \multirow{-5}{*}{\cellcolor[HTML]{FFFFFF}Raw Sample} & \textit{Separate + eLASSO} & \underline{125.71}$^{\dagger}$ & \underline{\textbf{35.35}}$^{\dagger}$ & \underline{\textbf{26.09}}$^{\dagger}$ & \underline{7.80} & \underline{\textbf{22.62}}$^{\dagger}$ & \cellcolor[HTML]{FFFFFF} & \underline{-8.8}$^{\dagger}$ & \underline{\textbf{23.8}}$^{\dagger}$ & \underline{\textbf{9.6}}$^{\dagger}$ & \underline{4.7} & \underline{\textbf{8.9}}$^{\dagger}$ \\
    \midrule
    \rowcolor[HTML]{E7E6E6}
    \cellcolor[HTML]{FFFFFF} & \makecell[l]{\textit{Joint}\\\textit{(MinT-Shrink)}} & 148.79 & 35.65 & 27.35 & 7.34 & 24.74 & \cellcolor[HTML]{FFFFFF} & -28.8 & 23.2 & 5.2 & 10.4 & 0.4 \\
    \rowcolor[HTML]{FFFFFF}
    \cellcolor[HTML]{FFFFFF} & \textit{Joint + eRidge} & \underline{148.20} & 35.78 & 27.52 & 7.35 & \underline{24.73} & \cellcolor[HTML]{FFFFFF} & \underline{-28.2} & 22.9 & 4.6 & 10.2 & \underline{0.4} \\
    \rowcolor[HTML]{FFFFFF}
    \cellcolor[HTML]{FFFFFF} & \textit{Joint + eLASSO} & \underline{148.66} & 35.67 & 27.44 & 7.35 & 24.75 & \cellcolor[HTML]{FFFFFF} & \underline{-28.6} & 23.1 & 4.9 & 10.3 & \underline{0.4} \\
    \rowcolor[HTML]{FFFFFF}
    \cellcolor[HTML]{FFFFFF} & \textit{Separate + eRidge} & \underline{143.44}$^{\dagger}$ & \underline{35.50}$^{\dagger}$ & 27.41$^{\dagger}$ & 7.37 & \underline{24.25}$^{\dagger}$ & \cellcolor[HTML]{FFFFFF} & \underline{-24.1}$^{\dagger}$ & \underline{23.5}$^{\dagger}$ & 5.0$^{\dagger}$ & 9.9 & \underline{2.4}$^{\dagger}$ \\
    \rowcolor[HTML]{FFFFFF}
    \multirow{-5}{*}{\cellcolor[HTML]{FFFFFF}Diag Shrink} & \textit{Separate + eLASSO} & \underline{145.11}$^{\dagger}$ & \underline{\textbf{35.43}}$^{\dagger}$ & \underline{27.35}$^{\dagger}$ & 7.35$^{\dagger}$ & \underline{24.37}$^{\dagger}$ & \cellcolor[HTML]{FFFFFF} & \underline{-25.6}$^{\dagger}$ & \underline{\textbf{23.7}}$^{\dagger}$ & \underline{5.2}$^{\dagger}$ & 10.2 & \underline{1.9}$^{\dagger}$ \\
    \midrule
    \rowcolor[HTML]{E7E6E6}
    \cellcolor[HTML]{FFFFFF} & \makecell[l]{\textit{Joint}\\\textit{(Standard OLS)}} & 127.69 & 35.48 & 27.29 & 7.47 & 22.95 & \cellcolor[HTML]{FFFFFF} & -10.5 & 23.6 & 5.4 & 8.8 & 7.6 \\
    \rowcolor[HTML]{FFFFFF}
    \cellcolor[HTML]{FFFFFF} & \textit{Joint + eRidge} & \underline{127.65} & 35.49 & \underline{27.29} & \underline{7.47} & \underline{\textbf{22.95}} & \cellcolor[HTML]{FFFFFF} & \underline{-10.5} & 23.5 & \underline{5.4} & \underline{8.8} & \underline{\textbf{7.6}} \\
    \rowcolor[HTML]{FFFFFF}
    \cellcolor[HTML]{FFFFFF} & \textit{Joint + eLASSO} & \underline{127.69} & \underline{\textbf{35.48}} & \underline{27.29} & \underline{7.47} & \underline{\textbf{22.95}} & \cellcolor[HTML]{FFFFFF} & \underline{-10.5} & \underline{\textbf{23.6}} & \underline{5.4} & \underline{8.8} & \underline{\textbf{7.6}} \\
    \rowcolor[HTML]{FFFFFF}
    \cellcolor[HTML]{FFFFFF} & \textit{Separate + eRidge} & \underline{127.13}$^{\dagger}$ & \underline{\textbf{35.45}}$^{\dagger}$ & \underline{27.26}$^{\dagger}$ & \underline{7.46}$^{\dagger}$ & \underline{\textbf{22.90}}$^{\dagger}$ & \cellcolor[HTML]{FFFFFF} & \underline{-10.0}$^{\dagger}$ & \underline{\textbf{23.6}}$^{\dagger}$ & \underline{5.5}$^{\dagger}$ & \underline{8.8}$^{\dagger}$ & \underline{\textbf{7.8}}$^{\dagger}$ \\
    \rowcolor[HTML]{FFFFFF}
    \multirow{-5}{*}{\cellcolor[HTML]{FFFFFF}OLS} & \textit{Separate + eLASSO} & 127.77 & 35.51 & 27.31 & \underline{7.47}$^{\dagger}$ & 22.96 & \cellcolor[HTML]{FFFFFF} & -10.6 & 23.5 & 5.3 & 8.7 & 7.5 \\
    \midrule
    \rowcolor[HTML]{FFFFFF}
    \cellcolor[HTML]{FFFFFF} & \cellcolor[HTML]{FFFFFF}\textit{Joint} & 130.71 & \textbf{34.39} & \textbf{25.00} & \textbf{7.26} & \textbf{22.60} & \cellcolor[HTML]{FFFFFF} & -13.1 & \textbf{25.9} & \textbf{13.3} & \textbf{11.3} & \textbf{9.0} \\
    \rowcolor[HTML]{FFFFFF}
    \cellcolor[HTML]{FFFFFF} & \textit{Joint + eRidge} & \underline{129.95} & \textbf{34.41} & \textbf{25.31} & \underline{\textbf{7.24}} & \underline{\textbf{22.58}} & \cellcolor[HTML]{FFFFFF} & \underline{-12.4} & \underline{\textbf{25.9}} & \textbf{12.3} & \underline{\textbf{11.6}} & \underline{\textbf{9.1}} \\
    \rowcolor[HTML]{FFFFFF}
    \cellcolor[HTML]{FFFFFF} & \textit{Joint + eLASSO} & \underline{130.37} & \underline{\textbf{34.37}} & \textbf{25.12} & \underline{\textbf{7.25}} & \underline{\textbf{22.59}} & \cellcolor[HTML]{FFFFFF} & \underline{-12.8} & \underline{\textbf{25.9}} & \textbf{12.9} & \underline{\textbf{11.4}} & \underline{\textbf{9.1}} \\
    \rowcolor[HTML]{FFFFFF}
    \cellcolor[HTML]{FFFFFF} & \textit{Separate + eRidge} & \underline{128.43}$^{\dagger}$ & \underline{\textbf{34.29}}$^{\dagger}$ & \textbf{25.19}$^{\dagger}$ & \textbf{7.28} & \underline{\textbf{22.42}}$^{\dagger}$ & \cellcolor[HTML]{FFFFFF} & \underline{-11.1}$^{\dagger}$ & \underline{\textbf{26.1}}$^{\dagger}$ & \textbf{12.7}$^{\dagger}$ & \textbf{11.1} & \underline{\textbf{9.7}}$^{\dagger}$ \\
    \rowcolor[HTML]{FFFFFF}
    \multirow{-5}{*}{\cellcolor[HTML]{FFFFFF}Factor} & \textit{Separate + eLASSO} & \underline{128.02}$^{\dagger}$ & \underline{\textbf{34.23}}$^{\dagger}$ & \underline{\textbf{24.99}}$^{\dagger}$ & \underline{\textbf{7.26}} & \underline{\textbf{22.34}}$^{\dagger}$ & \cellcolor[HTML]{FFFFFF} & \underline{-10.8}$^{\dagger}$ & \underline{\textbf{26.2}}$^{\dagger}$ & \underline{\textbf{13.4}}$^{\dagger}$ & \underline{\textbf{11.3}} & \underline{\textbf{10.1}}$^{\dagger}$ \\
    \bottomrule
    \end{tabular}
    \end{adjustbox}

  \end{table}

  Table~\ref{labour_result} reports the forecasting results for the grouped
  labor force  dataset. Relative to the previous electricity hierarchy,
  this application has significantly larger size
  ($n_b=48$ versus $n_b=15$, $n = 63$ versus $n = 23$). Therefore, this application provides a useful context for assessing whether the proposed framework remains effective in a larger grouped setting.

  We first note that the relative performance of the benchmarks in the gray rows differs sharply from Table~\ref{energy_result}.
  Whereas \textit{MinT(Shrink)} is the strongest benchmark at all
  levels in the electricity study, no single benchmark dominates here:
  \textit{Base} is strongest at Level~1 ($115.57$), \textit{OLS} at
  Level~2 ($35.48$) and overall ($22.95$), \textit{MinT(Sample)} at
  Level~3 ($26.37$), and \textit{MinT(Shrink)} at Level~4 ($7.34$). This
  suggests the labor force application poses a less uniform
  reconciliation problem than the previous electricity application.

  \paragraph{Bottom-up and LCC-type benchmarks.}
The \textit{Bottom-up}, \textit{LCC}, and \textit{LCC-Variant} benchmarks perform unevenly on the labor force dataset. Specifically, they improve on
\textit{Base} at some lower levels, especially Level~2 and Level~4, but
all three are worse than \textit{Base}
overall. This contrasts with Table~\ref{energy_result}, where
\textit{LCC-Variant} achieved a modest positive overall skill score.
 Thus, as in the electricity study but more sharply here,
 the local-combination benchmarks are less competitive overall than the proposed global combination framework.

  \paragraph{Factor covariance specification.}
  The strongest results in Table~\ref{labour_result} again favor the
  \textit{Factor} covariance specification, with the advantage being more
  pronounced than in Table~\ref{energy_result}. Among unpenalized
  configurations, \textit{Factor+Joint}  attains the lowest overall RMSE,
  $22.60$, improving on the strongest gray-row benchmark,
  \textit{OLS}, at $22.95$. With penalization,
  \textit{Factor+Separate+eLASSO} achieves the lowest overall RMSE in the
  table, $22.34$, with
  \textit{Factor+Separate+eRidge} close behind at $22.42$.
  All five \textit{Factor} variants outperform the strongest benchmark at
  Levels~2--4 and overall. The exception is Level~1, where
  no reconciliation method improves on the unreconciled \textit{Base}
  forecasts, suggesting limited room for improvement once aggregation has
  already smoothed the top-level series.

  \paragraph{Egalitarian penalization for the combining weights.}
  The benefit of egalitarian penalization is similar to that in
  Table~\ref{energy_result}, but the gains are larger. At the overall level, penalization matches or improves the corresponding unpenalized \textit{Joint} baseline in most configurations. The largest gain occurs under \textit{Raw Sample}, where \textit{Separate+eLASSO} reduces overall RMSE from $24.14$ to $22.62$. 
  Under \textit{OLS}, the improvement is much smaller, from $22.95$ to
  $22.90$ for \textit{Separate+eRidge}. Under \textit{Factor}, the
  unpenalized baseline is already strong, but penalization combined with
  \textit{Separate} still reduces RMSE from $22.60$ to $22.34$. By
  contrast, \textit{Diag Shrink} remains relatively weak even
  after penalization. These results reinforce the pattern from
  Table~\ref{energy_result}: the penalty is most valuable when the
  covariance estimator provides less stabilization on its own.

  \paragraph{Joint versus Separate under combining weight penalization.}
The contrast between \textit{Joint} and \textit{Separate} is stronger
on the labor force dataset than on the electricity dataset. At the
overall level, \textit{Separate} matches or improves its
\textit{Joint} counterpart in most
cases. The
largest improvement appears for \textit{Raw Sample}, where
\textit{Separate+eLASSO} improves on \textit{Joint+eLASSO} by $1.13$
RMSE. For \textit{Factor}, the two \textit{Separate} variants improve
on their \textit{Joint} counterparts by $0.16$--$0.25$ RMSE. Hence, in
this larger grouped setting, the 
\textit{Separate} strategy is not only computationally simpler but  also more
accurate.

\subsection{Summary of Empirical Findings}

Taken together, the two empirical studies show that the proposed
combination-based reconciliation framework is not merely an alternative
implementation of existing reconciliation methods, but a useful way to
organize and extend them. Across both applications, the framework
recovers familiar MinT-style specifications as special
cases, while also allowing additional modeling choices such as covariance shrinkage,  regularization of combination weights, and \textit{Joint}
versus \textit{Separate} estimation.

The comparison with LCC-type benchmarks illustrates the value of moving
from locally induced combinations to the global combination formulation
developed in this paper. LCC-type methods are closest in spirit to a
combination interpretation of reconciliation, but they operate through
local two-level reconciliation and a more restricted set of candidate
forecasts. Our results show
that the broader candidate construction and globally optimized
combination weights in our framework can translate the combination
perspective into stronger empirical performance.

The \textit{Factor} covariance specification is the most consistently
strong covariance choice. It matches the strongest benchmark in
the electricity study and clearly gains in the labor force
study, where all \textit{Factor} variants improve on the strongest
benchmark at most levels. 

Regularizing the combination weights adds value beyond covariance
estimation alone. The gains are not uniform across all specifications,
but they follow a clear pattern: regularization is most useful when the
covariance estimate is relatively weak or noisy, as under
\textit{Raw Sample}, and more modest when the covariance specification
is already stabilized, as under \textit{Factor} or \textit{Diag
Shrink}. This means that 
covariance stabilization and combining weight regularization play
complementary roles, with the combination representation providing a
natural place to integrate them. 

Without combining weight regularization, \textit{Joint} and \textit{Separate} coincide as
implied by the theory. With weight penalization, however, the two strategies
can differ, and the empirical differences favor \textit{Separate} in
most comparisons. The advantage is small in the compact
electricity hierarchy but noticeably larger in the labor force grouped
hierarchy. Thus, \textit{Separate} is not only a computational
simplification; in larger  settings, it can also act as an
effective form of statistical stabilization.

The results remain level-dependent. For example, in the labor force
application, no reconciliation method improves on the unreconciled
\textit{Base} forecast at Level~1, while the proposed methods deliver
their clearest gains at lower and overall levels. The empirical value
of the framework therefore lies not in guaranteeing uniform dominance
at every aggregation level, but in providing a flexible and
interpretable structure for deciding how information is borrowed across
the hierarchy, how covariance estimates are stabilized, and how
combination weights are regularized. Across the two datasets, these
choices produce competitive or leading overall accuracy
while preserving coherence.

The Electronic Companion complements the accuracy comparisons with two
additional analyses of the proposed framework. Section~\ref{app:weight-interpretation}
visualizes the estimated candidate-combination weights, making transparent how
direct, ancestor, and collateral candidate forecasts contribute to the
reconciled bottom-level forecasts. It also examines how egalitarian
regularization changes these weight profiles. Section~\ref{app:computation-time}
compares the computation times of the \textit{Separate} and \textit{Joint} implementations, showing that \textit{Separate} can deliver
substantial efficiency gains, particularly for the larger penalized labor force dataset.

\section{Conclusion}
\label{conclusion}

This paper develops a forecast-combination framework for hierarchical and grouped forecast reconciliation. 
By constructing a maximal set of structurally interpretable candidate forecasts for each bottom-level series, we show that standard linear reconciliation under the standard unbiasedness constraint can be equivalently represented as a forecast combination problem. 
This equivalence provides a new way to understand reconciliation: reconciled forecasts are not only projections onto the coherent forecast space, but also weighted combinations of direct and indirect forecast sources generated by the aggregation structure. 
In this sense, the framework turns reconciliation from an abstract cross-series adjustment into a transparent and interpretable combination problem.

The framework also uncovers a hidden connection between MinT reconciliation and Bates--Granger optimal combination, two methods that have developed in largely separate  literatures. 
We prove that the MSE-optimal combination weights in our framework induce exactly the MinT reconciliation solution, while the optimal weight problem decomposes into independent per-series subproblems with Bates--Granger solutions. 
Thus, MinT can be understood as a collection of optimal forecast-combination problems over hierarchy-induced candidate forecasts. 
This result deepens the theoretical understanding of existing reconciliation methods by showing that a central method in hierarchical forecasting has an implicit optimal-combination structure.

Building on this representation, we propose a unified finite-sample implementation that accommodates covariance-side stabilization and direct regularization of combination weights. 
Existing MinT variants arise as special cases, while new extensions, such as factor-based covariance estimation, egalitarian weight regularization, and series-wise \emph{Separate} estimation, become natural within the same framework. 
Empirical results on electricity generation and labor force data show that the proposed approach is practically implementable, competitive with existing reconciliation methods, and capable of improving accuracy while preserving coherence. 
Overall, the forecast-combination framework provides a new perspective for understanding forecast reconciliation and expands the methodological toolkit for developing new, principled reconciliation procedures.

\clearpage

\bibliographystyle{informs2014}
\bibliography{HF_combing}

\ECSwitch
\EquationsNumberedBySection

\makeatletter
\renewcommand*{\theHsection}{EC.\arabic{section}}
\renewcommand*{\theHsubsection}
  {EC.\arabic{section}.\arabic{subsection}}
\renewcommand*{\theHsubsubsection}
  {EC.\arabic{section}.\arabic{subsection}.\arabic{subsubsection}}
\makeatother

\begin{center}
{\Large\bfseries Electronic Companion to}\\[8pt]

{\Large\bfseries
``A Forecast Combination Framework for Hierarchical and\\
Grouped Time Series Reconciliation''
}
\end{center}

\vspace{1.5em}

\noindent
This electronic companion contains the following sections:

\vspace{0.8em}

\begin{tabularx}{\textwidth}{@{}>{\bfseries}p{1.6cm}X@{}}

\hyperref[sec:ec-proofs]{\ref*{sec:ec-proofs}:}
&
\hyperref[sec:ec-proofs]{\nameref*{sec:ec-proofs}}
\\[4pt]

\hyperref[app:grouped]{\ref*{app:grouped}:}
&
\hyperref[app:grouped]{\nameref*{app:grouped}}
\\[4pt]

\hyperref[sec:unbalanced-structure]
  {\ref*{sec:unbalanced-structure}:}
&
\hyperref[sec:unbalanced-structure]
  {\nameref*{sec:unbalanced-structure}}
\\[4pt]

\hyperref[sec:LCC]{\ref*{sec:LCC}:}
&
\hyperref[sec:LCC]{\nameref*{sec:LCC}}
\\[4pt]

\hyperref[unbalanced-example]{\ref*{unbalanced-example}:}
&
\hyperref[unbalanced-example]{\nameref*{unbalanced-example}}
\\[4pt]

\hyperref[sec:unconstrained]{\ref*{sec:unconstrained}:}
&
\hyperref[sec:unconstrained]{\nameref*{sec:unconstrained}}
\\[4pt]

\hyperref[app:empirical-details]
  {\ref*{app:empirical-details}:}
&
\hyperref[app:empirical-details]
  {\nameref*{app:empirical-details}}
\\

\end{tabularx}

\clearpage

\section{Proofs and Additional Theoretical Results}
\label{sec:ec-proofs}

\subsection{Proof of Lemma~\ref{indep-set}}
\label{app:proof-indep-set}
By standard results on linear systems (see, e.g.,~\cite{strang2022introduction}),
the solution set of $S'\mathbf{c}_i=\mathbf{e}_i$ is an affine space of the form
\[
\mathbf{c}_i=\mathbf{v}_{0}^{(i)}+\mathbf{v}, \qquad 
\mathbf{v}\in\mathcal{N}(S'),
\]
where $\mathbf{v}_{0}^{(i)}$ is a particular solution corresponding to the
direct forecast of bottom-level series $i$.
Since $S$ has full column rank $n_b$, it follows that
$\dim\mathcal{N}(S')=n-n_b=n_a$, and hence $\mathcal{C}_i$ is an affine subspace of
dimension $n_a$.

To characterize $\mathcal{N}(S')$, note that $\mathbf{v}\in\mathcal{N}(S^{\prime})$ if and only if $\mathbf{v}^{\prime}S=\mathbf{0}$, which is equivalent to requiring $\mathbf{v}^{\prime}\mathbf{y}_{t+h}=0$ for every coherent vector $\mathbf{y}_{t+h}=S\mathbf{b}_{t+h}$ (for all $\mathbf{b}_{t+h} \in\mathbb{R}^{n_b}$).  
Thus any null-space direction must encode an aggregation constraint.  
For each aggregated series $j$, define $\mathbf{v}^{(j)}$ as in~\eqref{v_jl_compute}:  
one unit on the aggregated series and $-1$ on each of its bottom-level descendants, and
zeros elsewhere.  Direct verification shows $S'\mathbf{v}^{(j)}=0$, so $\mathbf{v}^{(j)}\in \mathcal{N}(S')$.

Moreover, the vectors $\left\{\mathbf{v}^{(j)}\right\}$ are linearly independent because the $j$-th vector has a nonzero entry at coordinate $j$ (the position of aggregated series $j$ ) while all other vectors in the set have a zero at that coordinate; hence, if $\sum_{j}\alpha_j\mathbf{v}^{(j)}=\mathbf{0}$, inspecting aggregate coordinate $j$ forces $\alpha_j=0$ for every $j$, which establishes linear independence. As their number equals
$\dim\mathcal{N}(S')$, they form a basis of the null space.

Therefore any $\mathbf{v}\in\mathcal{N}(S')$ admits the representation
$\mathbf{v}=\sum_{j=1}^{n_a}\alpha_j\mathbf{v}^{(j)}$, which yields~\eqref{coeffs}.
Since $S^{\prime}\mathbf{v}_0^{(i)}=\mathbf{e}_i\neq\mathbf{0}$ while
$S^{\prime}\mathbf{v}=\mathbf{0}$ on $\mathcal{N}(S^{\prime})$, we have
$\mathbf{v}_0^{(i)}\notin\mathcal{N}(S^{\prime})$; hence the linear span of
$\mathcal{C}_i$ is
$\operatorname{span}\{\mathbf{v}_0^{(i)}\}\oplus\mathcal{N}(S^{\prime})$, of
dimension $n_a+1$. As candidate forecasts are linearly independent if and only if
their coefficient vectors are, no candidate-generation matrix with admissible
rows can have rank exceeding $n_a+1$, a bound attained by the construction in
\Cref{maximality}.

\subsection{Proof of Proposition~\ref{maximality}}
\label{app:proof-maximality}
By Lemma~\ref{indep-set}, the null space $\mathcal{N}(S^{\prime})$ 
has dimension $n_a$ with basis $\{\mathbf{v}^{(j)}\}_{j=1}^{n_a}
\subset\mathbb{R}^n$ indexed by the aggregated series 
$\mathcal{A} = \{\mathrm{Total}, a_1, \ldots, a_{n_a-1}\}$. 
Algorithm~\ref{algorithm-1} constructs the $n_a$ non-trivial rows 
of $C^{(i)}$ by adding null-space vectors to the particular solution 
$\mathbf{v}_0^{(i)}\in\mathbb{R}^n$: the 
ancestor-based rows contribute 
$\{\mathbf{v}^{(a)}\}_{a\in\mathrm{Anc}(i)}$, and the 
collateral-based rows contribute 
$\{\mathbf{v}^{(\mathrm{Total})}-\mathbf{v}^{(k)}\}_{k\in\mathrm{Col}(i)}$.

We show these $n_a$ null-space contributions form a basis of
$\mathcal{N}(S^{\prime})$. Collect them as
\[
\mathcal{V}:=\{\mathbf{v}^{(a)}\}_{a\in\mathrm{Anc}(i)}\cup
\{\mathbf{v}^{(\mathrm{Total})}-\mathbf{v}^{(k)}\}_{k\in\mathrm{Col}(i)},
\]
each element of which lies in $\mathcal{N}(S^{\prime})$. Under the convention
$\mathrm{Total}\in\mathrm{Anc}(i)$, every ancestor vector $\mathbf{v}^{(a)}$, in
particular $\mathbf{v}^{(\mathrm{Total})}$, already belongs to $\mathcal{V}$, and
for each collateral $k\in\mathrm{Col}(i)$,
\[
\mathbf{v}^{(k)}=\mathbf{v}^{(\mathrm{Total})}
-\bigl(\mathbf{v}^{(\mathrm{Total})}-\mathbf{v}^{(k)}\bigr)\in\operatorname{span}(\mathcal{V})
.
\]
Hence $\operatorname{span}(\mathcal{V})$ contains the standard basis
$\{\mathbf{v}^{(j)}\}_{j=1}^{n_a}$, so $\operatorname{span}(\mathcal{V})=\mathcal{N}(S^{\prime
})$;
since $|\mathcal{V}|=n_a=\dim\mathcal{N}(S^{\prime})$, $\mathcal{V}$ is a basis of
$\mathcal{N}(S^{\prime})$.

Finally, we show that the $n_a+1$ rows of $C^{(i)}$ are linearly independent. They are the direct row $(\mathbf{v}_0^{(i)})^{\prime}$, the ancestor-based rows $(\mathbf{v}_0^{(i)}+\mathbf{v}^{(a)})^{\prime}$ for $a\in\mathrm{Anc}(i)$, and the collateral-based rows $(\mathbf{v}_0^{(i)}+\mathbf{v}^{(\mathrm{Total})}-\mathbf{v}^{(k)})^{\prime}$ for $k\in\mathrm{Col}(i)$. Subtracting the direct row from each non-direct row is an invertible row operation, after which the rows become $(\mathbf{v}_0^{(i)})^{\prime}$ together with the $n_a$ null-space contributions $\{\mathbf{v}^{(a)}\}_{a\in\mathrm{Anc}(i)}$ and $\{\mathbf{v}^{(\mathrm{Total})}-\mathbf{v}^{(k)}\}_{k\in\mathrm{Col}(i)}$. Since $S^{\prime}\mathbf{v}_0^{(i)}=\mathbf{e}_i\neq\mathbf{0}$, we have $\mathbf{v}_0^{(i)}\notin\mathcal{N}(S^{\prime})$, while those contributions form a basis of $\mathcal{N}(S^{\prime})$; hence the transformed rows are linearly independent. As row operations preserve rank, the original rows of $C^{(i)}$ are linearly independent, so $\operatorname{rank}(C^{(i)})=n_a+1$.

Moreover, the row space of $C^{(i)}$ equals $\operatorname{span}\{\mathbf{v}_0^{(i)}\}+\mathcal{N}(S^{\prime})$, which contains the admissible set $\mathcal{C}_i=\mathbf{v}_0^{(i)}+\mathcal{N}(S^{\prime})$. Hence every $\mathbf{c}\in\mathcal{C}_i$ satisfies $\mathbf{c}^{\prime}=\boldsymbol{\beta}^{\prime}C^{(i)}$ for some $\boldsymbol{\beta}\in\mathbb{R}^{n_a+1}$, so the corresponding candidate forecast $\mathbf{c}^{\prime}\hat{\mathbf{y}}_{t+h\mid t}=\boldsymbol{\beta}^{\prime}C^{(i)}\hat{\mathbf{y}}_{t+h\mid t}$ is a linear combination of the structured candidates in $\mathcal{F}_{t+h\mid t}^{\mathrm{str}}(i)$. This proves the maximality claim.

\subsection{Proof of Theorem~\ref{Equivalence}: Equivalence of Reconciliation and Combination}
\label{equiv}

\paragraph{Proof of part~(i) (Sufficiency).}
We first show that any $\Phi(\mathbf{w})$ of the form~\eqref{eq:phi},
with weights satisfying $\mathbf{1}_{n_a+1}^\prime \mathbf{w}_i = 1$
for each $i$, together with any candidate-generation matrix $C$
in~\eqref{C}, yields a reconciliation matrix
$P = \Phi(\mathbf{w})C$ satisfying $PS = I_{n_b}$.

By~\eqref{linear-equality}, every row $\mathbf{c}^\prime$ of $C^{(i)}$
satisfies $\mathbf{c}^\prime S = \mathbf{e}_i^\prime$, so
$C^{(i)}S = \mathbf{1}_{n_a+1}\mathbf{e}_i^\prime$ for each $i$,
and therefore $CS = I_{n_b} \otimes \mathbf{1}_{n_a+1}$
by the block structure of $C$ in~\eqref{C}.
Due to the block-diagonal structure of $\Phi(\mathbf{w})$
in~\eqref{eq:phi} and the summing-to-one constraint on the combination
weights, it follows that $\Phi(\mathbf{w})CS = I_{n_b}$, so
$P := \Phi(\mathbf{w})C$ satisfies the unbiasedness condition.

\paragraph{Proof of part~(ii) (Necessity).}
Conversely, suppose $P\in\mathbb{R}^{n_b\times n}$ satisfies $PS = I_{n_b}$,
and let $\mathbf{p}_i^\prime$ denote its $i$-th row. Writing the block
structure of $C$  in~\eqref{C} and $\Phi(\mathbf{w})$ in~\eqref{eq:phi}, the identity $P=\Phi(\mathbf{w})C$ holds row-by-row if
and only if
\begin{equation}
\label{eq:row-system}
\mathbf{w}_i^\prime C^{(i)} = \mathbf{p}_i^\prime,
\qquad i=1,\ldots,n_b,
\end{equation}
so constructing $\Phi(\mathbf{w})$ reduces to solving these $n_b$ decoupled
 linear systems in $\mathbf{w}_i\in\mathbb{R}^{n_a+1}$. 

We first verify that the system has a solution. From $PS=I_{n_b}$, 
$\mathbf{p}_i$ satisfies $S^{\prime}\mathbf{p}_i=\mathbf{e}_i$, so 
$\mathbf{p}_i\in\mathbf{v}_0^{(i)}+\mathcal{N}(S^{\prime})$. The rows of 
$C^{(i)}$ include $(\mathbf{v}_0^{(i)})^{\prime}$ (the particular-solution row) 
and span all of $\mathcal{N}(S^{\prime})$ as additive directions (by 
Proposition~\ref{maximality}'s null-space basis structure), so 
$\mathbf{v}_0^{(i)}+\mathcal{N}(S^{\prime})\subseteq \operatorname{col}(C^{(i)\prime})$, 
giving $\mathbf{p}_i\in \operatorname{col}(C^{(i)\prime})$. Hence~\eqref{eq:row-system} 
admits a solution; uniqueness follows from $C^{(i)\prime}$ having full 
column rank $n_a+1$ (Proposition~\ref{maximality}). Since 
$C^{(i)}C^{(i)\prime}$ is invertible by full-row-rank of $C^{(i)}$, 
left-multiplying $C^{(i)\prime}\mathbf{w}_i=\mathbf{p}_i$ by $C^{(i)}$ 
gives the closed form
\begin{equation}
\label{eq:w_i-block}
\mathbf{w}_i = \bigl(C^{(i)}C^{(i)\prime}\bigr)^{-1}\,C^{(i)}\mathbf{p}_i.
\end{equation}

It remains to verify $\mathbf{1}_{n_a+1}^\prime\mathbf{w}_i = 1$.
By~\eqref{eq:row-system}, $\mathbf{w}_i^\prime C^{(i)} = \mathbf{p}_i^\prime$;
right-multiplying by $S$ and using $\mathbf{p}_i^\prime S = \mathbf{e}_i^\prime$
(from $PS = I_{n_b}$) gives
\[
\mathbf{w}_i^\prime C^{(i)} S = \mathbf{p}_i^\prime S = \mathbf{e}_i^\prime.
\]
Applying~\eqref{linear-equality} to each row of $C^{(i)}$ gives $C^{(i)} S = \mathbf{1}_{n_a+1}\mathbf{e}_i^\prime$,
so the left-hand side equals
$(\mathbf{1}_{n_a+1}^\prime\mathbf{w}_i)\,\mathbf{e}_i^\prime$.
Since $\mathbf{e}_i^\prime \neq \mathbf{0}$, we conclude
$\mathbf{1}_{n_a+1}^\prime\mathbf{w}_i = 1$.
Stacking all $\mathbf{w}_i^\prime$ in~\eqref{eq:w_i-block} into a
block-diagonal matrix produces a $\Phi(\mathbf{w})$ of the
form~\eqref{eq:phi}, ensuring $P = \Phi(\mathbf{w})C$.

\subsection{Proof of Theorem~\ref{quadratic-theorem}: Quadratic Programming Formulation}
\label{quadratic-reform}
We show that the objective function in~\eqref{obj1} admits an equivalent
quadratic form in the combination weights.
From~\eqref{obj1} and~\eqref{C}, substituting
$\hat{\mathbf{y}}_{\mathrm{cand},t+h|t} = C\hat{\mathbf{y}}_{t+h|t}$,
the optimization problem can be written as
\begin{equation}
\label{obj2}
\min_{\mathbf{w}}\;
\mathbb{E}[
\|
\mathbf{y}_{t+h}
- S\,\Phi(\mathbf{w})\,\hat{\mathbf{y}}_{\mathrm{cand},t+h|t}\|_2^2
\;\big|\;\mathcal{I}_t].
\end{equation}
Recalling that $\mathbf{y}_{t+h} = S\mathbf{b}_{t+h}$
from~\eqref{coherence}, define the vector of combined $h$-step-ahead
forecast errors at the bottom level as
$
\tilde{\mathbf{e}}_{t+h|t}
:= \mathbf{b}_{t+h} - \Phi(\mathbf{w})\hat{\mathbf{y}}_{\mathrm{cand},t+h|t}
\in \mathbb{R}^{n_b}.
$
Since $\mathbf{y}_{t+h} - S\Phi(\mathbf{w})\hat{\mathbf{y}}_{\mathrm{cand},t+h|t}
= S\tilde{\mathbf{e}}_{t+h|t}$,
applying the identity $\|M \mathbf{x}\|_2^2 = \mathbf{x}^{\prime}M^{\prime}M\mathbf{x}$ with
$M = S$ and $\mathbf{x} = \tilde{\mathbf{e}}_{t+h|t}$ yields
\begin{equation}
\label{obj3}
\mathbb{E}[
\|S\tilde{\mathbf{e}}_{t+h|t}\|_2^2
\,|\,\mathcal{I}_t]
=
\mathbb{E}[
\sum_{i=1}^{n_b}\sum_{i^{\prime}=1}^{n_b}
z_{i,i^{\prime}}\,
\tilde{e}_{t+h|t}^{(i)}
\tilde{e}_{t+h|t}^{(i^{\prime})}
\,|\,\mathcal{I}_t].
\end{equation}
where $z_{i,i^{\prime}}$ denotes the $(i,i^{\prime})$-th entry of $S^{\prime}S$.

For each bottom-level series $i$, write the combined bottom-level error using the
candidate-error vector
$\hat{\boldsymbol{\xi}}_{t+h|t}^{(i)}
:= b_{t+h}^{(i)}\mathbf{1}_{n_a+1} - \hat{\mathbf{y}}_{\mathrm{cand},t+h|t}^{(i)}
\in \mathbb{R}^{n_a+1}$.
Because the combination weights sum to one,
$\mathbf{1}_{n_a+1}^{\prime}\mathbf{w}_i = 1$, we have
$b_{t+h}^{(i)} = \mathbf{w}_i^{\prime}\bigl(b_{t+h}^{(i)}\mathbf{1}_{n_a+1}\bigr)$,
so the combined error is a linear combination of the candidate errors:
\begin{equation}
\label{error}
\begin{aligned}
    \tilde{e}_{t+h|t}^{(i)}
&= b_{t+h}^{(i)} - \mathbf{w}_i^{\prime}\hat{\mathbf{y}}_{\mathrm{cand},t+h|t}^{(i)}
\\&= \mathbf{w}_i^{\prime}\bigl(b_{t+h}^{(i)}\mathbf{1}_{n_a+1}
   - \hat{\mathbf{y}}_{\mathrm{cand},t+h|t}^{(i)}\bigr)
= \mathbf{w}_i^{\prime}\hat{\boldsymbol{\xi}}_{t+h|t}^{(i)}.
\end{aligned}
\end{equation}
By admissibility, $C^{(i)}S = \mathbf{1}_{n_a+1}\mathbf{e}_i^{\prime}$, so
$C^{(i)}\mathbf{y}_{t+h} = C^{(i)}S\mathbf{b}_{t+h} = b_{t+h}^{(i)}\mathbf{1}_{n_a+1}$;
together with $\hat{\mathbf{y}}_{\mathrm{cand},t+h|t}^{(i)} = C^{(i)}\hat{\mathbf{y}}_{t+h|t}$
from~\eqref{eq:candidates}, the candidate error in~\eqref{error} satisfies
\[
\hat{\boldsymbol{\xi}}_{t+h|t}^{(i)}
= C^{(i)}\bigl(\mathbf{y}_{t+h} - \hat{\mathbf{y}}_{t+h|t}\bigr)
= C^{(i)}\hat{\mathbf{e}}_{t+h|t}.
\]
Since the base forecast errors are conditionally unbiased,
$\mathbb{E}[\hat{\mathbf{e}}_{t+h|t}\mid\mathcal{I}_t]=\mathbf{0}$, it follows that
$\mathbb{E}[\hat{\boldsymbol{\xi}}_{t+h|t}^{(i)}\mid\mathcal{I}_t]=\mathbf{0}$.
Substituting~\eqref{error} into~\eqref{obj3} and taking expectations yields
\begin{equation}
\label{obj4}
\min_{\mathbf{w}_1,\ldots,\mathbf{w}_{n_b}}
\sum_{i=1}^{n_b}\sum_{i^{\prime}=1}^{n_b}
\mathbf{w}_i^{\prime}
[Q_h]_{i,i^{\prime}}
\mathbf{w}_{i^{\prime}},
\end{equation}
where, using $\hat{\boldsymbol{\xi}}_{t+h|t}^{(i)} = C^{(i)}\hat{\mathbf{e}}_{t+h|t}$ and
$\mathrm{Var}(\hat{\mathbf{e}}_{t+h|t}\mid\mathcal{I}_t)=\Sigma_h$,
\begin{align*}
[Q_h]_{i,i^{\prime}}
=& z_{i,i^{\prime}}\,\operatorname{Cov}\!\bigl(\hat{\boldsymbol{\xi}}_{t+h|t}^{(i)},\,\hat{\boldsymbol{\xi}}_{t+h|t}^{(i^{\prime})}\mid\mathcal{I}_t\bigr)
\\=& z_{i,i^{\prime}}\,C^{(i)}\Sigma_h C^{(i^{\prime})\prime}
\in \mathbb{R}^{(n_a+1)\times(n_a+1)},
\end{align*}
which coincides with the block form in Theorem~\ref{quadratic-theorem}.
The objective~\eqref{obj4} defines a convex quadratic form in the
combination weights, which leads directly to the quadratic program
in~\eqref{quadratic-formulation}.
\subsection{Proof of Proposition~\ref{closed_solution}: Closed-Form Combination Weights}
\label{closed-solution}
Consider the constrained quadratic program
\[
\min_{\mathbf{w}}\ \mathbf{w}^{\prime} Q_h \mathbf{w}
\quad\text{s.t.}\quad \Pi\mathbf{w}=\mathbf{1}_{n_b}.
\]

\emph{Positive definiteness of $Q_h$ and $\Pi Q_h^{-1}\Pi'$.}
Let $\bar{C}=\mathrm{blkdiag}(C^{(1)},\dots,C^{(n_b)})$. The block 
definition of $Q_h$ in Theorem~\ref{quadratic-theorem}, combined with 
the Kronecker product structure of $S'S\otimes\Sigma_h$ and block-wise 
matrix multiplication, yields
\begin{equation}
Q_h=\bar{C}\,(S'S\otimes\Sigma_h)\,\bar{C}'.
\label{eq:Qh-sandwich}
\end{equation}
The standing error conditions give $\Sigma_h\succ 0$, and $S$ has full 
column rank, so $S'S\succ 0$; hence $S'S\otimes\Sigma_h\succ 0$ 
(the Kronecker product of two positive-definite matrices is positive 
definite). Proposition~\ref{maximality} gives that each $C^{(i)}$ has 
full row rank, and since the rank of a block-diagonal matrix equals 
the sum of its block ranks, $\mathrm{rank}(\bar{C})=n_b(n_a+1)$, so 
$\bar{C}$ has full row rank. Applying the sandwich rule for 
positive-definite matrices\footnote{If $M\succ 0$ and $N$ has full row 
rank, then $NMN'\succ 0$.} to~\eqref{eq:Qh-sandwich} yields 
$Q_h\succ 0$. The constraint matrix $\Pi$ has full row rank by 
construction, so another sandwich application gives 
$\Pi Q_h^{-1}\Pi'\succ 0$. In particular, $Q_h^{-1}$ and $(\Pi Q_h^{-1}\Pi')^{-1}$ 
both exist.

\emph{Closed-form solution.} 
With $Q_h\succ 0$, the objective is strictly convex, so the quadratic 
program admits a unique minimizer characterized by its first-order 
conditions. The Lagrangian is
\begin{equation}
\label{lagr-m}
\mathcal{L}(\mathbf{w}, \boldsymbol{\eta}) 
= \mathbf{w}^{\prime}Q_h\mathbf{w} 
- \boldsymbol{\eta}^{\prime}(\Pi \mathbf{w}-\mathbf{1}_{n_b}),
\end{equation}
with $\boldsymbol{\eta}$ the Lagrange multiplier vector. The first-order 
conditions are
\begin{equation}
\label{ktt-conds}
2Q_h\mathbf{w} = {\Pi}^{\prime}\boldsymbol{\eta}, 
\qquad 
\Pi \mathbf{w} =\mathbf{1}_{n_b}.
\end{equation}
From the first equation, $\mathbf{w}=\tfrac12 Q_h^{-1}\Pi'\boldsymbol{\eta}$. 
Substituting into the second and solving for $\boldsymbol{\eta}$,
\[
\Pi Q_h^{-1}{\Pi}^{\prime}\boldsymbol{\eta} = 2\mathbf{1}_{n_b}
\quad\Longrightarrow\quad
\boldsymbol{\eta} = 2\,(\Pi Q_h^{-1}{\Pi}^{\prime})^{-1}\mathbf{1}_{n_b}.
\]
Substituting back yields the optimal closed-form solution
\begin{equation}
\label{w}
\mathbf{w}^\ast = Q_h^{-1}{\Pi}^{\prime}(\Pi Q_h^{-1}{\Pi}^{\prime})^{-1}\mathbf{1}_{n_b}.
\end{equation}

\subsection{Proof of Proposition~\ref{prop:equivalence}, Part I:
Bates--Granger Formula under Block-Diagonal $Q_h^{\mathrm{bd}}$}
\label{app:proof-sep}

As established in the main text, under the block-diagonal specification
$Q_h^{\mathrm{bd}}$, problem~\eqref{quadratic-formulation-bd}
decomposes into the $n_b$ independent subproblems
in~\eqref{eq:subproblem}. For each $i=1,\ldots,n_b$, the KKT conditions
for the $i$-th subproblem are
\[
2[Q_h]_{i,i}\mathbf{w}_i
=
\eta_i\mathbf{1}_{n_a+1},
\qquad
\mathbf{1}_{n_a+1}^{\prime}\mathbf{w}_i=1,
\]
where $\eta_i$ is the Lagrange multiplier. Since $[Q_h]_{i,i}$ is
positive definite, these conditions are necessary and sufficient and
imply
\begin{equation}
\label{eq:sep-block-solution}
\mathbf{w}_i^{\mathrm{sep}}
=
\frac{
[Q_h]_{i,i}^{-1}\mathbf{1}_{n_a+1}
}{
\mathbf{1}_{n_a+1}^{\prime}
[Q_h]_{i,i}^{-1}
\mathbf{1}_{n_a+1}
}.
\end{equation}

Using
$[Q_h]_{i,i}
=
z_{i,i}C^{(i)}\Sigma_h C^{(i)\prime}$,
with $z_{i,i}>0$, the scalar factor $z_{i,i}^{-1}$ cancels from the
numerator and denominator of~\eqref{eq:sep-block-solution}. Therefore,
\[
\mathbf{w}_i^{\mathrm{sep}}
=
\frac{
\bigl(C^{(i)}\Sigma_h C^{(i)\prime}\bigr)^{-1}
\mathbf{1}_{n_a+1}
}{
\mathbf{1}_{n_a+1}^{\prime}
\bigl(C^{(i)}\Sigma_h C^{(i)\prime}\bigr)^{-1}
\mathbf{1}_{n_a+1}
},
\]
which is exactly the Bates--Granger formula in~\eqref{eq:BG}.

\subsection{Proof of Proposition~\ref{prop:equivalence}, Part II: Equivalence of Joint and Separate Optima via KKT}
\label{app:proof}

The proof proceeds via two steps: we first establish a key lemma on the 
proportionality of cross-covariance terms, and then use it to verify that 
the block-diagonal solution $\mathbf{w}^{\mathrm{sep}}$ satisfies the KKT 
conditions of the full problem.

\begin{mylemma}[Proportionality of cross-covariance terms]
\label{lem:proportionality}
For any fixed $i' \in \{1, \ldots, n_b\}$ and any $i \in \{1, \ldots, n_b\}$,
\[
C^{(i)}\Sigma_h C^{(i^{\prime})\prime}
\bigl(C^{(i^{\prime})}\Sigma_h C^{(i^{\prime})\prime}\bigr)^{-1}
\mathbf{1}_{n_a+1}
= \frac{k_i^{(i')}}{k_{i^{\prime}}^{(i')}}\,\mathbf{1}_{n_a+1},
\]
where $k_1^{(i')}, \ldots, k_{n_b}^{(i')}$ are scalars depending on
both $S$ and $\Sigma_h$, with $k_{i^{\prime}}^{(i')} \neq 0$. In
particular, the left-hand side is always proportional to
$\mathbf{1}_{n_a+1}$.
\end{mylemma}

\paragraph{Geometric meaning.}
The vector $(C^{(i')}\Sigma_h C^{(i')'})^{-1}\mathbf{1}_{n_a+1}$ is 
(up to normalization) the Bates--Granger optimal combination weight 
for series $i'$, i.e., $\mathbf{w}_{i'}^{\mathrm{sep}}$ 
in~\eqref{eq:BG}. Lemma~\ref{lem:proportionality} therefore states 
a specific geometric fact: the off-diagonal cross-covariance block 
$C^{(i)}\Sigma_h C^{(i')'}$ maps the Bates--Granger direction 
of series $i'$ to the $\mathbf{1}_{n_a+1}$ direction. This mapping 
is not arbitrary, as we show in the proof of 
Proposition~\ref{prop:equivalence}, it is precisely what allows 
the \emph{Joint} and \emph{Separate} optima to coincide: the off-diagonal terms 
in the \emph{Joint} first-order conditions, once contracted with the 
\emph{Separate} solution, align with the Lagrange multiplier direction and 
are absorbed without affecting the optimum.
\begin{myproof}
Fix $i' \in \{1, \ldots, n_b\}$ throughout; the scalars
$k_1^{(i')}, \ldots, k_{n_b}^{(i')}$ constructed below all depend on
this choice of $i'$.

\emph{Proof strategy: denote the target vector by 
$\boldsymbol{\nu}_{i'}$; we show that $\boldsymbol{\nu}_{i'}$ lies 
in $\operatorname{col}(\tilde{S})$ after reparameterization, so that 
applying $\tilde{C}^{(i)}$ reduces to extracting the $i$-th column 
of $\tilde{C}^{(i)}\tilde{S}$, which has a sparse structure by 
construction of $C^{(i)}$.}

\medskip\noindent\textit{Absorbing $\Sigma_h$ via reparameterization.}
Since $\Sigma_h \succ 0$, it admits a symmetric positive
definite square root $\Sigma_h^{1/2}$ satisfying
$\Sigma_h^{1/2}\Sigma_h^{1/2} = \Sigma_h$, with inverse $\Sigma_h^{-1/2}$.
Define
$\tilde{C}^{(i)} := C^{(i)}\Sigma_h^{1/2}$ and
$\tilde{S} := \Sigma_h^{-1/2}S$. Then
$C^{(i)}\Sigma_h C^{(i^{\prime})\prime}
= \tilde{C}^{(i)}\tilde{C}^{(i^{\prime})\prime}$
and
$C^{(i^{\prime})}\Sigma_h C^{(i^{\prime})\prime}
= \tilde{C}^{(i^{\prime})}\tilde{C}^{(i^{\prime})\prime}$,
so the left-hand side of the lemma equals
$\tilde{C}^{(i)}\tilde{C}^{(i^{\prime})\prime}
\bigl(\tilde{C}^{(i^{\prime})}\tilde{C}^{(i^{\prime})\prime}\bigr)^{-1}
\mathbf{1}_{n_a+1}$.
The structural property
$\tilde{C}^{(i)}\tilde{S}
= C^{(i)}\Sigma_h^{1/2}\Sigma_h^{-1/2}S
= C^{(i)}S
= \mathbf{1}_{n_a+1}\mathbf{e}_i^{\prime}$
is preserved under the reparameterization
(using~\eqref{linear-equality} applied row-wise to $C^{(i)}$ and
stacked). It therefore suffices to prove the lemma for $\tilde{C}^{(i)}$
and $\tilde{S}$.

\medskip\noindent\textit{Step 1: Reformulation.}
\emph{Strategy: the target expression equals 
$\tilde{C}^{(i)}\boldsymbol{\nu}_{i'}$ for a specific vector 
$\boldsymbol{\nu}_{i'}$ defined below. We characterize 
$\boldsymbol{\nu}_{i'}$ as the unique vector in 
$\operatorname{col}(\tilde{C}^{(i')\prime})$ that $\tilde{C}^{(i')}$ maps 
to $\mathbf{1}_{n_a+1}$.}

Define
$
\boldsymbol{\nu}_{i^{\prime}} :=
\tilde{C}^{(i^{\prime})\prime}
\bigl(\tilde{C}^{(i^{\prime})}\tilde{C}^{(i^{\prime})\prime}\bigr)^{-1}
\mathbf{1}_{n_a+1}.
$
Since $\boldsymbol{\nu}_{i^{\prime}}$ is a linear combination of columns of
$\tilde{C}^{(i^{\prime})\prime}$, it lies in
$\operatorname{col}(\tilde{C}^{(i^{\prime})\prime})$.
Premultiplying by $\tilde{C}^{(i^{\prime})}$ gives
$
\tilde{C}^{(i^{\prime})}\boldsymbol{\nu}_{i^{\prime}}
= \tilde{C}^{(i^{\prime})}\tilde{C}^{(i^{\prime})\prime}
  \bigl(\tilde{C}^{(i^{\prime})}\tilde{C}^{(i^{\prime})\prime}\bigr)^{-1}
  \mathbf{1}_{n_a+1}
= \mathbf{1}_{n_a+1}.
$
Since $\tilde{C}^{(i^{\prime})}\tilde{C}^{(i^{\prime})\prime}$ is invertible
by the full row rank of $\tilde{C}^{(i^{\prime})}$,
$\boldsymbol{\nu}_{i^{\prime}}$ is the unique vector in
$\operatorname{col}(\tilde{C}^{(i^{\prime})\prime})$ satisfying
$\tilde{C}^{(i^{\prime})}\boldsymbol{\nu}_{i^{\prime}} = \mathbf{1}_{n_a+1}$.
The target expression after reparameterization equals
$\tilde{C}^{(i)}\boldsymbol{\nu}_{i^{\prime}}$, so it suffices
to show
$\tilde{C}^{(i)}\boldsymbol{\nu}_{i^{\prime}}
= (k_i^{(i')}/k_{i^{\prime}}^{(i')})\mathbf{1}_{n_a+1}$.

\medskip\noindent\textit{Step 2: Existence of a nonzero vector in
$\operatorname{col}(\tilde{C}^{(i^{\prime})\prime}) \cap
\operatorname{col}(\tilde{S})$.}
\emph{Strategy: we seek a vector that is simultaneously in 
$\operatorname{col}(\tilde{C}^{(i')\prime})$ (where $\boldsymbol{\nu}_{i'}$ 
lives) and in $\operatorname{col}(\tilde{S})$ (where the sparse structure 
$\tilde{C}^{(i)}\tilde{\mathbf{s}}_{i''} = \mathbf{0}$ for $i'' \neq i$ 
can be exploited). Dimension counting establishes a non-trivial 
intersection.}

The subspaces $\operatorname{col}(\tilde{C}^{(i^{\prime})\prime})$ and
$\operatorname{col}(\tilde{S})$ of $\mathbb{R}^n$ have dimensions $n_a+1$
and $n_b$, respectively, with $n = n_a + n_b$. Their dimensions sum
to $n+1 > n$, so
$
\dim\bigl(\operatorname{col}(\tilde{C}^{(i^{\prime})\prime})
\cap \operatorname{col}(\tilde{S})\bigr)
\geq (n_a+1) + n_b - n = 1,
$
and there exists a nonzero
$\boldsymbol{\zeta} \in \operatorname{col}(\tilde{C}^{(i^{\prime})\prime})
\cap \operatorname{col}(\tilde{S})$. Since
$\boldsymbol{\zeta} \in \operatorname{col}(\tilde{S})$, write
\begin{equation}
\label{eq:zeta-decomp}
\boldsymbol{\zeta}
= \sum_{i^{\prime \prime}=1}^{n_b} k_{i^{\prime \prime}}^{(i')} \tilde{\mathbf{s}}_{i^{\prime \prime}},
\end{equation}
where $\tilde{\mathbf{s}}_{i^{\prime \prime}}$ denotes the ${i^{\prime \prime}}$-th column of
$\tilde{S}$ and $k_1^{(i')}, \ldots, k_{n_b}^{(i')}$ are real
scalars (depending on $\boldsymbol{\zeta}$, hence on $i'$).

\medskip\noindent\textit{Step 3: Showing
$\boldsymbol{\nu}_{i^{\prime}}
= \boldsymbol{\zeta}/k_{i^{\prime}}^{(i')}
\in \operatorname{col}(\tilde{S})$.}
\emph{Strategy: we compute $\tilde{C}^{(i')}\boldsymbol{\zeta}$ 
explicitly; the sparse structure of $\tilde{C}^{(i')}\tilde{S}$ 
picks out only the $i'$-th term, giving $k_{i'}^{(i')}\mathbf{1}_{n_a+1}$. 
Rescaling $\boldsymbol{\zeta}$ by $k_{i'}^{(i')}$ yields 
$\tilde{C}^{(i')}(\boldsymbol{\zeta}/k_{i'}^{(i')}) = \mathbf{1}_{n_a+1}$, 
and uniqueness from Step 1 forces 
$\boldsymbol{\nu}_{i'} = \boldsymbol{\zeta}/k_{i'}^{(i')}$---thereby 
placing $\boldsymbol{\nu}_{i'}$ in $\operatorname{col}(\tilde{S})$.}

From~\eqref{linear-equality},
$\tilde{C}^{(i^{\prime})}\tilde{S}
= \mathbf{1}_{n_a+1}\mathbf{e}_{i^{\prime}}^{\prime}$, so
$\tilde{C}^{(i^{\prime})}\tilde{\mathbf{s}}_{i^{\prime}}
= \mathbf{1}_{n_a+1}$ and
$\tilde{C}^{(i^{\prime})}\tilde{\mathbf{s}}_{i^{\prime \prime}} = \mathbf{0}_{n_a+1}$
for all ${i^{\prime \prime}} \neq i^{\prime}$. Applying $\tilde{C}^{(i^{\prime})}$
to~\eqref{eq:zeta-decomp} gives
$
\tilde{C}^{(i^{\prime})}\boldsymbol{\zeta}
= \sum_{{i^{\prime \prime}}=1}^{n_b} k_{i^{\prime \prime}}^{(i')}\,
  \tilde{C}^{(i^{\prime})}\tilde{\mathbf{s}}_{i^{\prime \prime}}
= k_{i^{\prime}}^{(i')}\,\mathbf{1}_{n_a+1}.
$
We first show $k_{i^{\prime}}^{(i')} \neq 0$. If
$k_{i^{\prime}}^{(i')} = 0$, the display above gives
$\tilde{C}^{(i^{\prime})}\boldsymbol{\zeta} = \mathbf{0}$. Writing
$\boldsymbol{\zeta} = \tilde{C}^{(i^{\prime})\prime}\mathbf{x}$ (since
$\boldsymbol{\zeta} \in \operatorname{col}(\tilde{C}^{(i^{\prime})\prime})$),
we obtain
$\tilde{C}^{(i^{\prime})}\tilde{C}^{(i^{\prime})\prime}\mathbf{x}
= \mathbf{0}$, and invertibility of
$\tilde{C}^{(i^{\prime})}\tilde{C}^{(i^{\prime})\prime}$ forces
$\mathbf{x} = \mathbf{0}$, hence
$\boldsymbol{\zeta} = \mathbf{0}$, contradicting
$\boldsymbol{\zeta} \neq \mathbf{0}$. Thus
$k_{i^{\prime}}^{(i')} \neq 0$, and
$\boldsymbol{\zeta}/k_{i^{\prime}}^{(i')}$ is well-defined.

It remains to show
$\boldsymbol{\nu}_{i^{\prime}}
= \boldsymbol{\zeta}/k_{i^{\prime}}^{(i')}$. Both vectors lie in
$\operatorname{col}(\tilde{C}^{(i^{\prime})\prime})$
($\boldsymbol{\nu}_{i^{\prime}}$ by Step~1 and
$\boldsymbol{\zeta}/k_{i^{\prime}}^{(i')}$ by construction), and
$
\tilde{C}^{(i^{\prime})}
\bigl(\boldsymbol{\nu}_{i^{\prime}}
- \boldsymbol{\zeta}/k_{i^{\prime}}^{(i')}\bigr)
= \mathbf{1}_{n_a+1} - \mathbf{1}_{n_a+1}
= \mathbf{0}.
$
Writing
$\boldsymbol{\nu}_{i^{\prime}}
- \boldsymbol{\zeta}/k_{i^{\prime}}^{(i')}
= \tilde{C}^{(i^{\prime})\prime}\mathbf{z}$ for some $\mathbf{z}$,
applying $\tilde{C}^{(i^{\prime})}$ gives
$\tilde{C}^{(i^{\prime})}\tilde{C}^{(i^{\prime})\prime}\mathbf{z}
= \mathbf{0}$; invertibility forces $\mathbf{z} = \mathbf{0}$, so
$\boldsymbol{\nu}_{i^{\prime}}
= \boldsymbol{\zeta}/k_{i^{\prime}}^{(i')}$. Combining
with~\eqref{eq:zeta-decomp},
\begin{equation}
\label{eq:nu-decomp}
\boldsymbol{\nu}_{i^{\prime}}
= \frac{\boldsymbol{\zeta}}{k_{i^{\prime}}^{(i')}}
= \sum_{{i^{\prime \prime}}=1}^{n_b}
  \frac{k_{i^{\prime \prime}}^{(i')}}{k_{i^{\prime}}^{(i')}}\tilde{\mathbf{s}}_{i^{\prime \prime}}
\in \operatorname{col}(\tilde{S}).
\end{equation}

\medskip\noindent\textit{Step 4: Computing
$\tilde{C}^{(i)}\boldsymbol{\nu}_{i^{\prime}}$.}
\emph{Strategy: having placed $\boldsymbol{\nu}_{i'}$ in 
$\operatorname{col}(\tilde{S})$ via~\eqref{eq:nu-decomp}, applying 
$\tilde{C}^{(i)}$ reduces to the sparse structure 
$\tilde{C}^{(i)}\tilde{\mathbf{s}}_{i''} = \mathbf{1}_{n_a+1}$ if 
$i'' = i$ and $\mathbf{0}$ otherwise; only the $i''=i$ term in the 
sum survives, giving the claimed scalar 
$k_i^{(i')}/k_{i'}^{(i')}$.}

From~\eqref{linear-equality},
$\tilde{C}^{(i)}\tilde{S}
= \mathbf{1}_{n_a+1}\mathbf{e}_i^{\prime}$, so
$\tilde{C}^{(i)}\tilde{\mathbf{s}}_i = \mathbf{1}_{n_a+1}$
and $\tilde{C}^{(i)}\tilde{\mathbf{s}}_{i^{\prime \prime}} = \mathbf{0}_{n_a+1}$
for all ${i^{\prime \prime}} \neq i$. Applying $\tilde{C}^{(i)}$
to~\eqref{eq:nu-decomp}, only the ${i^{\prime \prime}} = i$ term survives:
\[
\tilde{C}^{(i)}\boldsymbol{\nu}_{i^{\prime}}
= \sum_{{i^{\prime \prime}}=1}^{n_b}
  \frac{k_{i^{\prime \prime}}^{(i')}}{k_{i^{\prime}}^{(i')}}\,
  \tilde{C}^{(i)}\tilde{\mathbf{s}}_{i^{\prime \prime}}
= \frac{k_i^{(i')}}{k_{i^{\prime}}^{(i')}}\,\mathbf{1}_{n_a+1}.
\]

\medskip\noindent\textit{Geometric takeaway.}
The vector 
$\boldsymbol{\nu}_{i'} = \tilde{C}^{(i')\prime}
(\tilde{C}^{(i')}\tilde{C}^{(i')\prime})^{-1}\mathbf{1}_{n_a+1}$
is the $n$-dimensional coefficient representation obtained by applying
$\tilde{C}^{(i')\prime}$ to the unnormalised Bates--Granger direction
for series $i'$. The lemma shows that
$\tilde{C}^{(i)}\boldsymbol{\nu}_{i'}$ is proportional to
$\mathbf{1}_{n_a+1}$. Equivalently, since
$\tilde{C}^{(i)}\tilde{C}^{(i')\prime}
=
C^{(i)}\Sigma_h C^{(i')\prime}$,
the cross-covariance block maps the unnormalised Bates--Granger
direction for series $i'$ to a vector in the
$\mathbf{1}_{n_a+1}$ direction. This alignment is the key property
behind the separability result in Proposition~\ref{prop:equivalence}:
it allows the off-diagonal contributions in the \emph{Joint} first-order
conditions to be absorbed into the Lagrange multiplier without
changing the Separate optimum. 

\end{myproof}

\begin{myproof}
The proof strategy is to show that $\mathbf{w}^{\mathrm{sep}}$
satisfies the KKT conditions of the full program~\eqref{quadratic-formulation};
since~\eqref{quadratic-formulation} is strictly convex, its KKT conditions
are both necessary and sufficient, which then forces
$\mathbf{w}^\ast = \mathbf{w}^{\mathrm{sep}}$.

\medskip\noindent\textit{Step 1: KKT conditions of the full program.}
The KKT conditions~\eqref{ktt-conds} hold if and only if every block of
$Q_h\mathbf{w}$ is proportional to $\mathbf{1}_{n_a+1}$, since
${\Pi}^{\prime} = \mathrm{blkdiag}(\mathbf{1}_{n_a+1}, \ldots,
\mathbf{1}_{n_a+1})$.
It therefore suffices to verify this proportionality for
$\mathbf{w} = \mathbf{w}^{\mathrm{sep}}$.

\medskip\noindent\textit{Step 2: KKT conditions of the block-diagonal program.}
Since $\mathbf{w}^{\mathrm{sep}}$ solves the block-diagonal program,
its $i^{\prime}$-th block satisfies the KKT condition
\begin{equation}
\label{eq:KKT-block}
C^{(i^{\prime})}\Sigma_h C^{(i^{\prime})\prime}
\mathbf{w}_{i^{\prime}}^{\mathrm{sep}}
= \frac{\tilde{\eta}_{i^{\prime}}}{2z_{i^{\prime},i^{\prime}}}
  \mathbf{1}_{n_a+1},
\qquad i^{\prime} = 1, \ldots, n_b,
\end{equation}
where $\tilde{\eta}_{i^{\prime}}$ is the Lagrange multiplier for series
$i^{\prime}$ in the block-diagonal program.
Rearranging~\eqref{eq:KKT-block} gives
\begin{equation}
\label{eq:w-sep-explicit}
\mathbf{w}_{i^{\prime}}^{\mathrm{sep}}
= \frac{\tilde{\eta}_{i^{\prime}}}{2z_{i^{\prime},i^{\prime}}}
  \bigl(C^{(i^{\prime})}\Sigma_h C^{(i^{\prime})\prime}\bigr)^{-1}
  \mathbf{1}_{n_a+1}.
\end{equation}

The $i$-th block $\bigl[Q_h\mathbf{w}^{\mathrm{sep}}\bigr]_{i}
\in \mathbb{R}^{n_a+1}$ expands as
\begin{equation}
\label{eq:block-expand}
\begin{aligned}
\bigl[Q_h\mathbf{w}^{\mathrm{sep}}\bigr]_{i}
&= z_{i,i}\,C^{(i)}\Sigma_h C^{(i)\prime}\mathbf{w}_i^{\mathrm{sep}}
\\&\quad + \sum_{i^{\prime} \neq i}
  z_{i,i^{\prime}}\,C^{(i)}\Sigma_h C^{(i^{\prime})\prime}
  \mathbf{w}_{i^{\prime}}^{\mathrm{sep}}.    
\end{aligned}
\end{equation}
For the diagonal term in~\eqref{eq:block-expand},
substituting~\eqref{eq:KKT-block} with $i^{\prime} = i$ gives
$z_{i,i}\,C^{(i)}\Sigma_h C^{(i)\prime}\mathbf{w}_i^{\mathrm{sep}}
= \frac{\tilde{\eta}_i}{2}\,\mathbf{1}_{n_a+1}$.
For each off-diagonal term in~\eqref{eq:block-expand},
substituting~\eqref{eq:w-sep-explicit} gives
\begin{align*}
&z_{i,i^{\prime}}\,C^{(i)}\Sigma_h C^{(i^{\prime})\prime}
\mathbf{w}_{i^{\prime}}^{\mathrm{sep}}
\\&\quad = \frac{z_{i,i^{\prime}}\,\tilde{\eta}_{i^{\prime}}}{2z_{i^{\prime},i^{\prime}}}
\underbrace{
  C^{(i)}\Sigma_h C^{(i^{\prime})\prime}
  \bigl(C^{(i^{\prime})}\Sigma_h C^{(i^{\prime})\prime}\bigr)^{-1}
  \mathbf{1}_{n_a+1}
}_{\displaystyle=\,\frac{k_i^{(i')}}{k_{i^{\prime}}^{(i')}}\mathbf{1}_{n_a+1}
  \text{ by Lemma~\ref{lem:proportionality}}},
\end{align*}
which is proportional to $\mathbf{1}_{n_a+1}$.

\medskip\noindent\textit{Step 3: Conclusion.}
Summing the diagonal and off-diagonal contributions,
every block $\bigl[Q_h\mathbf{w}^{\mathrm{sep}}\bigr]_{i}$ is proportional
to $\mathbf{1}_{n_a+1}$, so $2Q_h\mathbf{w}^{\mathrm{sep}} = {\Pi}^{\prime}\boldsymbol{\eta}$
holds for some $\boldsymbol{\eta} \in \mathbb{R}^{n_b}$.
The constraint $\Pi \mathbf{w}^{\mathrm{sep}} = \mathbf{1}_{n_b}$ holds by
construction of the block-diagonal program.
Hence $\mathbf{w}^{\mathrm{sep}}$ satisfies both KKT conditions
of~\eqref{ktt-conds}, and since~\eqref{quadratic-formulation} is
strictly convex, $\mathbf{w}^\ast = \mathbf{w}^{\mathrm{sep}}$.
\end{myproof}

\subsection{Proof of Proposition~\ref{equ-mint}: MinT as a Combination Procedure}
\label{eqv-mint}
\begin{proof}
Let $P=\Phi(\mathbf{w})C$ be the reconciliation matrix induced by the combination
weights $\mathbf{w}$, where $\Phi(\mathbf{w})$ is the block-diagonal matrix
in~\eqref{eq:phi} with $i$-th row block $\mathbf{w}_i^{\prime}$. The proof
identifies the combination weight problem with the MinT problem: the feasible
weights induce exactly the unbiased reconciliation matrices $\{P:PS=I_{n_b}\}$
(Steps~1--2), the combination and MinT objectives agree on this class (Step~3),
and its unique minimiser $P_{\mathrm{mint}}^\ast$ determines the optimal weights
(Step~4).

\emph{Step 1 (structure of $CS$).}
For each block, $C^{(i)}S=\mathbf{1}_{n_a+1}\mathbf{e}_i^{\prime}$; stacking over
$i$ gives $CS=I_{n_b}\otimes\mathbf{1}_{n_a+1}$, so by the block-diagonal form of
$\Phi(\mathbf{w})$,
\[
\Phi(\mathbf{w})CS
=\operatorname{diag}\bigl(\mathbf{1}_{n_a+1}^{\prime}\mathbf{w}_1,\ldots,
\mathbf{1}_{n_a+1}^{\prime}\mathbf{w}_{n_b}\bigr).
\]

\emph{Step 2 (feasible sets coincide).}
Consequently,
\[
PS=I_{n_b}
\ \Longleftrightarrow\ 
\mathbf{1}_{n_a+1}^{\prime}\mathbf{w}_i=1\ \ \forall i
\ \Longleftrightarrow\ 
\Pi \mathbf{w}=\mathbf{1}_{n_b},
\]
so $\mathbf{w}\mapsto\Phi(\mathbf{w})C$ maps
$\{\mathbf{w}:\Pi \mathbf{w}=\mathbf{1}_{n_b}\}$ into the unbiased class
$\{P:PS=I_{n_b}\}$. Conversely, by Theorem~\ref{Equivalence}, every $P$ with
$PS=I_{n_b}$ admits a representation $P=\Phi(\mathbf{w})C$ for some $\mathbf{w}$
with $\Pi \mathbf{w}=\mathbf{1}_{n_b}$. Hence the map is a surjection of
$\{\Pi \mathbf{w}=\mathbf{1}_{n_b}\}$ onto $\{P:PS=I_{n_b}\}$.

\emph{Step 3 (objectives coincide).}
Substituting $P=\Phi(\mathbf{w})C$, the combination
objective~\eqref{obj1} becomes
\begin{align*}
&\mathbb{E}\bigl[(\mathbf{y}_{t+h}-SP\hat{\mathbf{y}}_{t+h\mid t})^{\prime}
(\mathbf{y}_{t+h}-SP\hat{\mathbf{y}}_{t+h\mid t})\mid\mathcal{I}_t\bigr]
\\&\quad =\operatorname{tr}(
\mathbb{E}\bigl[(\mathbf{y}_{t+h}-SP\hat{\mathbf{y}}_{t+h\mid t})
(\mathbf{y}_{t+h}-SP\hat{\mathbf{y}}_{t+h\mid t})^{\prime}\mid\mathcal{I}_t\bigr]),
\end{align*}
which is exactly the MinT criterion in~\eqref{mint-loss} evaluated at the same
$P$ (equivalently $\mathbf{w}^{\prime}Q_h\mathbf{w}$ by
Theorem~\ref{quadratic-theorem}).

\emph{Step 4 (uniqueness and conclusion).}
By Steps~2--3 the two formulations minimise the same objective over the same class
$\{P:PS=I_{n_b}\}$. Since $\Sigma_h\succ0$ and $S$ has full column rank, this
objective is strictly convex in $P$ and the MinT problem has the unique minimiser
$P_{\mathrm{mint}}^\ast=(S^{\prime}\Sigma_h^{-1}S)^{-1}S^{\prime}\Sigma_h^{-1}$
in~\eqref{mint-solution}. Hence the optimal combination weights $\mathbf{w}^\ast$
in~\eqref{solution-w} satisfy $\Phi(\mathbf{w}^\ast)C=P_{\mathrm{mint}}^\ast$, and by
Proposition~\ref{prop:equivalence}, $\mathbf{w}^\ast=\mathbf{w}^{\mathrm{sep}}$, so
$\Phi(\mathbf{w}^{\mathrm{sep}})C=P_{\mathrm{mint}}^\ast$ as well. Combining,
\[
\Phi(\mathbf{w}^\ast)C=\Phi(\mathbf{w}^{\mathrm{sep}})C=P_{\mathrm{mint}}^\ast. \qedhere
\]
\end{proof}

\subsection{Proof of Proposition~\ref{prop:divergence}:
            Joint and Separate Optima under Egalitarian Penalization}
\label{divergence}

\begin{myproposition}[Divergence of \textit{Joint} and \textit{Separate}
  under egalitarian penalization]
  \label{prop:divergence}
  Under egalitarian penalization the Joint and Separate penalized solutions need
  not coincide:
  \begin{itemize}
  \item[\textup{(i)}] For the eRidge penalty, there exist
  $(\hat{Q}_h,\lambda)$ such that
  $\hat{\mathbf{w}}^{\mathrm{joint}}\neq
  \hat{\mathbf{w}}^{\mathrm{sep}}$.
  \item[\textup{(ii)}] For the eLASSO penalty, there exist
  $(\hat{Q}_h,\lambda)$ such that
  $\hat{\mathbf{w}}^{\mathrm{joint}}\neq
  \hat{\mathbf{w}}^{\mathrm{sep}}$.
  \end{itemize}
  Explicit constructions, using a $\hat{Q}_h$ realizable within the framework,
  are provided in Example~\ref{ex:ridge-divergence} and
  Example~\ref{ex:lasso-divergence} below.
  \end{myproposition}

Both parts assert the existence of $(\hat{Q}_h,\lambda)$ with
$\hat{\mathbf{w}}^{\mathrm{joint}}\neq\hat{\mathbf{w}}^{\mathrm{sep}}$, and both
are witnessed by a single structured instance, ensuring divergence occurs for a
$\hat{Q}_h$ the framework can actually produce. Take the hierarchy
$\mathrm{Total}=b_1+b_2$ in the order $(\mathrm{Total},b_1,b_2)$, the summing
matrix $S$ and candidate-generation blocks
\[
S=\begin{pmatrix}1&1\\1&0\\0&1\end{pmatrix},\quad
C^{(1)}=\begin{pmatrix}0&1&0\\1&0&-1\end{pmatrix},\quad
C^{(2)}=\begin{pmatrix}0&0&1\\1&-1&0\end{pmatrix},
\]
and the positive-definite base covariance
$\hat{\Sigma}_h=\bigl(\begin{smallmatrix}18&14&-2\\14&15&-5\\-2&-5&6\end{smallmatrix}\bigr)$.
Recalling $[\hat{Q}_h]_{i,i'}=z_{i,i'}\,C^{(i)}\hat{\Sigma}_h C^{(i')\prime}$ with
$z_{i,i'}=(S^{\prime}S)_{i,i'}$, this yields
\begin{equation}\label{eq:struct-witness}
\hat{Q}_h=\begin{pmatrix}30&38&-5&-1\\38&56&-8&1\\-5&-8&12&6\\-1&1&6&10\end{pmatrix}\succ0,
\qquad \lambda=1.
\end{equation}
This $\hat{Q}_h$ is realizable by construction, being built from a genuine
$\hat{\Sigma}_h\succ0$ and the candidate-generation blocks $C^{(i)}$; in
particular it satisfies the structural identity
$[\hat{Q}_h]_{1,2}[\hat{Q}_h]_{2,2}^{-1}\mathbf{1}_2\propto\mathbf{1}_2$ of
Lemma~\ref{lem:proportionality}. Example~\ref{ex:ridge-divergence} treats the
smooth eRidge penalty by direct comparison of closed forms on
\eqref{eq:struct-witness}; Example~\ref{ex:lasso-divergence} treats the
non-smooth eLASSO penalty on the same $\hat{Q}_h$ via subdifferential KKT
conditions.

\begin{myexample}[eRidge divergence]
\label{ex:ridge-divergence}
For the structured $\hat{Q}_h$ of~\eqref{eq:struct-witness}, the eRidge closed
forms of Proposition~\ref{prop:closed-form} are
\begin{align*}
\hat{\mathbf{w}}^{\mathrm{joint}}
&=\bigl(\tfrac{79}{48},-\tfrac{31}{48},\tfrac{23}{48},\tfrac{25}{48}\bigr)^{\prime},
\qquad\\
\hat{\mathbf{w}}^{\mathrm{sep}}
&=\bigl(\tfrac{37}{22},-\tfrac{15}{22},\tfrac{9}{22},\tfrac{13}{22}\bigr)^{\prime}.
\end{align*}
These differ, for instance the indirect weight of series~$1$ differs by
$-\tfrac{31}{48}+\tfrac{15}{22}=\tfrac{19}{528}\neq0$, which proves part~(i).
\end{myexample}

\begin{myexample}[eLASSO divergence]
\label{ex:lasso-divergence}
We analyse the same structured $\hat{Q}_h$ of~\eqref{eq:struct-witness} under the
non-smooth eLASSO penalty. With $n_a=1$, the constraint
$\mathbf{1}_2^{\prime}\mathbf{w}_i=1$ lets us write $\mathbf{w}_i=(1-w_i,\,w_i)^{\prime}$
with indirect weight $w_i:=J\mathbf{w}_i$ and egalitarian target
$\tfrac{1}{n_a+1}=\tfrac12$. Writing
$[\hat{Q}_h]_{i,i}=\bigl(\begin{smallmatrix}a_i&q_i\\ q_i&d_i\end{smallmatrix}\bigr)$
and $[\hat{Q}_h]_{1,2}=\bigl(\begin{smallmatrix}u_1&u_2\\ u_3&u_4\end{smallmatrix}\bigr)$,
each diagonal block contributes $g_iw_i^2+2\psi_iw_i$ with $g_i:=a_i-2q_i+d_i>0$
and $\psi_i:=q_i-a_i$, while the cross term
$2\mathbf{w}_1^{\prime}[\hat{Q}_h]_{1,2}\mathbf{w}_2$ contributes
$2g_{12}w_1w_2+2\theta_1w_1+2\theta_2w_2$ with $g_{12}:=u_1-u_2-u_3+u_4$,
$\theta_1:=u_3-u_1$ and $\theta_2:=u_2-u_1$. The Sep program uses only the
diagonal blocks; the Joint program additionally feels $[\hat{Q}_h]_{1,2}$
through both the coupling $g_{12}$ and the linear shifts $\theta_1,\theta_2$. For
\eqref{eq:struct-witness} this gives $g_1=g_2=10$, $\psi_1=8$, $\psi_2=-6$,
$g_{12}=5$, $\theta_1=-3$ and $\theta_2=4$. Each objective below is convex on all
of $\mathbb{R}$ (resp.\ $\mathbb{R}^2$), a positive-definite quadratic plus a
convex $\ell_1$ term, so any point meeting the (subdifferential) stationarity
conditions is the global minimizer.

\medskip\noindent\textit{Sep solution.}
Series~$1$ minimizes $10w_1^2+16w_1+\lambda|w_1-\tfrac12|$; its unconstrained
optimum lies below $\tfrac12$, where the penalty derivative is $-1$, giving
$20w_1+16-\lambda=0$, i.e. $\hat{w}_1^{\mathrm{sep}}=-\tfrac34$. Series~$2$
minimizes $10w_2^2-12w_2+\lambda|w_2-\tfrac12|$; its optimum lies above
$\tfrac12$ (derivative $+1$), giving $20w_2-12+\lambda=0$, i.e.
$\hat{w}_2^{\mathrm{sep}}=\tfrac{11}{20}$.

\medskip\noindent\textit{Joint solution.}
The Joint objective reduces to
$10w_1^2+10w_2^2+10w_1w_2+10w_1-4w_2+\lambda(|w_1-\tfrac12|+|w_2-\tfrac12|)$.
We seek an optimum with $w_1<\tfrac12$ and $w_2=\tfrac12$ (the latter at the
kink). The $w_1$-stationarity condition (penalty derivative $-1$) reads
$20w_1+10w_2+10-\lambda=0$; with $w_2=\tfrac12$ this gives
$\hat{w}_1^{\mathrm{joint}}=-\tfrac{7}{10}$. At the kink, subdifferential
optimality of $w_2$ requires the smooth part's $w_2$-derivative $20w_2+10w_1-4$
to lie in $[-\lambda,\lambda]$; here it equals $-1$, and $|-1|\le\lambda=1$, so
the condition holds and
$(\hat{w}_1^{\mathrm{joint}},\hat{w}_2^{\mathrm{joint}})=(-\tfrac{7}{10},\tfrac12)$
is the global Joint optimum.

\medskip\noindent\textit{Comparison.}
$\hat{w}_1^{\mathrm{joint}}-\hat{w}_1^{\mathrm{sep}}=-\tfrac{7}{10}+\tfrac34
=\tfrac1{20}\neq0$, so the indirect weights differ and hence
$\hat{\mathbf{w}}^{\mathrm{joint}}\neq\hat{\mathbf{w}}^{\mathrm{sep}}$, proving
part~(ii). The $\ell_1$ penalty pins $\hat{w}_2^{\mathrm{joint}}$ exactly at the
egalitarian target, illustrating how the off-diagonal block, discarded by the
Sep program, reshapes the Joint optimum.
\end{myexample}

\begin{proof}[Proof of Proposition~\ref{prop:divergence}]
Part~(i) follows from Example~\ref{ex:ridge-divergence}, and
part~(ii) follows from Example~\ref{ex:lasso-divergence}.
\end{proof}

\subsection{Closed-Form Solutions for eRidge}
\label{solut-eRidge}

For the auxiliary results below, define the indirect-weight selector
\begin{align*}
J &= \bigl[\,\mathbf{0}_{n_a\times 1}\ \ I_{n_a}\,\bigr]
\in\mathbb{R}^{n_a\times(n_a+1)},
\\ 
\tilde{J} &= I_{n_b}\otimes J
\in\mathbb{R}^{n_b n_a\times n_b(n_a+1)}.
\end{align*}
Thus $J\mathbf{w}_i$ collects all entries of $\mathbf{w}_i$ except the
direct-forecast weight, and $\tilde{J}\mathbf{w}$ stacks these indirect
weights across bottom-level series.

\begin{myproposition}[Closed-form solutions for eRidge]
  \label{prop:closed-form}
  Suppose $\hat{Q}_h\succ0$ and $\lambda\ge0$, and define
  \begin{align*}
  &\tilde{G}:=\hat{Q}_h+\lambda\tilde{J}^{\prime}\tilde{J},
  \quad
 \tilde{\mathbf{g}}:=\tfrac{\lambda}{n_a+1}\tilde{J}^{\prime}\mathbf{1}_{n_bn_a},
  \\\qquad
  &G_i:=[\hat{Q}_h]_{i,i}+\lambda J^{\prime}J,
  \quad
  \mathbf{g}:=\tfrac{\lambda}{n_a+1}J^{\prime}\mathbf{1}_{n_a}.
  \end{align*}
  Then for $p=2$ the Joint and Separate penalized programs admit the closed-form
  solutions

  \begin{align}
  &\hat{\mathbf{w}}^{\mathrm{joint}}
  =\tilde{G}^{-1}\tilde{\mathbf{g}}
  +\tilde{G}^{-1}{\Pi}^{\prime}\bigl(\Pi \tilde{G}^{-1}{\Pi}^{\prime}\bigr)^{-1}
   \bigl(\mathbf{1}_{n_b}-\Pi \tilde{G}^{-1}\tilde{\mathbf{g}}\bigr),
  \label{eq:joint-closed-form}\\
  &\hat{\mathbf{w}}_i^{\mathrm{sep}}
  =G_i^{-1}\mathbf{g}
  +\frac{1-\mathbf{1}_{n_a+1}^{\prime}G_i^{-1}\mathbf{g}}
        {\mathbf{1}_{n_a+1}^{\prime}G_i^{-1}\mathbf{1}_{n_a+1}}\,G_i^{-1}\mathbf{1}_{n_a+1},
  \quad i=1,\ldots,n_b.
  \label{eq:sep-closed-form}
  \end{align}
  When $\lambda=0$, these reduce to the unpenalized plug-in solution of
  Proposition~\ref{closed_solution} and the Bates--Granger solution
  in~\eqref{eq:BG}.
  \end{myproposition}

\begin{proof}
The joint eRidge objective is
$\mathbf{w}^{\prime}\hat{Q}_h\mathbf{w}
+\lambda\|\tilde{J}\mathbf{w}-\tfrac{1}{n_a+1}\mathbf{1}_{n_bn_a}\|_2^2$.
Expanding the penalty,
$\lambda \|\tilde{J}\mathbf{w}-\tfrac{1}{n_a+1}\mathbf{1}_{n_bn_a} \|_2^2
=\mathbf{w}^{\prime}(\lambda\tilde{J}^{\prime}\tilde{J})\mathbf{w}
-2\tilde{\mathbf{g}}^{\prime}\mathbf{w}+\mathrm{const}$, so the objective equals
$\mathbf{w}^{\prime}\tilde{G}\mathbf{w}-2\tilde{\mathbf{g}}^{\prime}\mathbf{w}+\mathrm{const}$.
Since $\hat{Q}_h\succ0$ and $\lambda\tilde{J}^{\prime}\tilde{J}\succeq0$ we have
$\tilde{G}\succ0$, and likewise each block $G_i\succ0$, so the program is a
strictly convex equality-constrained quadratic program. Following the proof of
Proposition~\ref{closed_solution}, with Lagrange multiplier vector
$\boldsymbol{\eta}$ the Lagrangian
$\mathcal{L}(\mathbf{w},\boldsymbol{\eta})
=\mathbf{w}^{\prime}\tilde{G}\mathbf{w}-2\tilde{\mathbf{g}}^{\prime}\mathbf{w}
-\boldsymbol{\eta}^{\prime}(\Pi \mathbf{w}-\mathbf{1}_{n_b})$
has first-order conditions
$2\tilde{G}\mathbf{w}-2\tilde{\mathbf{g}}={\Pi}^{\prime}\boldsymbol{\eta}$,
$\Pi \mathbf{w}=\mathbf{1}_{n_b}$. Thus
$\mathbf{w}=\tilde{G}^{-1}\bigl(\tilde{\mathbf{g}}+\tfrac12 \Pi^{\prime}\boldsymbol{\eta}\bigr)$;
substituting into the constraint gives
$\boldsymbol{\eta}=2 (\Pi \tilde{G}^{-1}{\Pi}^{\prime})^{-1}
(\mathbf{1}_{n_b}-\Pi \tilde{G}^{-1}\tilde{\mathbf{g}})$,
and back-substitution yields~\eqref{eq:joint-closed-form}. The separable program
decouples into $n_b$ blocks; for block $i$, with scalar multiplier $\eta_i$ for the
constraint $\mathbf{1}_{n_a+1}^{\prime}\mathbf{w}_i=1$, the same step gives
$2G_i\mathbf{w}_i-2\mathbf{g}=\eta_i\mathbf{1}_{n_a+1}$, hence
$\eta_i=2(1-\mathbf{1}_{n_a+1}^{\prime}G_i^{-1}\mathbf{g})/
(\mathbf{1}_{n_a+1}^{\prime}G_i^{-1}\mathbf{1}_{n_a+1})$ and~\eqref{eq:sep-closed-form}.
At $\lambda=0$, $\tilde{\mathbf{g}}=\mathbf{0}$ and $\tilde{G}=\hat{Q}_h$, so the
solutions collapse to the plug-in solution of Proposition~\ref{closed_solution} and
the Bates--Granger solution~\eqref{eq:BG} of Proposition~\ref{prop:equivalence}.
\end{proof}

\section{Grouped Time Series Forecast Reconciliation}
\label{app:grouped}

The main text develops the framework using hierarchical time series for
expositional clarity. This section clarifies that grouped time series require
no separate formalism: once their aggregation structure is encoded by the same
summing matrix $S$, all definitions, reconciliation maps, and combination-based
constructions apply unchanged. 

\paragraph{Grouped versus hierarchical structures.}
In a hierarchical (nested) structure, each aggregate splits into disjoint subgroups, so the series
form a single tree and the ancestors of each bottom-level series form a single chain. A grouped time
series instead classifies the bottom-level series by two or more \emph{crossed} attributes: the
bottom-level series are indexed by the observed combinations of the attribute levels, and each
aggregate is obtained by marginalizing over one or more attributes
\citep{athanasopoulos2024forecast}. A bottom-level series may therefore belong to several non-nested
aggregates. Unlike in a tree hierarchy, there is no unique aggregation path from a
bottom-level series to the total series. This needs no new machinery: both cases use the same relation
$\mathbf{a}_t=A_{\mathrm{agg}}\mathbf{b}_t$ and summing matrix $S$ in~\eqref{coherence}; the grouped structure
only changes the pattern of $A_{\mathrm{agg}}$, whose rows are $0/1$ membership indicators over the
bottom-level series. 

\paragraph{A small grouped example.}
Let $4$ bottom-level series be indexed by two attributes with classes $\{A,B\}$ and
$\{X,Y\}$ respectively, so $\mathcal{B}=\{AX,AY,BX,BY\}$. They aggregate along attribute~1 into $A,B$, along
attribute~2 into $X,Y$, and along both into $\mathrm{Total}$, giving $\mathcal{A}=\{\mathrm{Total},A,B,X,Y\}$. In this grouped system, we have 
$n_a=5$ aggregate series, $n_b = 4$ bottom-level series, and $n=9$ series in total:
\begin{align*}
  \mathbf{a}_t=\big(&y_t^{\mathrm{Total}},y_t^{A},y_t^{B},y_t^{X}, y_t^{Y}\big)^{\prime}, \\\mathbf{b}_t=\big(&y_t^{AX},y_t^{AY},y_t^{BX},y_t^{BY}\big)^{\prime},
\\\mathbf{y}_t=\big(&y_t^{\mathrm{Total}},y_t^{A},y_t^{B},y_t^{X},y_t^{Y},\\&y_t^{AX},y_t^{AY},y_t^{BX},y_t^{BY}\big)^{\prime}.
\end{align*}
This example has two crossed classification dimensions, so each bottom-level
cell contributes to one aggregate along each dimension as well as to the total. As in Section~2 we stack aggregates above bottom-level series, so that the rows of
$S$ follow the order of $\mathbf{y}_t$ and its columns follow the order of $\mathbf{b}_t$. With this fixed ordering, the aggregation matrix is
\[
A_{\mathrm{agg}}=
\begin{bmatrix}
1&1&1&1\\ 1&1&0&0\\ 0&0&1&1\\ 1&0&1&0\\ 0&1&0&1
\end{bmatrix},
\]
and the summing matrix is, as in~\eqref{coherence},
\[
S=\begin{bmatrix}A_{\mathrm{agg}}\\ I_4\end{bmatrix}\in\mathbb{R}^{9\times 4},
\qquad S\mathbf{b}_t=\mathbf{y}_t.
\]
Here $A_{\mathrm{agg}}$ records membership: $[A_{\mathrm{agg}}]_{ji}=1$ exactly when bottom-level series $i$ belongs to
aggregate $j$. The product $\mathbf{a}_t=A_{\mathrm{agg}}\mathbf{b}_t$ then performs the corresponding additive
aggregation, summing each aggregate over its member series. The crossing shows up as overlapping
membership: the rows for $A$ and $X$ both mark $AX$ as a member, yet neither row's member set is
contained in the other, so $AX$ belongs to two non-nested aggregates. Such a two-attribute
structure can be viewed informally as the union of two hierarchical trees that share their top- and
bottom-level series \citep{athanasopoulos2024forecast}; see Figure~\ref{fig:grouped}.

The complete $2\times2$ design is used only for illustration. The framework
also covers grouped structures with more attributes, more levels, or missing
attribute combinations, as long as the observed bottom-level series and their
aggregates can be encoded through the membership rows of $A_{\mathrm{agg}}$.

\begin{figure}[t]
\centering
\resizebox{\textwidth}{!}{\begin{tikzpicture}[
  >=latex,
  gnode/.style={draw,ellipse,fill=white,minimum width=12mm,minimum height=7.5mm,
                inner sep=1pt,font=\small},
  edge/.style={draw,line width=0.5pt},
]
\node[gnode] (gT)  at (0,4.0)   {$\mathrm{Total}$};
\node[gnode] (gX)  at (1.7,4.0) {$X$};
\node[gnode] (gY)  at (3.4,4.0) {$Y$};
\node[gnode] (gA)  at (0,2.6)   {$A$};
\node[gnode] (gAX) at (1.7,2.6) {$AX$};
\node[gnode] (gAY) at (3.4,2.6) {$AY$};
\node[gnode] (gB)  at (0,1.2)   {$B$};
\node[gnode] (gBX) at (1.7,1.2) {$BX$};
\node[gnode] (gBY) at (3.4,1.2) {$BY$};
\node[font=\Large] at (4.3,2.6) {$=$};
\node[gnode] (T1)  at (7.1,4.4) {$\mathrm{Total}$};
\node[gnode] (A1)  at (5.7,2.6) {$A$};
\node[gnode] (B1)  at (8.5,2.6) {$B$};
\node[gnode] (AX1) at (5.0,0.7) {$AX$};
\node[gnode] (AY1) at (6.4,0.7) {$AY$};
\node[gnode] (BX1) at (7.8,0.7) {$BX$};
\node[gnode] (BY1) at (9.2,0.7) {$BY$};
\draw[edge] (T1)--(A1); \draw[edge] (T1)--(B1);
\draw[edge] (A1)--(AX1); \draw[edge] (A1)--(AY1);
\draw[edge] (B1)--(BX1); \draw[edge] (B1)--(BY1);
\node[font=\Large] at (10.1,2.6) {$\bigcup$};
\node[gnode] (T2)  at (13.0,4.4) {$\mathrm{Total}$};
\node[gnode] (X2)  at (11.6,2.6) {$X$};
\node[gnode] (Y2)  at (14.4,2.6) {$Y$};
\node[gnode] (AX2) at (10.9,0.7) {$AX$};
\node[gnode] (BX2) at (12.3,0.7) {$BX$};
\node[gnode] (AY2) at (13.7,0.7) {$AY$};
\node[gnode] (BY2) at (15.1,0.7) {$BY$};
\draw[edge] (T2)--(X2); \draw[edge] (T2)--(Y2);
\draw[edge] (X2)--(AX2); \draw[edge] (X2)--(BX2);
\draw[edge] (Y2)--(AY2); \draw[edge] (Y2)--(BY2);
\end{tikzpicture}}
\caption{The $2\times2$ grouped time series of the example: the left panel displays the crossed
classification table, with rows indexed by attribute~1 and columns indexed by attribute~2, viewed
informally as the union of two hierarchical trees with common top- and bottom-level series, one per
attribute. This depiction is a visual device following \citet{athanasopoulos2024forecast}, not a
graph-theoretic equivalence.}
\label{fig:grouped}
\end{figure}
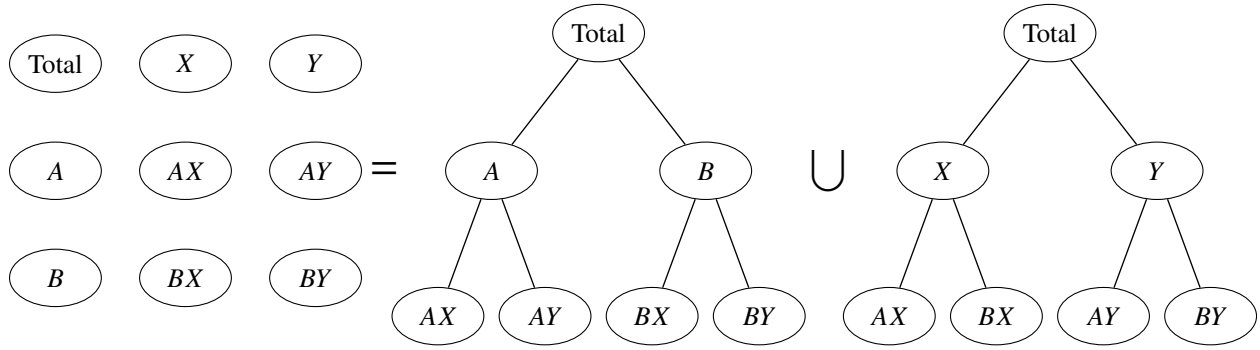

\paragraph{Ancestors and collaterals on a DAG.}
A grouped structure is a directed acyclic graph (DAG) rather than a tree: a bottom-level series
may have several parents (here, one per attribute). The membership notation of Section~2 still
specializes directly, because it is defined through set membership rather than tree position. The
bottom-level descendant sets are $\mathrm{Desc}_{\mathrm{bot}}(\mathrm{Total})=\{AX,AY,BX,BY\}$, $\mathrm{Desc}_{\mathrm{bot}}(A)=\{AX,AY\}$,
$\mathrm{Desc}_{\mathrm{bot}}(B)=\{BX,BY\}$, $\mathrm{Desc}_{\mathrm{bot}}(X)=\{AX,BX\}$, and $\mathrm{Desc}_{\mathrm{bot}}(Y)=\{AY,BY\}$. Since
$\mathrm{Anc}(i)=\{j:i\in\mathrm{Desc}_{\mathrm{bot}}(j)\}$ and $\mathrm{Col}(i)=\mathcal{A}\setminus\mathrm{Anc}(i)$ depend only on these
membership sets, they remain well-defined on the DAG; for $AX$, $\mathrm{Anc}(AX)=\{\mathrm{Total},A,X\}$ and
$\mathrm{Col}(AX)=\{B,Y\}$. Tree-specific notions should therefore be interpreted with care in the
grouped case. In particular, the ancestors of $AX$ are not totally ordered, since $A$ and $X$ are
incomparable. Accordingly, the
tree-based description of $\mathrm{Anc}(i)$ in Section~2 should be read, in the grouped case, simply as the
set of aggregates that contain $i$.

\paragraph{The reconciliation framework.}
The example satisfies the coherence relation~\eqref{coherence} with a binary $S$, so it is a
special case of the Section~2 aggregation structure, differing only in the pattern of $A_{\mathrm{agg}}$.
Since every construction in the paper is expressed through $S$, all carry over unchanged: coherent
forecasts (Definition~\ref{def-coherence}), the linear reconciliation map~\eqref{linear-recon-1}
under $PS=I_{n_b}$, the MinT solution~\eqref{mint-solution}, and the combination framework
developed later all apply analogously. This matches the literature, where the same reconciliation
methods are applied to hierarchical and grouped structures through a common summing matrix
\citep{hyndman2016fast,wickramasuriya2019optimal}; the empirical study uses a grouped dataset on
exactly this basis. Consequently, the grouped empirical application in
Section~\ref{sec:empirical} is not an extension of the method, but a direct application of the same
method.

\section{Unbalanced Hierarchical Structures}
\label{sec:unbalanced-structure}

This section illustrates how the notation and aggregation representation
introduced in the main text apply to an unbalanced hierarchy. A hierarchy
is balanced if all paths from the top-level series to the bottom-level
series have the same depth. It is unbalanced if bottom-level or terminal
series occur at different depths, so that the top-to-bottom path lengths
are unequal~\citep{di2024forecast}.

Specifically, consider the unbalanced hierarchy shown in
Figure~\ref{fig:hierarchy}, with
$
\mathbf{y}_t
=
\bigl(
y_t^{U_1},
y_t^{U_2},
y_t^{U_3},
y_t^{U_4},
y_t^{B_1},
y_t^{B_2},
y_t^{B_3},
y_t^{B_4},
y_t^{B_5}
\bigr)^{\prime},
$
$
\mathbf{a}_t
=
\bigl(
y_t^{U_1},
y_t^{U_2},
y_t^{U_3},
y_t^{U_4}
\bigr)^{\prime},
$
and
$
\mathbf{b}_t
=
\bigl(
y_t^{B_1},
y_t^{B_2},
y_t^{B_3},
y_t^{B_4},
y_t^{B_5}
\bigr)^{\prime}.
$
The corresponding summing matrix is
\[
S=
\begin{bmatrix}
1&1&1&1&1\\
1&1&1&0&0\\
0&0&0&1&1\\
0&1&1&0&0\\
& & I_5 & &
\end{bmatrix}.
\]

\begin{figure}[htbp]
\centering
\begin{tikzpicture}[
  every node/.style={font=\small,inner sep=1.5pt},
  line width=0.5pt,
  x=1cm,y=1cm
]
\node (U1) at (3.6,4.2){$U_1$};
\node (U2) at (1.4,2.8){$U_2$};
\node (U3) at (5.8,2.8){$U_3$};
\node (B1) at (0.5,1.4){$B_1$};
\node (U4) at (2.2,1.4){$U_4$};
\node (B4) at (5.0,1.4){$B_4$};
\node (B5) at (6.8,1.4){$B_5$};
\node (B2) at (1.7,0.0){$B_2$};
\node (B3) at (2.8,0.0){$B_3$};

\draw (U1)--(U2);
\draw (U1)--(U3);
\draw (U2)--(B1);
\draw (U2)--(U4);
\draw (U3)--(B4);
\draw (U3)--(B5);
\draw (U4)--(B2);
\draw (U4)--(B3);
\end{tikzpicture}
\caption{\centering A hierarchical structure with unbalanced depth.}
\label{fig:hierarchy}
\end{figure}
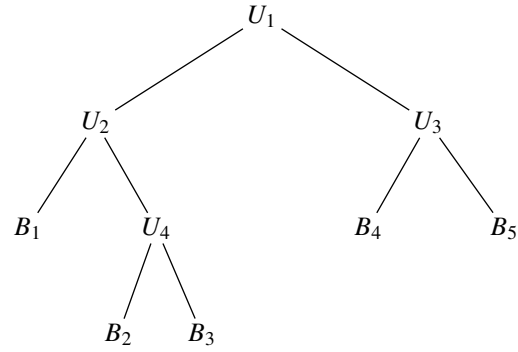

For example, the aggregate ancestors of bottom-level series $B_1$ are
$
\mathrm{Anc}(B_1)=\{U_1,U_2\},
$
while its collateral aggregate series are
$
\mathrm{Col}(B_1)=\{U_3,U_4\}.
$
Hence, the definitions of ancestor and collateral aggregate series remain
well defined even though the hierarchy is unbalanced. The notation,
coherence representation, and summing-matrix formulation introduced in
the main text therefore apply without requiring a balanced structure.
\section{The Local Combination Interpretation of LCC}\label{sec:LCC}
This section provides a concrete example of the combination interpretation
implicit in LCC, which is used in the main text as the closest existing
connection between reconciliation and forecast combination. In the context of
hierarchical forecasting, \citet{hollyman2021understanding} established this
connection through their level-conditional coherent (LCC) method.
LCC decomposes a multi-level hierarchy into a sequence of two-level subsystems, each consisting of the aggregate series at one selected upper level and the bottom-level series that reconcile to those aggregates. For example, the three-level hierarchy in Figure~\ref{example_1}, with aggregated series $\{\mathrm{Total}, A, B\}$ and bottom-level series $\{AA, AB, BA, BB\}$, can be decomposed into two two-level systems:
(i) the intermediate system $\{A,B\}$ versus all bottom-level series, shown in the lower panel of Figure~\ref{two_level_syst}; and
(ii) the top system $\{\mathrm{Total}\}$ versus all bottom-level series, shown in the upper panel of Figure~\ref{two_level_syst}.
\begin{figure}[htbp]
\centering
\includegraphics[width=0.4\textwidth]{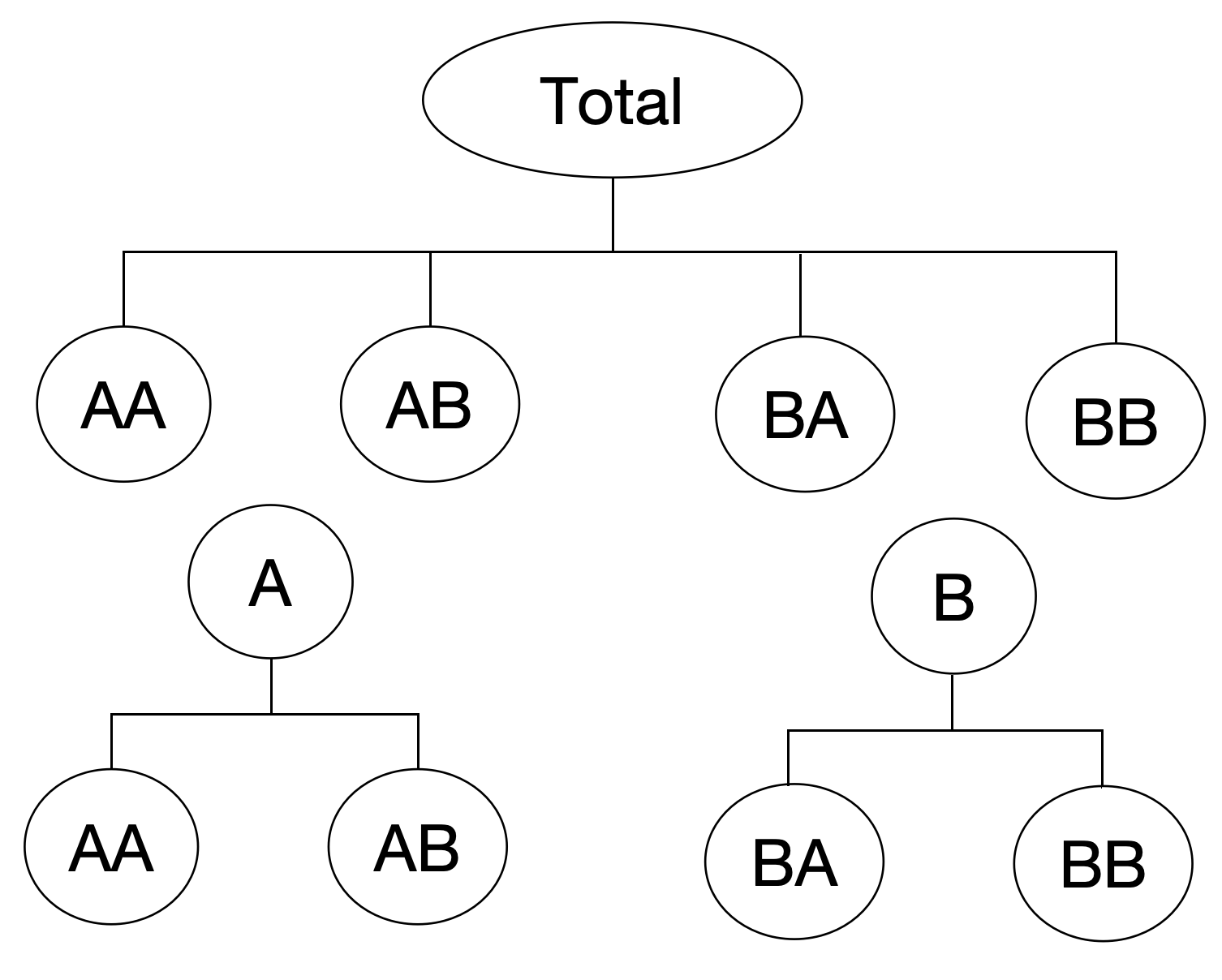}
\vspace{-10pt}
\caption{\centering Two two-level subsystems derived from the three-level hierarchy in Figure~\ref{example_1}.}
\label{two_level_syst}
\end{figure}

At each selected upper level, LCC produces bottom-level forecasts that are
locally coherent with the aggregate base forecasts at that level. LCC relies on
the simplifying assumption that bottom-level base forecast errors are
uncorrelated, so that the error covariance matrix in each subproblem is
diagonal. Under this assumption, each adjusted bottom-level forecast admits the
form of a linear combination of its base forecast and an indirect forecast
obtained by taking the relevant aggregate forecast and subtracting the base
forecasts of the other bottom-level descendants in the same local subsystem.
After the level-specific bottom-level forecasts are averaged, the final
bottom-level vector is mapped through $S$ to obtain a globally coherent forecast
vector.

To illustrate, at the intermediate level in the lower panel of
Figure~\ref{two_level_syst}, reconciliation is performed within this subsystem,
subject to the constraints that the adjusted forecasts for the descendants of
$A$ and $B$ sum to $\hat{y}_{t+h|t}^{A}$ and $\hat{y}_{t+h|t}^{B}$,
respectively. Under the assumption of uncorrelated errors, the reconciled
forecast for series $AA$ can be written as
\[
\tilde{b}_{t+h|t}^{(AA,A)}
=
w_{AA}^{(A)}\hat{y}_{t+h|t}^{AA}
+(1-w_{AA}^{(A)})
\bigl(\hat{y}_{t+h|t}^{A}-\hat{y}_{t+h|t}^{AB}\bigr),
\]
where $w_{AA}^{(A)}$ is the corresponding combination weight specific to
reconciling series $AA$ with its ancestor series $A$. The term
$\hat{y}_{t+h|t}^{A}-\hat{y}_{t+h|t}^{AB}$ serves as an indirect forecast of
$AA$, obtained by adjusting the aggregate forecast of its ancestor series $A$ to
remove the contribution of its sibling series $AB$. Analogously, the reconciled
forecast for $BA$ is
\[
\tilde{b}_{t+h|t}^{(BA,B)}
=
w_{BA}^{(B)}\hat{y}_{t+h|t}^{BA}
+(1-w_{BA}^{(B)})
\bigl(\hat{y}_{t+h|t}^{B}-\hat{y}_{t+h|t}^{BB}\bigr),
\]
where $w_{BA}^{(B)}$ is the level-specific combination weight assigned to the
base forecast of $BA$ with its ancestor series $B$.

A second reconciliation step is performed at the top level between 
$\{\mathrm{Total}\}$ and all bottom-level series in the upper panel of
Figure~\ref{two_level_syst}. This yields another set of revised bottom-level
forecasts. For series $AA$, the top-level reconciled forecast is
$\tilde{b}_{t+h|t}^{(AA,\mathrm{Total})}
=
w_{AA}^{(\mathrm{Total})}\hat{y}_{t+h|t}^{AA}
+
(1-w_{AA}^{(\mathrm{Total})})
(
\hat{y}_{t+h|t}^{\mathrm{Total}}
-\hat{y}_{t+h|t}^{AB}
-\hat{y}_{t+h|t}^{BA}
-\hat{y}_{t+h|t}^{BB})$,
where $w_{AA}^{(\mathrm{Total})}$ is a distinct top-level combination weight.
The term
$\hat{y}_{t+h|t}^{\mathrm{Total}}-\hat{y}_{t+h|t}^{AB}
-\hat{y}_{t+h|t}^{BA}-\hat{y}_{t+h|t}^{BB}$ is the corresponding top-level
indirect forecast for $AA$, obtained by removing the contributions of all
sibling series $\{AB, BA, BB\}$ from the aggregate forecast of series
$\mathrm{Total}$.

Combining the two level-specific adjustments, the final reconciled forecast $\tilde{b}_{t+h \mid t}^{(AA)}$
for series $AA$ is obtained by averaging the intermediate- and top-level
adjustments, yielding a linear combination of three candidate forecasts:
\begin{equation}
\label{averaging}
\begin{aligned}
\tilde{b}_{t+h \mid t}^{(AA)}
&= \underbrace{\frac{w_{AA}^{(A)}+w_{AA}^{(\mathrm{Total})}}{2}}_{\text{direct weight}}\,
   \underbrace{\hat{y}_{t+h \mid t}^{AA}}_{\text{direct forecast}}
 \\&+ \underbrace{\frac{1-w_{AA}^{(A)}}{2}}_{\text{weight (}A\text{)}}\,
   \underbrace{\bigl(\hat{y}_{t+h \mid t}^{A}-\hat{y}_{t+h \mid t}^{AB}\bigr)}_{\text{indirect forecast (ancestor }A\text{)}} \\
\quad
 &+ \underbrace{\frac{1-w_{AA}^{(\mathrm{Total})}}{2}}_{\text{weight (}\mathrm{Total}\text{)}}\,
   \underbrace{\bigl(\hat{y}_{t+h \mid t}^{\mathrm{Total}}-\hat{y}_{t+h \mid t}^{AB}
   -\hat{y}_{t+h \mid t}^{BA}-\hat{y}_{t+h \mid t}^{BB}\bigr)}_{\text{indirect forecast (ancestor }\mathrm{Total}\text{)}}.
\end{aligned}
\end{equation}

\noindent
As shown in~\eqref{averaging}, the direct base forecast $\hat{y}_{t+h|t}^{AA}$
receives a weight equal to the average of the two level-specific weights
$w_{AA}^{(A)}$ and $w_{AA}^{(\mathrm{Total})}$, while each indirect forecast receives
the complementary weight  from its corresponding ancestor level.
The indirect forecast from ancestor $A$ is obtained by subtracting the sibling
forecast $\hat{y}_{t+h|t}^{AB}$ from the aggregate forecast $\hat{y}_{t+h|t}^{A}$;
the indirect forecast from ancestor $\mathrm{Total}$ is obtained analogously by removing
the contributions of all sibling series $\{AB, BA, BB\}$.
Because each level-specific adjustment preserves the weight-sum-to-one property
and the final forecast is their average, the three coefficients automatically
sum to one. The same procedure applies to all bottom-level series, yielding
the vector of reconciled bottom-level forecasts $\tilde{\mathbf{b}}_{t+h \mid t}$,
from which coherent forecasts for all hierarchical series are recovered via
$\tilde{\mathbf{y}}_{t+h\mid t} = S\,\tilde{\mathbf{b}}_{t+h\mid t}$.

This example shows why LCC can be interpreted as a forecast-combination
procedure: its final bottom-level forecast is a weighted average of the direct
base forecast and ancestor-based indirect forecasts. For series $AA$, the
implicit candidate set consists of
$\hat{y}_{t+h|t}^{AA}$,
$\hat{y}_{t+h|t}^{A}-\hat{y}_{t+h|t}^{AB}$, and
$\hat{y}_{t+h|t}^{\mathrm{Total}}-\hat{y}_{t+h|t}^{AB}
-\hat{y}_{t+h|t}^{BA}-\hat{y}_{t+h|t}^{BB}$.
However, this set does not exhaust the structurally valid candidate forecasts
generated by the aggregation constraints. For example, the non-ancestor
aggregate series $B$ can also be used to replace the bottom-level sibling sum
$\hat{y}_{t+h|t}^{BA}+\hat{y}_{t+h|t}^{BB}$ in the total-level identity,
yielding another valid indirect forecast for $AA$: $\hat{y}_{t+h|t}^{\mathrm{Total}}
-\hat{y}_{t+h|t}^{AB}
-\hat{y}_{t+h|t}^{B}$. 
This candidate is valid because, under coherence,
$y_{t+h}^{B}=y_{t+h}^{BA}+y_{t+h}^{BB}$, so the expression still targets
$y_{t+h}^{AA}$. It is nevertheless absent from LCC, because LCC constructs
candidate forecasts only through local two-level reconciliations involving the
target series' ancestor aggregates.

Furthermore, in LCC, the final weights are induced
by solving separate local reconciliation problems and then averaging the
resulting level-specific forecasts. They are not obtained from a
single optimization problem over the full set of structurally valid candidate
forecasts. By contrast, our framework first constructs a maximal linearly
independent set of valid candidates, including both ancestor-based and
collateral aggregate-based candidates such as the one above, and then estimates
combination weights directly through a global objective, without the simplifying error uncorrelatedness assumption. Overall, LCC provides a
useful local combination interpretation of reconciliation under simplifying assumptions, whereas our approach
gives a general combination formulation that preserves the full candidate space
and optimizes the weights jointly.

\section{Application of Algorithm~\ref{algorithm-1}
to an Unbalanced Hierarchy}
\label{unbalanced-example}

We now apply Algorithm~\ref{algorithm-1} to the unbalanced hierarchy
introduced in Section~\ref{sec:unbalanced-structure}. We take the
bottom-level series $B_1$ as an example and construct its candidate
forecasts step by step.

\paragraph{Step 1.}
Following line~1 of Algorithm~\ref{algorithm-1}, we first compute the
four vectors $\mathbf{v}^{(U_1)},\ldots,\mathbf{v}^{(U_4)}$, each
corresponding to the aggregation constraint associated with one
aggregate series. These vectors are
\[
\begin{aligned}
\mathbf{v}^{(U_1)}
&=(1,0,0,0,-1,-1,-1,-1,-1)^\prime,\\
\mathbf{v}^{(U_2)}
&=(0,1,0,0,-1,-1,-1,0,0)^\prime,\\
\mathbf{v}^{(U_3)}
&=(0,0,1,0,0,0,0,-1,-1)^\prime,\\
\mathbf{v}^{(U_4)}
&=(0,0,0,1,0,-1,-1,0,0)^\prime.
\end{aligned}
\]
They correspond, respectively, to the aggregation constraints
\[
\begin{aligned}
y_{t+h}^{U_1}
&=
y_{t+h}^{B_1}
+y_{t+h}^{B_2}
+y_{t+h}^{B_3}
+y_{t+h}^{B_4}
+y_{t+h}^{B_5},\\
y_{t+h}^{U_2}
&=
y_{t+h}^{B_1}
+y_{t+h}^{B_2}
+y_{t+h}^{B_3},\\
y_{t+h}^{U_3}
&=
y_{t+h}^{B_4}
+y_{t+h}^{B_5},\\
y_{t+h}^{U_4}
&=
y_{t+h}^{B_2}
+y_{t+h}^{B_3}.
\end{aligned}
\]

\paragraph{Step 2.}
Following line~3 of Algorithm~\ref{algorithm-1}, the direct-forecast
coefficient vector for $B_1$ is
\[
\mathbf{v}_0^{(B_1)}
=
(0,0,0,0,1,0,0,0,0)^{\prime}.
\]

\paragraph{Step 3.}
Using
$\mathrm{Anc}(B_1)=\{U_1,U_2\}$ and following line 3 of Algorithm 1, the two ancestor-based candidate
coefficient vectors are
\[
\begin{aligned}
\mathbf{v}_0^{(B_1)}+\mathbf{v}^{(U_1)}
&=
(1,0,0,0,0,-1,-1,-1,-1)^\prime,\\
\mathbf{v}_0^{(B_1)}+\mathbf{v}^{(U_2)}
&=
(0,1,0,0,0,-1,-1,0,0)^\prime.
\end{aligned}
\]

\paragraph{Step 4.}
Using
$\mathrm{Col}(B_1)=\{U_3,U_4\}$ and  following line 4 of Algorithm 1, the two collateral-based candidate
coefficient vectors are
\[
\begin{aligned}
\mathbf{v}_0^{(B_1)}
+\mathbf{v}^{(U_1)}
-\mathbf{v}^{(U_3)}
&=
(1,0,-1,0,0,-1,-1,0,0)^\prime,\\
\mathbf{v}_0^{(B_1)}
+\mathbf{v}^{(U_1)}
-\mathbf{v}^{(U_4)}
&=
(1,0,0,-1,0,0,0,-1,-1)^\prime.
\end{aligned}
\]

\paragraph{Step 5.}
Applying the candidate coefficient vectors to
$\hat{\mathbf{y}}_{t+h\mid t}$ produces the following candidate
forecasts:
\[
\begin{aligned}
&\hat{y}^{B_1}_{t+h\mid t},\\
&\hat{y}^{U_1}_{t+h\mid t}
-\hat{y}^{B_2}_{t+h\mid t}
-\hat{y}^{B_3}_{t+h\mid t}
-\hat{y}^{B_4}_{t+h\mid t}
-\hat{y}^{B_5}_{t+h\mid t},\\
&\hat{y}^{U_2}_{t+h\mid t}
-\hat{y}^{B_2}_{t+h\mid t}
-\hat{y}^{B_3}_{t+h\mid t},\\
&\hat{y}^{U_1}_{t+h\mid t}
-\hat{y}^{U_3}_{t+h\mid t}
-\hat{y}^{B_2}_{t+h\mid t}
-\hat{y}^{B_3}_{t+h\mid t},\\
&\hat{y}^{U_1}_{t+h\mid t}
-\hat{y}^{U_4}_{t+h\mid t}
-\hat{y}^{B_4}_{t+h\mid t}
-\hat{y}^{B_5}_{t+h\mid t}.
\end{aligned}
\]

These five candidate forecasts are linearly independent and maximal:
all other valid candidate forecasts for series $B_1$ can be expressed
as their linear combinations. This example illustrates how
Algorithm~\ref{algorithm-1} constructs a maximal set of interpretable
candidate forecasts in an unbalanced hierarchical structure,
complementing the balanced hierarchy considered in the main text.

\section{Unconstrained Optimization Reformulation for Practical Implementation}
\label{sec:unconstrained}
The optimization problems in~\eqref{quadratic-formulation} and
\eqref{empirical-general-penalized} are  formulated as
equality-constrained quadratic programs. Although these constraints are useful
for theoretical characterization, they can make practical implementation
less convenient, especially when estimating penalized weights with standard
optimization routines. We therefore use a simple reparameterization that
removes the equality constraints while preserving the feasible set exactly.
The key idea is to express each weight vector as an affine transformation
of a free parameter vector, so that the sum-to-one constraint is satisfied
by construction.

Specifically, write
$\mathbf{w}_i=Z\boldsymbol{\rho}_i+\mathbf{d}$, where
\[
Z= \begin{bmatrix}
- \mathbf{1}_{n_a}^{\prime} \\
{I}_{n_a}
\end{bmatrix}
, \quad
\mathbf{d} = \begin{bmatrix}
1 \\
\mathbf{0}_{n_a}
\end{bmatrix},
\]
and $\boldsymbol{\rho}_i\in\mathbb{R}^{n_a}$ give free unconstrained reparameterization. This reparameterization automatically enforces
$\mathbf{1}_{n_a+1}^\prime\mathbf{w}_i=1$. Let
$\boldsymbol{\rho}=(\boldsymbol{\rho}_1^\prime,\ldots,\boldsymbol{\rho}_{n_b}^\prime)^\prime$.

\begin{myproposition}[Unconstrained optimization reformulation]
\label{unconstrained-q}
Under the above reparameterization, problem~\eqref{empirical-general-penalized} 
is equivalent to the following unconstrained optimization problem
\begin{equation}
\label{q-r}
\min_{\boldsymbol{\rho}} \quad
\boldsymbol{\rho}^{\prime} (\tilde{Z}^{\prime} \hat Q_h \tilde{Z}) \boldsymbol{\rho}
+ (2\, \tilde{Z}^{\prime} \hat Q_h\tilde{\mathbf{d}})^{\prime} \boldsymbol{\rho} + \lambda \mathcal{P}(\tilde Z\boldsymbol{\rho}+ \tilde{\mathbf{d}}),
\end{equation}
where $\tilde{Z} = I_{n_b} \otimes Z $ is a block-diagonal transformation matrix, $\tilde{\mathbf{d}} = \mathbf{1}_{n_b} \otimes \mathbf{d}$ is an offset vector,  and $\otimes$ denotes the Kronecker product. If $\hat Q_h\succeq0$ and $\mathcal{P}$ is convex, the problem is convex.
\end{myproposition}

The reparameterization is especially simple for the egalitarian regularization
introduced in Section~\ref{e-penalty}. Since
\[
\mathbf{w}_i
=
\begin{bmatrix}
1-\mathbf{1}_{n_a}^{\prime}\boldsymbol{\rho}_i\\
\boldsymbol{\rho}_i
\end{bmatrix},
\]
the vector $\boldsymbol{\rho}_i$ is exactly the subvector of indirect-forecast
weights $\mathbf{w}_{i,2:(n_a+1)}$. Hence the egalitarian penalty in
\eqref{eq:egalitarian-penalty} becomes
\[
\mathcal{P}^{\mathrm{eq}}(\tilde Z\boldsymbol{\rho}+\tilde{\mathbf{d}})
=
\sum_{i=1}^{n_b}
\|
\boldsymbol{\rho}_i
-\frac{1}{n_a+1}\mathbf{1}_{n_a}\|_p^p .
\]
Therefore the unconstrained form of the egalitarian problem is
\begin{equation}
\label{q-r-egalitarian}
\min_{\boldsymbol{\rho}} \quad
\boldsymbol{\rho}^{\prime} H \boldsymbol{\rho}
+2{\mathbf{q}}^{\prime}\boldsymbol{\rho}
+ \lambda
\sum_{i=1}^{n_b}
\|
\boldsymbol{\rho}_i
-\frac{1}{n_a+1}\mathbf{1}_{n_a}\|_p^p,
\end{equation}
where
\[
H=\tilde{Z}^{\prime}\hat Q_h\tilde{Z},
\qquad
\mathbf{q}=\tilde{Z}^{\prime}\hat Q_h\tilde{\mathbf{d}}.
\]
For $p=2$, \eqref{q-r-egalitarian} is an unconstrained smooth convex quadratic
problem. Let
$\boldsymbol{\tau}
=\mathbf{1}_{n_b}\otimes
\frac{1}{n_a+1}\mathbf{1}_{n_a}$.
Ignoring the additive constant $\lambda\boldsymbol{\tau}^{\prime}\boldsymbol{\tau}$,
the objective can be written as
\[
\boldsymbol{\rho}^{\prime}(H+\lambda I_{n_bn_a})\boldsymbol{\rho}
+2(\mathbf{q}-\lambda\boldsymbol{\tau})^{\prime}\boldsymbol{\rho}.
\]
Thus, when $H+\lambda I_{n_bn_a}$ is nonsingular, the eRidge solution is
\[
\hat{\boldsymbol{\rho}}
=-(H+\lambda I_{n_bn_a})^{-1}(\mathbf{q}-\lambda\boldsymbol{\tau}),
\qquad
\hat{\mathbf{w}}=\tilde Z\hat{\boldsymbol{\rho}}+\tilde{\mathbf{d}}.
\]
For $p=1$, \eqref{q-r-egalitarian} is an unconstrained convex but nonsmooth
problem. It has no equality constraints after reparameterization and can be
solved using standard algorithms for shifted LASSO-type objectives, such as
proximal-gradient, coordinate-descent, or ADMM methods. The \emph{Separate}
implementation is obtained by replacing $\hat Q_h$ with
$\hat Q_h^{\mathrm{bd}}$ in the definitions of $H$ and $\mathbf{q}$; in that case
\eqref{q-r-egalitarian} decomposes into $n_b$ independent low-dimensional
problems in the vectors $\boldsymbol{\rho}_i$.

\section{Additional Details for the Empirical Evaluation}
\label{app:empirical-details}

\subsection{Dataset Plots}

Figure~\ref{elct} displays the electricity generation series across
hierarchical levels. Aggregate total generation has pronounced weekly
patterns and seasonal peaks in winter and summer, while several
bottom-level series, such as wind and distillate, display more irregular
patterns.

\begin{figure}[htbp]
\centering
\includegraphics[width=1\textwidth]{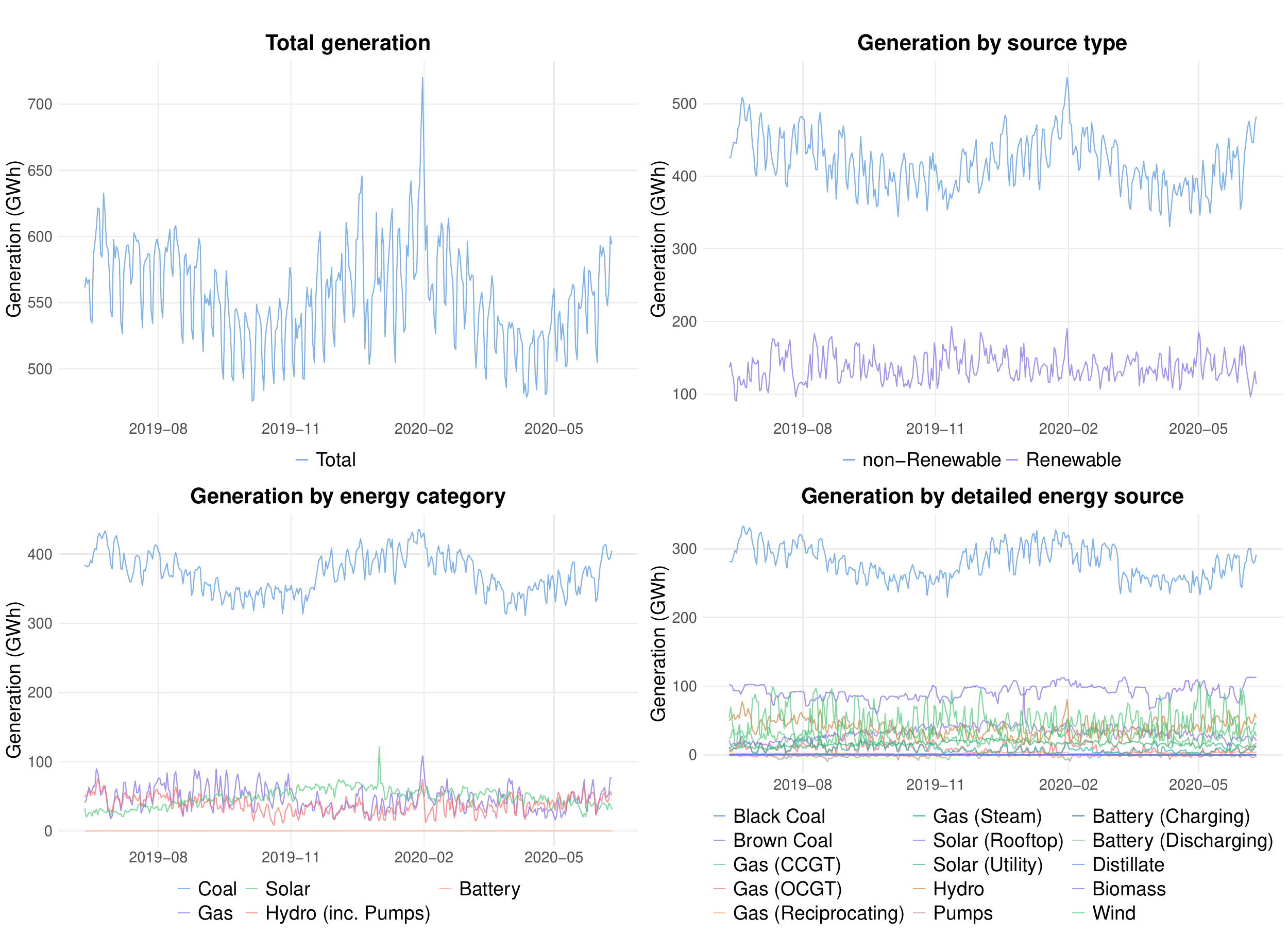}
\vspace{-10pt}
\caption{\centering Australian electricity generation from different sources of energy.}
\label{elct}
\end{figure}

Figure~\ref{labour} shows selected series from the grouped labor force
dataset. The top-level unemployment series has regular seasonality and a
sharp rise around the early-2020 COVID-19 restrictions, while lower-level
duration and state/territory series display more heterogeneous dynamics.

\begin{figure}[ht]
\centering
\includegraphics[width=1\textwidth]{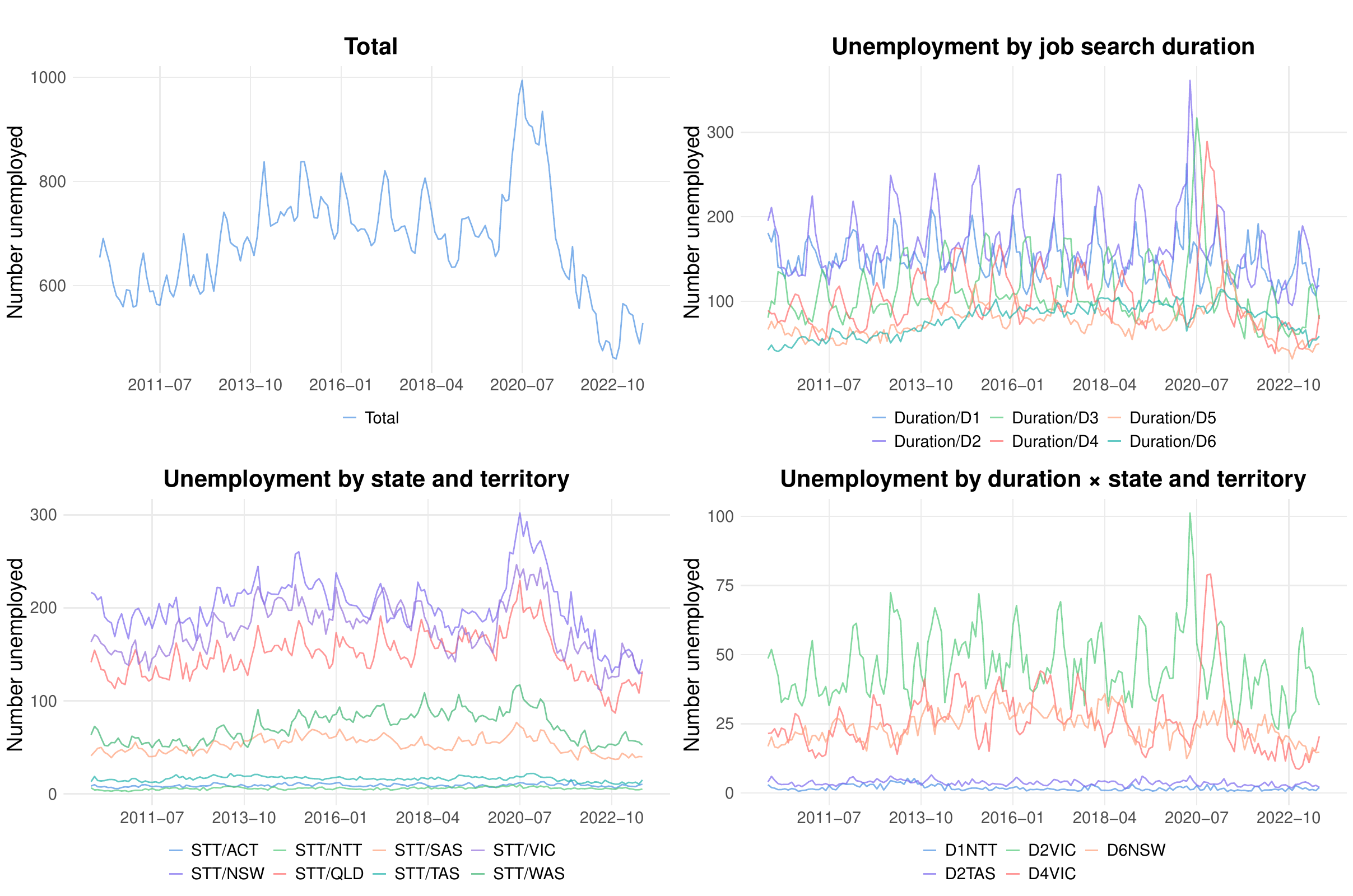}
\vspace{-10pt}
\caption{\centering Australian unemployment, disaggregated by duration of
job search and by state and territory}
\label{labour}
\end{figure}

\subsection{Implementation Details for the Empirical Evaluation}
\label{implem-details}

This subsection explains how forecasts are generated, reconciled, tuned, and
evaluated in the empirical studies. We first describe the rolling-origin
scheme and the role played by each part of the data. We then explain the
training-window estimation and tuning steps used to obtain the reconciliation
inputs at each origin. Finally, we describe how the testing data are forecast,
reconciled, and evaluated by RMSE.

\paragraph{Rolling-origin forecasting design.}
The last 20\% of observations were reserved as the outer test period. We use a
fixed-length rolling window: for origin $r$, models are estimated using
$\{\mathbf{y}_r,\mathbf{y}_{r+1},\ldots,\mathbf{y}_{r+T-1}\}$, and forecasts are
evaluated over the following $H$ periods
$\{\mathbf{y}_{r+T},\ldots,\mathbf{y}_{r+T+H-1}\}$. The origins are chosen so
that these forecast evaluation periods lie in the final holdout portion of the
sample. For the electricity dataset, we use $T=293$, $H=7$, and $R=67$
rolling origins. For the labor force dataset, we use $T=131$, $H=12$, and
$R=21$.

This rolling scheme separates the roles of the data. The rolling training
window is used for model fitting, covariance estimation, weight estimation, and
inner validation. The $H$ periods following the training window form the outer
holdout horizon at origin $r$; they are used only for the final testing-data
evaluation. Thus, the inner validation split is part of the
reconciliation-weight tuning step, whereas all reported RMSE values are
computed on the outer rolling-origin holdout periods.

\paragraph{Base forecasts and reconciliation at each origin.}
Within each training window, base forecasts are generated independently for every series using \code{auto.arima()} from the \proglang{R} package \pkg{forecast}~\citep{Rforecast}, with its default automatic model selection~\citep{hyndman2008automatic}. Each series is modelled with a weekly seasonal period of 7 for the daily electricity series and an annual period of 12 for the monthly labor force series, matching the pronounced weekly cycles and winter/summer peaks of the electricity data and the strong annual seasonality of the labor force data shown in Figures~\ref{elct} and~\ref{labour}. The procedure first selects the differencing orders, taking the number of seasonal differences from an STL-based seasonal-strength measure~\citep{wang2006characteristic} and the number of ordinary differences from repeated KPSS stationarity tests~\citep{kwiatkowski1992testing}; the seasonal differencing absorbs the recurring weekly and annual cycles, while the ordinary differencing removes the slow level shifts in electricity generation and the abrupt 2020 COVID-19 break in the labor force series, leaving an approximately stationary input. Conditional on the differencing, \code{auto.arima()} performs a stepwise search over the remaining orders, namely the numbers of non-seasonal and seasonal autoregressive and moving-average terms (together with the inclusion of a constant), restricted to stationary and invertible specifications: each candidate model is fitted by maximum likelihood and scored by its AICc, and the specification with the smallest AICc is selected. These fitted models generate one-step-ahead in-sample base forecasts $\{\hat{\mathbf{y}}_{r+1|r},\ldots,\hat{\mathbf{y}}_{r+T-1|r+T-2}\}$ and their errors $\{\hat{\mathbf{e}}_{r+1|r},\ldots,\hat{\mathbf{e}}_{r+T-1|r+T-2}\}$. These in-sample errors are used to estimate covariance matrices and, for penalized procedures, to tune $\lambda$.

The combination weights solve the penalized quadratic
program~\eqref{empirical-general-penalized}. For covariance estimation, we use
the proportional-covariance approximation reviewed in Section~\ref{framework}:
$\Sigma_h=\kappa_h\Sigma_1$ with $\kappa_h=1$. Thus the same estimated
covariance matrix $\hat{\Sigma}_1$ is used for all forecast horizons. Within
each rolling window, $\hat{\Sigma}_1$ is estimated from the one-step-ahead
in-sample base forecast errors. With $\hat{\Sigma}_h=\hat{\Sigma}_1$ for all
$h$, the induced matrix $\hat Q_h$ is common across horizons; we retain the
subscript $h$ only to match the notation in the theoretical formulation.

Among the four covariance specifications, the \textit{Factor} estimator follows
the shrinkage estimator of \citet{ledoit2003improved},
obtained from \code{covEstimation(..., control = list(type = "large"))} in the
\pkg{RiskPortfolios} package~\citep{ardia2017riskportfolios}. When
$\lambda=0$, weights are estimated with the full matrix $\hat Q_h$; by
Proposition~\ref{prop:equivalence}, the \textit{Separate} specification yields
identical weights, so a single set of unpenalized results is reported under
\textit{Joint} for all four covariance specifications. The unpenalized and
eRidge cases are obtained in closed form. The eLASSO case is reformulated into an equivalent unconstrained optimization problem in Section~\ref{sec:unconstrained} and solved using the
\pkg{OSQP} quadratic-program solver~\citep{osqp}. As benchmarks, \textit{CCC}
and \textit{LCC} are computed with \code{cslcc()} from the \pkg{FoReco}
package~\citep{FoReco}.

For penalized methods, the regularization parameter $\lambda$ is tuned
separately at each rolling origin and for each method. Because the eRidge and
eLASSO penalties operate on different scales, each uses its own logarithmically
spaced grid,
$\lambda_{\mathrm{eRidge}}\in\{0\}\cup[10^{-1},10^{5}]$ and
$\lambda_{\mathrm{eLASSO}}\in\{0\}\cup[10^{-2},10^{4}]$, each with 26 values at a
step of $0.25$ in $\log_{10}$.

For each rolling origin, the one-step-ahead in-sample forecast-error dates
within the training window are randomly split into an inner training set and an
inner validation set in an 80/20 proportion. This random split is used for
computational convenience in tuning the reconciliation penalty, treating the
in-sample one-step-ahead error dates as validation observations within the
rolling training window. It does not introduce data leakage into the reported
accuracy measures, because no observations from the outer testing horizon are
used in this tuning step and all final accuracy results are computed on the
out-of-sample rolling-origin holdout periods. For each candidate $\lambda$, the
inner training set is used to estimate $\hat{\Sigma}_1$ and the corresponding
combination weights. These weights are then applied to the one-step-ahead base
forecasts in the inner validation set, and $\lambda$ is selected by minimizing
the average per-series validation RMSE, computed separately for each series and
then averaged across all $n$ series so that every series receives equal weight
regardless of scale. The final weights are then re-estimated on the full
training window using the selected $\lambda$. These steps produce the
procedure-specific reconciliation inputs and final weights for origin $r$,
using only the rolling training window.

\paragraph{Testing data evaluation.}
After the training-window estimation and tuning steps are completed, the fitted
ARIMA models produce out-of-sample base forecasts for the outer holdout horizon,
\[
\{\hat{\mathbf{y}}_{r+T\mid r+T-1},\ldots,
\hat{\mathbf{y}}_{r+T+H-1\mid r+T-1}\}.
\]
These forecasts are then transformed into the final testing forecasts for each
procedure. Let $\mathcal{G}$ denote the set of forecasting or reconciliation
procedures evaluated in the empirical analysis. We write $g\in\mathcal{G}$ for
a procedure label, such as $g=\mathrm{Base}$,
$g=\mathrm{Bottom\text{-}up}$, $g=\mathrm{LCC}$, or one of the proposed
combination-based reconciliation variants. For each origin $r$ and horizon
$h=1,\ldots,H$, let $\tilde{\mathbf{y}}_{r,h}^{[g]}$ denote the final testing
forecast vector for all series produced by procedure $g$. This compact notation
corresponds to the longer time-indexed forecast
$\tilde{\mathbf{y}}_{r+T+h-1\mid r+T-1}^{[g]}$, whose target time is
$r+T+h-1$ and whose forecast origin is $r+T-1$. The square-bracket superscript
$[g]$ is a procedure label.

For the unreconciled \textit{Base} procedure,
\[
\tilde{\mathbf{y}}_{r,h}^{[\mathrm{Base}]}
:=
\hat{\mathbf{y}}_{r+T+h-1\mid r+T-1}.
\]
Here the tilde denotes the final output used for evaluation; for
$g=\mathrm{Base}$, it means no reconciliation is applied. For a proposed
combination-based reconciliation procedure $g$, the evaluated forecast takes the form
\[
\tilde{\mathbf{y}}_{r,h}^{[g]}
=
S\Phi(\hat{\mathbf{w}}_{r}^{[g]})C
\hat{\mathbf{y}}_{r+T+h-1\mid r+T-1},
\]
where $\hat{\mathbf{w}}_{r}^{[g]}$ is the final weight vector estimated from
the rolling training window at origin $r$. Other benchmark procedures are
treated analogously: regardless of their internal construction, their final
forecasts are represented generically by $\tilde{\mathbf{y}}_{r,h}^{[g]}$ for
evaluation.

We assess forecast accuracy using RMSE at each hierarchy level. For rolling
origin $r=1,\ldots,R$, hierarchy level $\ell=1,\ldots,L$, and procedure
$g\in\mathcal{G}$, let
$E_{r,\ell}^{[g]}\in\mathbb{R}^{n_\ell\times H}$ collect the forecast errors
for all $n_\ell$ series at level $\ell$ and all $H$ forecast horizons. Writing
$\mathbf{y}_{t}^{\ell}$ for the subvector of realized values at level $\ell$
and $\tilde{\mathbf{y}}_{r,h,\ell}^{[g]}$ for the corresponding level-$\ell$
subvector of $\tilde{\mathbf{y}}_{r,h}^{[g]}$, define
\[
E_{r,\ell}^{[g]} =
\bigl[
\mathbf{y}_{r+T}^{\ell}
- \tilde{\mathbf{y}}_{r,1,\ell}^{[g]},\;
\ldots,\;
\mathbf{y}_{r+T+H-1}^{\ell}
- \tilde{\mathbf{y}}_{r,H,\ell}^{[g]}
\bigr].
\]
The level-$\ell$ RMSE for procedure $g$ averages the origin-level RMSEs:
\[
\mathrm{RMSE}^{[g]}(\ell) = \frac{1}{R}\sum_{r=1}^{R}
\frac{\|E_{r,\ell}^{[g]}\|_F}{\sqrt{n_\ell H}}.
\]
The corresponding RMSE skill score is the percentage improvement over the
unreconciled \textit{Base} forecasts:
\[
\mathrm{Skill}^{[g]}(\ell)
=
100\times
\left(
1-
\frac{\mathrm{RMSE}^{[g]}(\ell)}
{\mathrm{RMSE}^{[\mathrm{Base}]}(\ell)}
\right).
\]
In the electricity hierarchy, levels correspond to the total, source type,
energy category, and detailed energy-source components. In the labor force
grouped structure, they correspond to the top level, duration groups, STT
groups, and duration~$\times$~STT bottom-level series. The RMSE at level $\ell$
therefore aggregates errors across comparable series within that level, across
forecast horizons, and across rolling origins.
Figure~\ref{fig:rolling-scheme} summarizes both the outer rolling-origin
testing design and the inner validation loop for $\lambda$.

\definecolor{cTrain}{HTML}{595959}\definecolor{cTest}{HTML}{D95F0E}
\definecolor{cData}{HTML}{EAEFF4}\definecolor{cProc}{HTML}{F0F0F0}
\definecolor{cEst}{HTML}{E6F4E6}\definecolor{cFc}{HTML}{FDE0DC}
\providecommand{\splitbar}{\tikz[baseline=-0.5ex]{  \foreach \k in {0,...,7}{\fill[cTrain!22](\k*3pt,0) rectangle (\k*3pt+2.4pt,6.5pt);\draw[cTrain!70,line width=0.15pt](\k*3pt,0) rectangle (\k*3pt+2.4pt,6.5pt);}  \foreach \k in {8,9}{\fill[cTest!40](\k*3pt,0) rectangle (\k*3pt+2.4pt,6.5pt);\draw[cTest!80,line width=0.15pt](\k*3pt,0) rectangle (\k*3pt+2.4pt,6.5pt);}}}

\begin{figure}[htbp]\centering
\begin{adjustbox}{width=\textwidth, center}
\begin{tikzpicture}[
  io/.style   ={draw=cTrain!70, rounded corners=2pt, fill=cProc, align=center, inner sep=4pt, font=\small},
  proc/.style ={draw=cTrain!70, rounded corners=2pt, fill=cProc, align=center, inner sep=4pt, font=\small},
  est/.style  ={draw=cTrain!70, rounded corners=2pt, fill=cProc, align=center, inner sep=4pt, font=\small},
  pen/.style  ={draw=cTrain!70, rounded corners=2pt, fill=cProc, align=center, inner sep=4pt, font=\small},
  fc/.style   ={draw=cTrain!70, rounded corners=2pt, fill=cProc, align=center, inner sep=4pt, font=\small},
  ar/.style   ={-{Latex[length=2.6mm]}, semithick, gray!65},
  lbl/.style  ={font=\small\itshape, text=gray!55!black}]

\begin{scope}
  \def\Lx{13.0}\def\splitx{10.4}\def\T{10.4}\def\fc{0.18}\def\bh{0.46}
  \fill[cTrain!10] (0,0) rectangle (\splitx,\bh);
  \fill[cTest!14] (\splitx,0) rectangle (\Lx,\bh);
  \draw[gray!55] (0,0) rectangle (\Lx,\bh);
  \node[font=\normalsize\itshape, text=gray!55!black, anchor=west] at (0.05,\bh+0.40) {full sample of $n$ time series ($t \rightarrow$)};
  \node[font=\small\bfseries, text=cTrain] at (\splitx/2,\bh/2) {estimation period ($80\%$)};
  \node[font=\small\bfseries, text=cTest!85!black] at ({(\splitx+\Lx)/2},\bh/2) {test set ($20\%$)};
  \draw[dashed, cTest!80, thick] (\splitx,-3.15) -- (\splitx,\bh+0.15);
  \newcommand{\rollwin}[3]{    \fill[cTrain!20] (#2-\T,#1) rectangle (#2,#1+\bh);
    \draw[cTrain!70] (#2-\T,#1) rectangle (#2,#1+\bh);
    \foreach \k in {0,...,6}{\fill[cTest!30] ($(#2+\k*\fc,#1)$) rectangle ($(#2+\k*\fc+\fc,#1+\bh)$);
       \draw[cTest!70,line width=0.2pt] ($(#2+\k*\fc,#1)$) rectangle ($(#2+\k*\fc+\fc,#1+\bh)$);}
    \draw[dotted,gray!55,line width=0.3pt] (#2+7*\fc+0.05,#1+\bh/2) -- (13.35,#1+\bh/2);
    \node[font=\small,anchor=west,text=gray!55!black] at (13.48,#1+\bh/2) {#3};}
  \rollwin{-0.80}{10.40}{$r{=}1$}
  \rollwin{-1.58}{10.78}{$r{=}2$}
  \node[font=\large] at (11.55,-2.18) {$\vdots$};
  \rollwin{-2.92}{11.88}{$r{=}R$}
  \node[font=\small, text=cTrain!65!black] at (4.8,-0.57) {fixed training window $T$};
  \draw[ar] (0.20,-1.02) -- (0.20,-2.66);
  \node[lbl, rotate=90, anchor=south] at (0.0,-1.84) {roll forward by 1 step};
  \begin{scope}[shift={(3.05,-3.55)}]
    \fill[cTrain!20] (0,0) rectangle (0.36,0.26); \draw[cTrain!70](0,0) rectangle (0.36,0.26);
    \node[font=\small,anchor=west] at (0.44,0.13) {training sample};
    \fill[cTest!30] (3.0,0) rectangle (3.36,0.26); \draw[cTest!70](3.0,0) rectangle (3.36,0.26);
    \node[font=\small,anchor=west] at (3.44,0.13) {test set};
  \end{scope}
  \node[font=\bfseries] at (-0.5,\bh+0.30) {(a)};
\end{scope}

\begin{scope}[shift={(0,-6.4)}]
  \node[font=\bfseries] at (-0.5,2.15) {(b)};
  \node[font=\normalsize\itshape, text=gray!55!black, anchor=west] at (0.05,2.2) {at each rolling origin $r$:};
  \node[io, text width=2.8cm] (win) at (0,-0.95) {training sample\\$\{\mathbf y_r,\dots,\mathbf y_{r+T-1}\}$};
  \node[proc, text width=1.9cm, right=7mm of win] (arima) {fit ARIMA to each series};
  \coordinate (jx) at ($(arima.east)+(0.45,0)$);
  \node[io, text width=2.45cm, right=10mm of arima, yshift=20mm] (err)
     {in-sample one-step errors\\$\{\hat{\mathbf e}_{r+1|r},\dots,$\\$\hat{\mathbf e}_{r+T-1|r+T-2}\}$};
  \node[fc, text width=2.6cm, right=10mm of arima, yshift=-20mm] (base)
     {out-of-sample base forecasts\\$\{\hat{\mathbf y}_{r+T|r+T-1},$\\$\ldots,$\\$\hat{\mathbf y}_{r+T+H-1|r+T-1}\}$};
  \node[est, text width=2.7cm, right=13mm of err] (unpen)
     {\textbf{unpenalized} $(\lambda{=}0)$\\solve \eqref{empirical-general-penalized}, full $\hat Q_h$;\\Joint$\equiv$Sep (Prop~\ref{prop:equivalence})\\$\Rightarrow\hat{\mathbf w}$};
  \node[pen, text width=2.7cm, below=6mm of unpen] (penb)
     {\textbf{penalized} $(\lambda{>}0)$\\solve \eqref{empirical-general-penalized};\\$\lambda$ via inner validation\\$\Rightarrow\hat{\mathbf w}$};
  \node[proc, text width=2.9cm, right=8mm of unpen, yshift=-20mm] (recon)
     {reconcile\\
      $\tilde{\mathbf y}_{r+T+h-1\mid r+T-1}$\\
      $=S\,\Phi(\hat{\mathbf w})\,C\,$\\
      $\hat{\mathbf y}_{r+T+h-1\mid r+T-1},$\\
      $h=1,\dots,H$};
  \node[io, text width=1.5cm, below=8mm of recon] (eval) {RMSE per level};
  \draw[ar] (win) -- (arima);
  \draw[semithick,gray!65] (arima.east) -- (jx);
  \draw[ar] (jx) |- (err.west);
  \draw[ar] (jx) |- (base.west);
  \coordinate (fk) at ($(err.east)+(0.65,0)$);
  \draw[semithick,gray!65] (err.east) -- (fk);
  \draw[ar] (fk) |- (unpen.west);
  \draw[ar] (fk) |- (penb.west);
  \draw[ar] (unpen.east) -| (recon.north);
  \draw[ar] (penb.east) -- ++(0.45,0) |- (recon.west);
  \draw[ar] (base.south) -- ++(0,-0.5) -| ([xshift=6mm]recon.east) -- (recon.east);
  \draw[ar] (recon) -- (eval);
    \draw[rounded corners=3pt, draw=cTrain!65, fill=cProc] (0.8,-8.8) rectangle (12.95,-5.35);
  \node[anchor=west, font=\small\bfseries, text=cTrain!75!black] at (1.5,-5.62)
       {inner validation for $\lambda$ (per origin, within the training data)};
  \node[anchor=north, font=\footnotesize, align=center] at (2.8,-6.02)
       {shuffle one-step errors,\\split $80/20$};
  \node at (2.8,-7.8) {\splitbar};
  \draw[ar] (4.1,-7.35) -- (4.85,-7.35);
  \begin{scope}[shift={(5.6,-8.05)}]
    \draw[->,gray!70,line width=0.4pt] (0,0) -- (1.95,0);
    \draw[->,gray!70,line width=0.4pt] (0,0) -- (0,1.3);
    \draw[cTrain,line width=0.9pt,smooth] plot coordinates
      {(0.15,1.12)(0.50,0.62)(0.85,0.35)(1.08,0.26)(1.38,0.42)(1.70,0.74)(1.92,1.05)};
    \fill[cTest] (1.08,0.26) circle (1.6pt);
    \draw[dashed,cTest!75,line width=0.4pt] (1.08,0.26) -- (1.08,0);
    \node[font=\footnotesize,text=cTest!80!black,anchor=north] at (1.08,-0.02) {$\lambda^\star$};
    \node[font=\footnotesize,rotate=90,anchor=south,text=gray!55!black] at (-0.14,0.65) {RMSE};
  \end{scope}
  \node[align=left, font=\footnotesize, anchor=north west, text width=4.7cm] at (7.9,-5.75)
       {for each $\lambda$:\\
        $\;$80\% $\to\hat\Sigma_1,\hat{\mathbf w}(\lambda)$\\
        $\;$20\% $\to$ one-step recon.\ RMSE\\[1pt]
        $\lambda^\star{=}\arg\min_\lambda \mathrm{RMSE}(\lambda)$,\\
        refit $\hat{\mathbf w}$ on full sample};
  \draw[ar, dashed, cTrain!60] (penb.south) -- (penb.south|-{(0,-5.35)});
\end{scope}

\end{tikzpicture}
\end{adjustbox}
\caption{Rolling-origin forecasting and validation scheme:
\textbf{(a)} the rolling design on the timeline; \textbf{(b)} the estimation
pipeline applied at each origin.}
\label{fig:rolling-scheme}
\end{figure}
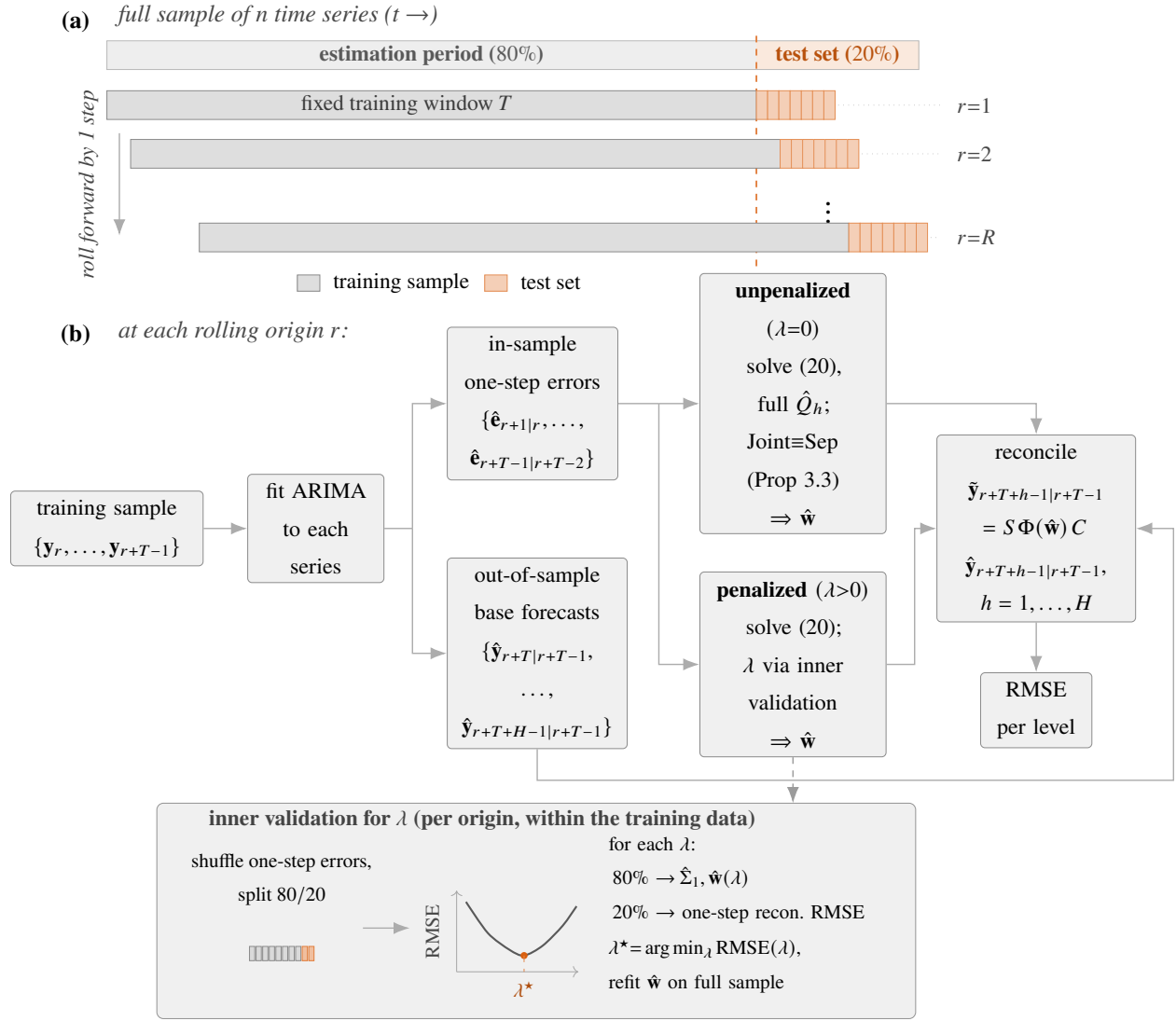

\subsection{Candidate-Weight Profiles}
\label{app:weight-interpretation}

This subsection visualizes the combination weight profiles to characterize
the contributions of the direct, ancestor, and collateral candidate forecasts,
defined in Section~\ref{constructing-algorithm}, to the reconciled
bottom-level forecasts. 
The unpenalized
\textit{Factor+Joint} configuration serves as the reference; by
Proposition~\ref{prop:equivalence}, it also represents the unpenalized
\textit{Factor+Separate} solution, so we label it as \textit{Factor+Joint(Separate)} in this part. We compare this reference with the
penalized 
\textit{Factor+Separate+eLASSO} configuration. Results for the
penalized 
\textit{Factor+Separate+eRidge} configuration are similar and hence omitted.   
The same diagnostic can be also 
constructed for the remaining covariance estimators and optimization
configurations.

For each method configuration $g$, let $\widehat w_{i,j,r}^{[g]}$ denote the estimated
combination weight on candidate forecast $j$ for bottom-level series $i$ at rolling
origin $r$. We summarize these weights by their rolling-origin mean:
\begin{equation}
\overline w_{i,j}^{[g]}
=
\frac{1}{R}\sum_{r=1}^{R}\widehat w_{i,j,r}^{[g]},
\qquad j=0,\ldots,n_a,
\label{eq:ec-mean-candidate-weight}
\end{equation}
where $R=67$ for the electricity application and $R=21$ for the labor force application. Thus,
$\overline w_{i,j}^{[g]}$ is the average signed weight assigned to candidate $j$
for series $i$ across rolling origins. Candidate $j=0$ is the direct
candidate, whereas $j=1,\ldots,n_a$ index the indirect candidates, including both ancestor-based candidates and collatorial-based candidates.

In Figures~\ref{fig:ec-weight-energy} and
\ref{fig:ec-weight-labour}, each row corresponds to a bottom-level series and
each point corresponds to one candidate's rolling-origin mean signed weight, for the electricity dataset and labor force dataset respectively. Colors
identify direct, ancestor, and collateral candidates. The solid and dashed
reference lines indicate zero and the equal-weight benchmark
$1/(n_a+1)$, respectively. From left to right, the panels show the
unpenalized \textit{Factor+Joint(Separate)} configuration and 
\textit{Factor+Separate+eLASSO}. The electricity panels use
method-specific horizontal scales, so comparisons across these panels should
be based on the axis values rather than on apparent panel widths. The labor force
panels share a common horizontal scale.

\begin{figure}[!htbp]
  \centering
  \includegraphics[width=\textwidth,height=0.72\textheight,keepaspectratio]
    {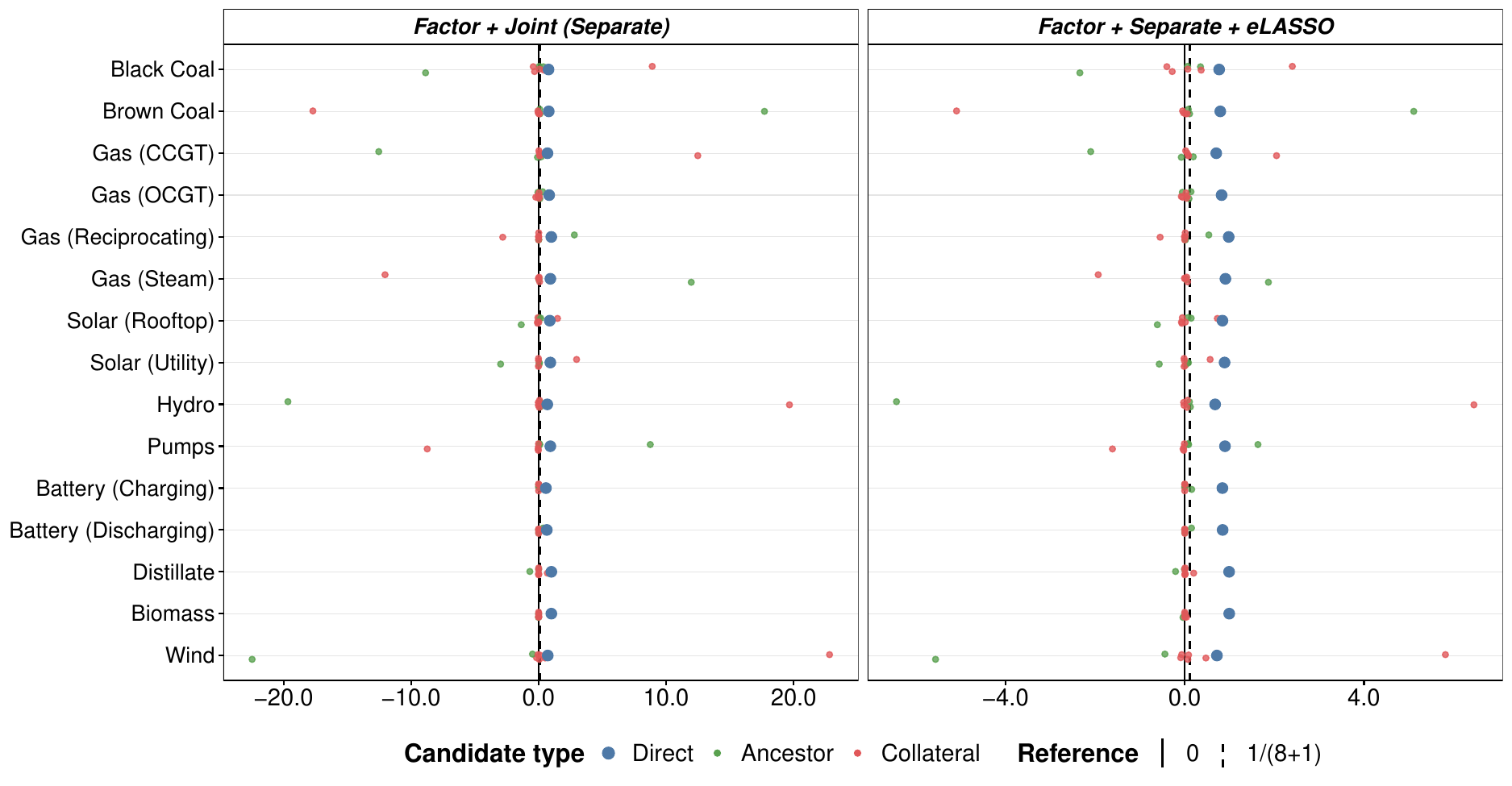}
  \caption{Rolling-origin mean signed candidate weights for the 15 bottom-level electricity series.}
  \label{fig:ec-weight-energy}
  \end{figure}

\paragraph{Electricity application.}
Figure~\ref{fig:ec-weight-energy} shows substantial heterogeneity in the
candidate-weight profiles for the electricity application. In the unpenalized \textit{Factor+Joint(Separate)}
configuration, the direct candidate receives a positive mean weight above the
equal-weight benchmark for all bottom-level series, but several indirect
candidates have large positive or negative mean weights. Thus, reconciliation
does not simply rescale the direct forecasts; it combines aggregate-derived
information in markedly different ways across energy sources.
\textit{Factor+Separate+eLASSO} reduces the magnitude of the most extreme
weights for the indirect candidates. However, the weight profiles for the  penalized configuration remain clearly
series-specific: direct weights remain prominent, and both ancestor and
collateral candidates retain nonzero contributions for selected series.
Hence, egalitarian regularization moderates the most extreme allocations
without imposing a common weight profile across the electricity hierarchy.

\begin{figure}[!htbp]
  \centering
  \includegraphics[width=\textwidth,height=0.80\textheight,keepaspectratio]
    {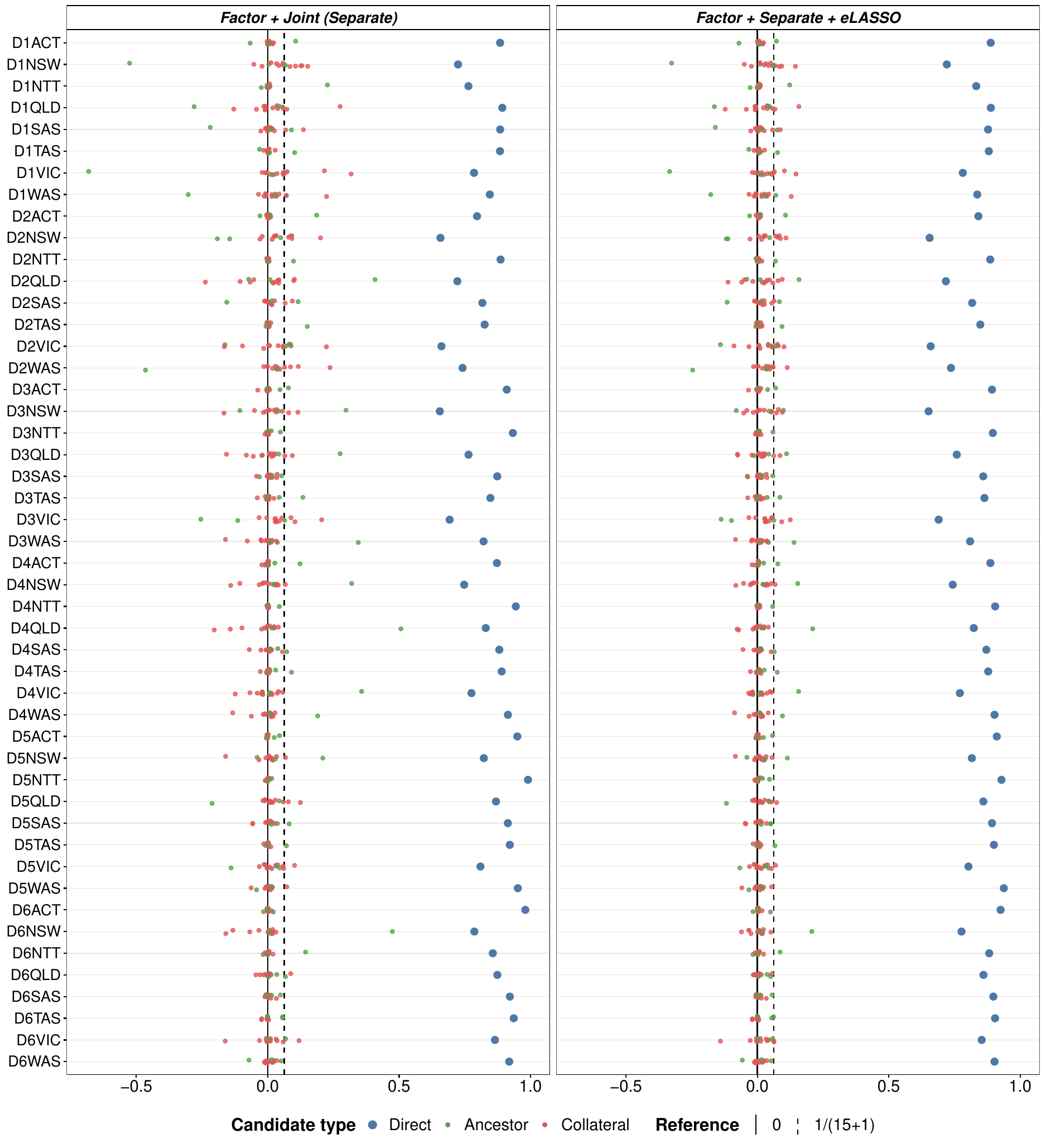}
  \caption{Rolling-origin mean signed candidate weights for the 48 bottom-level duration--STT labor force series.}
  \label{fig:ec-weight-labour}
  \end{figure}

\paragraph{Labor force application.}
Figure~\ref{fig:ec-weight-labour} displays a more concentrated pattern. The
labor force data are cross-classified by job-search duration and state or territory
(STT), so each bottom-level series has both ancestor and collateral candidates
in addition to its direct candidate. In the unpenalized \textit{Factor+Joint(Separate)} configuration, the
direct candidate receives the largest mean weight for most duration--STT
cells, whereas the indirect weights are generally closer to zero. Nevertheless,
selected ancestor and collateral candidates retain non-negligible weights, and
their allocations vary across cells. Because the weight profiles for the labor force unpenalized  \textit{Factor+Joint(Separate)} configuration are already relatively concentrated,
the visual effects of 
\textit{Factor+Separate+eLASSO} are more modest than
in the electricity application. The weight profiles of this penalized configuration still exhibit heterogeneous indirect
allocations across duration--STT cells, indicating that grouped information is
used selectively rather than through a common allocation.

\paragraph{Summary.} Overall, the two applications illustrate the same qualitative mechanism:
the direct forecast remains important, while reconciliation selectively draws
on indirect candidate forecasts. They differ in the extent of this adjustment:
the electricity application exhibits large and opposing indirect weights that are materially
tempered by regularization, whereas the labor force application has more concentrated weights
and correspondingly smaller regularization effects.

\subsection{Computational Comparison of \textit{Joint} and \textit{Separate} Implementations}
\label{app:computation-time}

This subsection quantifies the  computational benefit of the
\textit{Separate} implementation relative to \textit{Joint}.
The role of the \textit{Separate} implementation differs between the
unpenalized and penalized settings. For the unpenalized problem,
Proposition~\ref{prop:equivalence} shows that \textit{Joint} and
\textit{Separate} produce the same estimated weights. Their
difference is purely computational: \textit{Joint} solves the full
problem, whereas \textit{Separate} exploits its decomposition into
series-specific subproblems. With weight penalization, the two implementations generally produce different
weights. Nevertheless, the empirical results in
Section~\ref{sec:empirical} show that their forecasting performance is
closely comparable, with \textit{Separate} often achieving even better performance. This part further illustrates the computational benefit of \textit{Separate}, again using the \textit{Factor}-based covariance estimator as an example.

After imposing the sum-to-one constraints, the \textit{Joint} problem has
120 free parameters for the electricity application and 720 for the labor force application. By contrast,
\textit{Separate} decomposes these problems into 15 independent
eight-dimensional subproblems for the electricity application and 48 independent
fifteen-dimensional subproblems for the labor force application. Table~\ref{tab:ec-computation-time}
quantifies the resulting serial computation-time differences.
Specifically, this table reports mean serial elapsed seconds per
rolling origin for the \textit{Factor}-based implementations, with standard
deviations in parentheses, over 67 electricity rolling origins and 21 labor force rolling origins.
Timings include covariance estimation, validation-based tuning when
applicable, and final weight estimation, but exclude base-model estimation,
base-forecast generation, file input/output, and out-of-sample evaluation.

\begin{table}[!htbp]
\centering
\caption{Mean serial computation time per rolling origin for the \textit{Factor}-based implementations, with standard deviation in parentheses.}
\label{tab:ec-computation-time}
\small
\begin{tabular}{llccc}
\toprule
Dataset & Weight penalty &
\textit{Joint} (s) &
\textit{Separate} (s) &
\makecell{Time ratio\\(\textit{Sep.}/\textit{Joint}, \%)} \\
\midrule
\multirow{3}{*}{Electricity}
 & \textit{None}   & 0.017 (0.001) & 0.012 (0.001) & 69.0 \\
 & \textit{eRidge} & 0.155 (0.013) & 0.044 (0.001) & 28.2 \\
 & \textit{eLASSO} & 1.881 (0.884) & 1.838 (0.213) & 97.7 \\
\midrule
\multirow{3}{*}{Labor force}
 & \textit{None}   & 0.930 (0.098) & 0.142 (0.009) & 15.2 \\
 & \textit{eRidge} & 22.431 (1.959) & 0.315 (0.028) & 1.4 \\
 & \textit{eLASSO} & 82.265 (9.506) & 6.196 (0.683) & 7.5 \\
\bottomrule
\end{tabular}
\end{table}

To ensure comparability, all timing experiments were conducted on the same
MacBook Pro (16-inch, 2019) running macOS Big Sur 11.2.3, with a 2.3 GHz
8-core Intel Core i9 processor and 32 GB of memory. The code was run in R
through RStudio under a strictly serial configuration with one BLAS/OpenMP
thread. Elapsed time was measured once for each method--origin pair using
\texttt{proc.time()[["elapsed"]]}.

Table~\ref{tab:ec-computation-time} shows that \textit{Separate} is faster
than \textit{Joint} in every reported configuration. The gains are especially
pronounced for the labor force dataset, where \textit{Separate} needs only 1.4--15.2\% of
the corresponding \textit{Joint} time. For the electricity application, it needs 28.2--97.7\%,
so the gain depends on the penalty. The standard deviations are also generally
smaller under \textit{Separate}, suggesting more stable computation times. This
is clear for electricity with \textit{eLASSO}: although the two
implementations have similar mean runtimes, \textit{Separate} has a markedly
smaller standard deviation.
Moreover, the \textit{eLASSO} configurations are substantially slower than their
\textit{eRidge} counterparts under both implementations because \textit{eLASSO}
requires iterative nonsmooth optimization during validation and estimation.
Even so, for the larger labor force application, \textit{Separate} retains substantial
computational advantages under all penalties. Together with its comparable and
often improved forecasting accuracy, this reduction in computational cost makes
\textit{Separate} particularly attractive for larger reconciliation problems.

\end{document}